%% file: main.tex
\documentclass{lmcs}
\pdfoutput=1

\usepackage{lastpage}
\lmcsdoi{17}{2}{24}
\lmcsheading{}{\pageref{LastPage}}{}{}%
{Feb.~19,~2020}{Jun.~02,~2021}{}

\usepackage{verbatim}

\RequirePackage{doi}

\usepackage{cite} % http://www.ctan.org/pkg/cite
\usepackage{graphicx}
\usepackage{amssymb}
\usepackage{amsmath}
\usepackage{amsthm}
\usepackage{wasysym}
\usepackage{newtxmath} % only math part of the old txfonts
\usepackage{proof}

\usepackage[utf8]{inputenc}

\usepackage[shortlabels]{enumitem}

\usepackage[all]{xy}

\usepackage{spi}
\usepackage{prooftree}

\renewcommand{\vDash}{\models}
\renewcommand{\nvDash}{\not\models}

\newcommand{\OM}{\ensuremath{\mathcal{O\!M}}}

\newcommand{\ttt}{\mathtt{t\hspace*{-.25em}t}}

\newcommand{\boxm}[1]{\mathopen{\big[ #1 \big]}} %{[\hspace{.2ex}#1\hspace{.2ex}]}
\newcommand{\diam}[1]{\mathopen{\big\langle#1\big\rangle}}
\usepackage{todonotes}
\newcommand{\showcomment}{Uncomment to remove comments}
\newcommand{\mycomment}[3]{
\ifdefined\showcomment%
\todo[linecolor=#1,backgroundcolor=#1!25,bordercolor=#1]{[#2] #3}
\fi
}

\newcommand{\SM}[1]{\mycomment{blue}{SM}{#1}}

\title[Discovering ePassport Vulnerabilities using Bisimilarity]{Discovering ePassport Vulnerabilities \texorpdfstring{\\}{} using Bisimilarity}

\author[R.~Horne]{Ross Horne\rsuper{a}}
\address{\lsuper{a}Department of Computer Science, University of Luxembourg, Esch-sur-Alzette, Luxembourg}
\email{ross.horne@uni.lu}

\author[S.~Mauw]{Sjouke Mauw\rsuper{{a,b}}}
\address{\lsuper{b}SnT, University of Luxembourg, Esch-sur-Alzette, Luxembourg}
\email{sjouke.mauw@uni.lu}

% mandatory: Please provide 1-5 keywords

\begin{document}

% make the title area
\maketitle

%\SM{I find the macro for comments sub-optimal. They are displayed in the an awkward way in the margin, making it impossible to read them.}
\begin{abstract}
We uncover privacy vulnerabilities in the ICAO 9303 standard implemented by ePassports worldwide.
These vulnerabilities, confirmed by ICAO, enable an ePassport holder who recently passed through a checkpoint to be reidentified without opening their ePassport.
This paper explains how bisimilarity was used to discover these vulnerabilities,
which exploit the BAC protocol -- the original ICAO 9303 standard ePassport authentication protocol -- and remains valid for the PACE protocol, which improves on the security of BAC in the latest ICAO 9303 standards.
In order to tackle such bisimilarity problems, we develop here a chain of methods for the applied $\pi$-calculus
including a symbolic under-approximation of bisimilarity, called open bisimilarity, and a modal logic, called classical $\FM$, for describing and certifying attacks.
Evidence is provided to argue for a new scheme for specifying such unlinkability problems that more accurately reflects the capabilities of an attacker.
\end{abstract}

% Introduction
\input{intro} % sec:intro
\input{protocols} % sec:protocols
\input{symbolic} % sec:symbolic
\input{improved}

\input{pace}

\input{conclusion}

%\input{notes} % sec:symbolic

% references section
\bibliographystyle{halpha}
\bibliography{biblio}

\appendix
\input{proof}

\end{document}

%% file: intro.tex
\section{Introduction}

Most of us have the option to pass through automatic passport clearance at an airport.
Some of us also have electronic national cards that may be used for government services.
All of these  machine readable documents employ a protocol to authenticate with a reader, establishing that you really hold a valid machine readable document.
In order for ePassports to be read internationally, your passport almost certainly implements a standardised protocol for machine readable travel documents, defined by the International Civil Aviation Authority (ICAO) -- the UN agency responsible for international aviation standards.

Considerable work has been put into ensuring ePassports satisfy security properties, preventing your ePassport from being read by an unauthorised $3^{rd}$-party in the vicinity.
However, even if such security properties are satisfied, there may still be ways of exploiting a protocol to mount more subtle attacks on your privacy.
Notably, a requirement of ePassports, is that an unauthorised $3^{rd}$-party should not be able to use an ePassport to track the document holder.
This privacy property is called \emph{unlinkability}, and is an official requirement of the ICAO 9303 standard for machine readable travel documents~\cite{MRTD}.

It has been debated over the past decade
whether or not the ICAO 9303 standard satisfies unlinkability.
Vulnerabilities have been discovered by exploiting implementation specific features, such as the different error messages in the French ePassport, or the differences in response time of the ePassports of different nationalities~\cite{originalBACAttack,Avoine16}.
For example, an error message in the French ePassport indicates whether a message authentication code (MAC) test passed, despite authentication failing;
hence if a message with the same MAC key is replayed from a previous session with the same ePassport, then we can detect whether or not the same ePassport is present in the current session.
Notice this requires no access to the personal data stored inside the ePassport.
%in the  an authentication which indicates that the same passport was present.

Now put implementation-specific side channels aside
and consider whether unlinkability holds at the level of the specification of the Basic Access Control (BAC) protocol, as defined in the ICAO 9303 standard.
A claim was reported in CSF'10~\cite{Arapinis2010} that unlinkability does hold for ePassports that conservatively implement the BAC protocol.
In particular, the claim concerns implementations where the same plaintext message should be provided for all types of error, as is the case for the UK ePassport for example.
That claim was backed up by a formal model of a property that should hold if unlinkability of BAC holds, which is expressed as a bisimilarity problem in the applied $\pi$-calculus~\cite{Abadi2017}.
The problem is that this original claim was discovered to be false, as reported in ESORICS'19~\cite{ESORICS}.
This indicates a failure of the ICAO 9303 BAC protocol to meet its own requirements.

However, this is not the end of the story behind this privacy vulnerability, which has several twists.
A twist is that the original claim which we discovered to be flawed was based on a proof in ProVerif, that went through due to a bug, now resolved in Proverif.
When the bug was corrected the old proof didn't go through, but no proof or counter-example was reported until ESORICS'19~\cite{ESORICS}.
This indicates the need to improve methods and tools for supporting bisimilarity checking in the applied $\pi$-calculus,
so that false privacy claims about widely deployed protocols do not go undetected for a decade.

There are further twists in this story. In the effort to improve tools and methods for resolving such unlinkability problems,
several alternative models of the unlinkability of the BAC protocol have been proposed over the years.
Some of these models can be used to prove that there is no attack~\cite{Hirschi2019}, as first communicated in S\&P'16~\cite{BACUnboundedUnlinkability}.
The key difference between the original CSF'10 model, for which we discovered an attack in ESORICS'19, and the S\&P'16 model, where they prove there is no attack, is the use of trace equivalence rather than bisimilarity in the latter.
This is interesting,
since there are few, if any, protocols and properties where the use of bisimilarity rather than trace equivalence is essential for discovering vulnerabilities, at least for a widely-deployed protocol.
Furthermore, it provokes the question of which equivalence provides the appropriate attacker model: are vulnerabilities discovered using bisimilarity, but undetectable using trace equivalence exploitable; and, if so, are they perhaps less dangerous in some sense than attacks which can be described as a trace?

For the BAC protocol we have answered the question of exploitability in the positive.
Using a modified reader and ePassport we have demonstrated how the distinguishing game exposed by the failure of bisimilarity can be exploited to reidentify an ePassport.
That particular implementation of the vulnerability discovered using bisimilarity (there are infinitely many mutations of this attack)
was reported through a responsible disclosure process to ICAO in June 2019.

ICAO issued a public response made
%via Twitter\footnote{https://twitter.com/delanomagazine/status/1177129598359953409} on 26 September,
available via numerous press reports~\cite{Delano,PaperJam1,PaperJam2}.
In their response, ICAO make the following statement.
\begin{quote}
``It's also important to consider here that the described issue, which could be exploited for example at border controls or at other inspection system areas, would only allow adversaries to be able to know that somebody recently passed through a passport check-- and even without opening their ePassport. The personal data stored in the contactless chip, however, would not be disclosed.''
\end{quote}
Understandably, ICAO aim to contain this issue, and we have no interest in creating a scandal, only in ensuring the appropriate agencies receive accurate information.
However, please note that the above statement confirms that ICAO agree the vulnerability is real and would, we quote again for emphasis, ``allow adversaries to be able to know that somebody recently passed through a passport check-- and even without opening their ePassport.'' This exactly matches our own claims about the capabilities offered to an attacker exploiting the vulnerability discovered.
The word ``recent'' in the above context, means that the ePassport can only be tracked for as long as the attacker can keep open a session with the reader that the ePassport holder recently passed through; which, in practice, can only be a short period of time.
This contrasts to more serious implementation-specific vulnerabilities, which can be exploited to track the ePassport holder indefinitely.
Being able to reidentify someone within a time-limited period is nonetheless a violation of unlinkability.

It is also important that we clarify the following public response from ICAO, also, understandably, aiming to contain any fallout from a vulnerability affecting citizens using their ePassport standard worldwide.
\begin{quote}
``ICAO and experts have thoroughly reviewed this research and their initial analysis is that it is not linked to Doc 9303 specifications in their current version.
This is especially the case given that the newest Doc 9303 specifications incorporate the PACE protocol, which is considered a more secure alternative to the BAC protocol.''
%Additionally, the concerns being expressed are seen as pertaining more to the verification systems used to read ePassport data, rather than the documents themselves or their self-contained security measures.
\end{quote}
What the above means is that the vulnerability reported at ESORICS'19 was for the BAC protocol.
The BAC protocol is the authentication mechanism used to ensure the ePassport and reader are
really talking to each other before exchanging any personal data stored on the ePassport.
It has been used since the first generation of ePassports, issued since 2004, and, at the time of writing, is still supported by ePassports.
BAC has known security limitations, for example the keys are generated using information such as the passport number and expiry dates, which have low entropy~\cite{Bender2009}.
%\SM{I'm not sure if I understand the reasoning that a \emph{significant} amount of entropy leads to attacks. I would expect low-entropy input to lead to attacks.}
Thus there are attacks that can enable a user to compromise the \textit{secrecy} of the personal data on an ePassport protected by BAC\@.
For this reason, ICAO have developed the Password Authenticated Connection Establishment (PACE)
protocol, addressing such vulnerabilities that can lead to data breaches. Note a data breach would also immediately compromise unlinkability, since the attacker would have direct access to the identity of the ePassport holder.

Thus, PACE is an improvement over BAC from the perspective of secrecy; however, secrecy and privacy are not the same thing.
Indeed, we report here that, PACE is also vulnerable to attacks on unlinkability by adopting a similar strategy to the attacks on BAC reported in ESORICS'19.
As with BAC, we can formally account for this vulnerability by showing that PACE does
\emph{not} satisfy unlinkability, formalised in terms of bisimilarity.
Since ePassports implementing BAC or PACE are issued by over 150 countries\footnote{Gemalto on ePassport trends: \url{https://www.gemalto.com/govt/travel/electronic-passport-trends}}, the impact for society of this vulnerability is current and global.

This paper is an extended version of a paper presented at ESORICS'19.
In the conference version of this paper, we explained the privacy vulnerability discovered for the BAC protocol and explained how the attack can be implemented in a real-world setting.
This paper complements the conference version by focusing on our methodology for analysing such unlinkability problems rather than the implementation concerns.
%We explains that the same type of vulnerability is also present in PACE.
We explain the formal methods we developed and employed to quickly discover attacks on the unlinkability of the BAC protocol, and, going beyond the ESORICS'19 paper, also the PACE protocol.
% rather than the practicalities of implementing devices that can exploit the vulnerability.
%To analyse Protocols BAC, PACE, and related authentication protocols, fail to satisfy unlinkability, specified as a bisimilarity problem.
To approach the bisimilarity problem behind unlinkability, we employ a game between a prover aiming to show unlinkability holds and a disprover aiming to show there is an attack on unlinkability.
The prover uses symbolic techniques to try to construct a bisimulation for an under-approximation of bisimilarity (open bisimilarity);
while the disprover aims to verify whether attacks discovered are genuine distinguishing strategies invalidating the bisimilarity problem or whether they are
spurious counter-examples due to the fact that open bisimilarity in incomplete.
If a spurious counter-example is discovered, then the reason why it is spurious is used to refine and resume the search for a bisimulation.
If this game terminates, we should have constructed either a bisimulation (thereby proving the unlinkability property) or a modal logic formula explaining a distinguishing strategy (thereby discovering an attack on unlinkability).
Another notable feature of the method in this work is that we show that unlinkability problems,
which are traditionally expressed as a weak bisimilarity problem,
can be reduced to a strong bisimilarity problem, thereby ensuring the transition system is image finite, considerably simplifying the problem.
This is, in itself, a contribution of this work, since notions of strong bisimilarity had not previously been defined for the applied $\pi$-calculus,
nor had its characteristic modal logic previously been defined, for which we provide soundness and completeness results.

\subsubsection*{How to read this paper.}
The focus of this paper is on analysing unlinkability properties of ePassport protocols.
During the course of our discussion on unlinkability, we introduce various methods which play a role in our formal analysis.
In order to follow these methods some knowledge of the $\pi$-calculus and bisimilarity is a prerequisite.
It is not necessary to have knowledge of the applied $\pi$-calculus, since we introduce a state-of-the-art presentation of the semantics of the applied $\pi$-calculus facilitating the translation of recent advances in the theory of the $\pi$-calculus to the setting of the applied $\pi$-calculus.
We move quickly through such definitions, in order to get to the point, which is to explain how the methods are used to discover attacks on unlinkability.

In Sections~\ref{sec:strong} and~\ref{sec:symbolic}, we explain our methodology and how it can be used to efficiently analyse problems such as the unlinkability of ePassport protocols, thereby proving that we have closed the question of whether the original formulation of the unlinkability of ePassport protocol BAC is violated.
Section~\ref{sec:II} reflects on established notions of unlinkability, thereby making a case for instead employing a new notion of unlinkability which makes the realistic assumption that an attacker can distinguish between communications originating from different ePassport sessions.
Section~\ref{sec:pace} demonstrates that there is a similar attack on the unlinkability of the latest ePassport protocol PACE\@. The existence of a new attack on PACE, similar to the attack we discovered on BAC, is confirmed using our methodology.
We conclude in Sections~\ref{sec:related} and~\ref{sec:conclusion}, by acknowledging the wider discussion on ePassport privacy and unlinkability to which this paper contributes.

%% file: protocols.tex
\section{Reducing strong unlinkability to a strong bisimilarity problem}\label{sec:strong}

The ICAO 9303 standard recommends two authentication protocols for ePassports.
The Basic Access Control protocol (BAC) was the authentication protocol originally proposed.
The Password Authenticated Connection Establishment protocol (PACE) was added in the $7^{th}$ edition of the standard released in 2015.

In this section, here we briefly explain the BAC protocol and show how it can be modelled as processes in the applied $\pi$-calculus.
We capture a version of the BAC protocol implemented in UK ePassports, as defined in CSF'10~\cite{Arapinis2010}. We should clarify that  the UK version of the BAC protocol
captures the way countries should implement the BAC protocol; hence our analysis is not limited to UK ePassports -- it applies to ePassport worldwide.
We focus, in the next two sections, on a methodology that we used to discover unlinkability attacks on the BAC protocol.
An analysis of the PACE protocol appears later in Sec.~\ref{sec:pace}.

\subsection{The BAC protocol.}

\begin{figure}[b]
\centerline{
\includegraphics[width=0.9\textwidth]{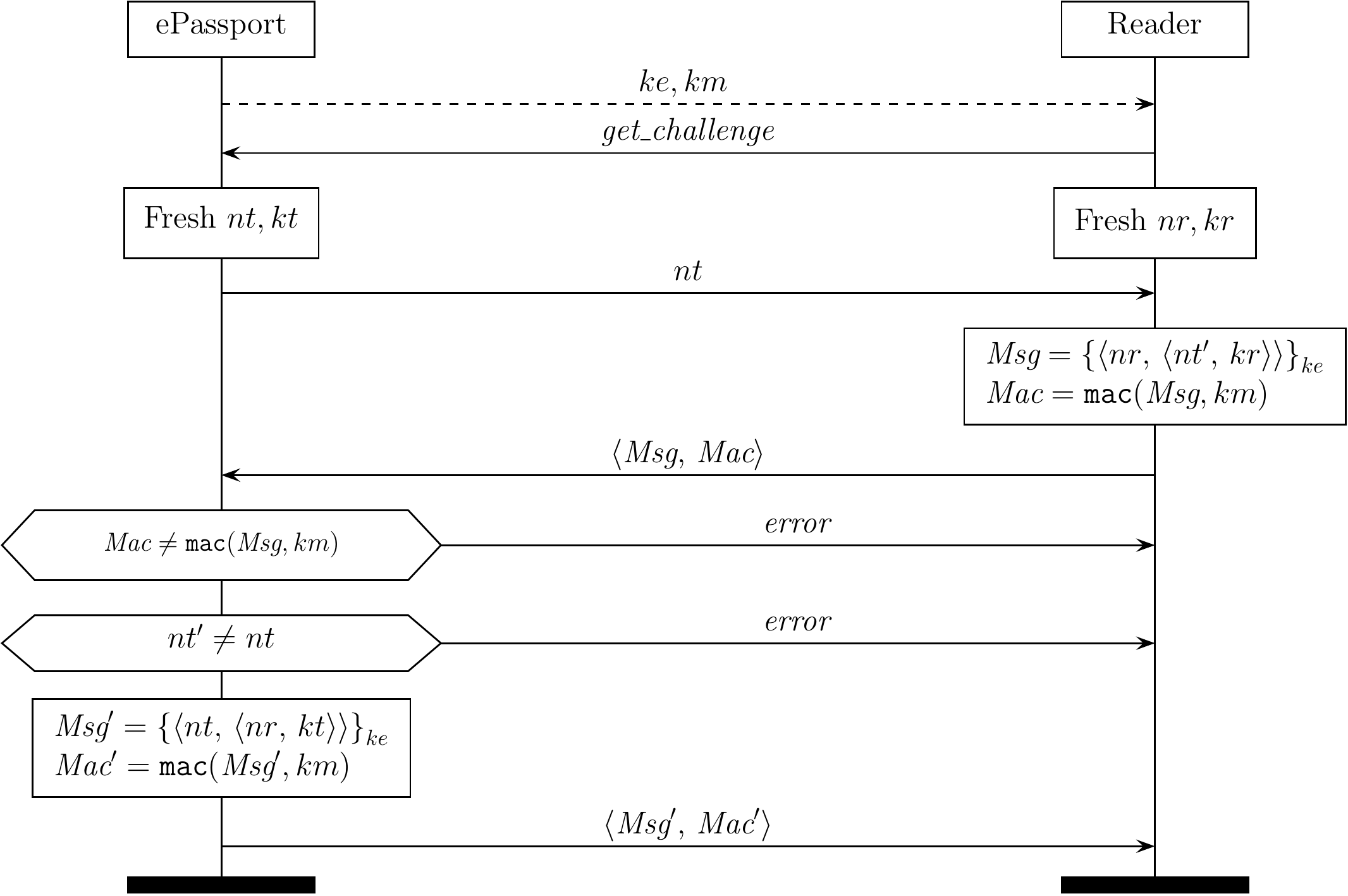}}
%\caption{The ICAO 9303 basic access control mechanism.}
    \caption{The BAC protocol with one error message for all reasons for failure.
% the attacker may intercept.
%\ZS{Reviewer 3: `get' is ambiguous, use   `get\_challenge'}
}%
\label{fig:BAC}
%\vspace{-5ex}
\end{figure}

The BAC protocol is sketched informally in Fig.~\ref{fig:BAC}.
Dashed lines $(\dashrightarrow)$ indicate a message transmitted via
an OCR session that reads the personal page of an ePassport.
The data read in the initial OCR session is used to calculate the symmetric keys $ke$ and $km$ used respectively for encryption and as the seed of message authentication codes (MACs).
Importantly, these keys are the same for every session involving the same ePassport.
Solid lines are wireless communications between a chip embedded in the ePassport and radio frequency reader.

The reader first
sends a constant message $\getchallenge$ requesting a challenge -- a nonce $nt$ sent by the ePassport -- which is
used during the mutual authentication of the ePassport and reader.
The reader shows it has the keys
by responding to the challenge with a message including nonce $nt$ encrypted and authenticated
using the keys, thereby authenticating the reader from the perspective of the ePassport.
In that message, the reader sends its own challenge $nr$, which the ePassport must respond to.
The ePassport responds to the reader with a message involving nonces $nr$ and $nt$
%\SM{The MSC in Fig. 1 uses $nt'$ at the Reader side in stead of $nt$. Is that a typo?}
encrypted and authenticated using the keys, thereby authenticating the ePassport to the reader.
Notice only the ePassport that shared keys $ke$ and $km$ and sent challenge $nt$ can respond in this way, assuming the keys are never exchanged with a malicious $3^{rd}$-party.

We can be precise about the functional properties that BAC achieves.
Firstly, BAC achieves an authentication property called (injective) \textit{agreement}~\cite{Lowe1997}.
Secondly, BAC establishes shared secrets $kr$ and $kt$ which are used to generate a symmetric key for transmitting personal data.
These properties are easily checked using automated tools such as Scyther~\cite{Cremers2008}.

We will see that, for unlinkability, the error branches in the protocol have an important role.
The ICAO 9303 standard specifies that an ``operating system dependent
error''~\cite{MRTD} should be sent when authentication fails.
Such a failure occurs when, upon the ePassport receiving an authentication request, either the message authentication code (MAC) is wrong, or a nonce in the message does not match the challenge $nt$ previously sent by the ePassport.
In this work, we assume all ``operating system dependent error'' messages are the same, since distinctions between error messages lead to known serious attacks, such as those discovered for an implementation of the French ePassport~\cite{Arapinis2010,originalBACAttack}.
Thus we consider the scenario where an ePassport manufacturer does not make the mistake of introducing this well known potential implementation flaw hence any attack we discover is valid for ePassport implementations worldwide.

\subsection{BAC in the applied \texorpdfstring{$\pi$}{pi}-calculus}

\begin{figure}[b]
\small
\[
\begin{array}{rlr}
P, Q \Coloneqq& 0 &\!\!\!\!\!\! \mbox{deadlock} \\
          \mid& \cout{M}{N}.P & \mbox{send} \\
          \mid& \cin{M}{y}.P & \mbox{receive} \\
          \mid& \texttt{if}\,M = N\,\texttt{then}\,P\,\texttt{else}\,Q & \mbox{branch} \\
          \mid& \mathopen{\left[M = N\right]}P & \mbox{match} \\
          \mid& \mathopen{\left[M \not= N\right]}P & \mbox{mismatch} \\
%          \mid& P + Q & \mbox{choice} \\
          \mid& \mathopen\nu x. P & \mbox{new} \\
          \mid& P \cpar Q & \mbox{parallel} \\
          \mid& \bang P &\!\!\!\!\!\!\!\!\! \mbox{replication}
\end{array}
\quad
\begin{gathered}
\begin{array}{rlr}
M, N \Coloneqq& x & \mbox{variable} \\
          \mid& \mac{M, N} & \mbox{mac} \\
          \mid& \left\langle M, N\right\rangle & \mbox{pair} \\
          \mid& \fst{M} & \mbox{left} \\
          \mid& \snd{M} & \mbox{right} \\
          \mid& \enc{M}{N} &\!\!\!\! \mbox{encryption} \\
          \mid& \dec{M}{N} & \mbox{decryption} \\
\end{array}
\\
\begin{array}{c}
\fst{\pair{M}{N}} \mathrel{=_{E}}M
\quad
\snd{\pair{M}{N}} \mathrel{=_{E}}N
\\
\dec{\enc{M}{K}}{K} \mathrel{=_{E}}M
\quad
\enc{\dec{M}{K}}{K} \mathrel{=_{E}}M
\end{array}
\end{gathered}
\]
\caption{
A syntax for applied $\pi$-calculus processes
with a message theory $E$.
}%
\label{figure:messages}
\end{figure}

We employ the applied $\pi$-calculus to model the operational behaviour of participants in the protocol and how they are combined to form a system.
The syntax of processes is
presented in Fig.~\ref{figure:messages}, along with a message theory featuring pairs and symmetric encryption
(encryption using a shared secret key).
The message theory also features a MAC function -- a cryptographic hash function with no equations.

For readability, we employ the abbreviation
$\clet{{x}}{{M}}{P} \triangleq P\sub{{x}}{{M}}$ in the following specifications of an ePassport ($\MainUK$) and ePassport reader ($\Reader$).
\[
\begin{array}{rl}
\MainUK \triangleq&
\begin{array}[t]{l}
 \mathopen{\cout{c_k}{ke, km}.d(x).}
 \match{ x = \getchallenge }\nu nt.\cout{c}{nt}.d(y). \\
     %& \clet{m_e,m_m}{\fst{y},\snd{y}} \\
 \texttt{if}\,\snd{y} = \mac{\fst{y}, km}\,\texttt{then} \\
\qquad
 \begin{array}[t]{l}
 \texttt{if}\,nt = \fst{\snd{\dec{\fst{y}}{ke}}}\,\texttt{then} \\
%      & \begin{array}[t]{ll}
%        \texttt{then}
\qquad
 \begin{array}[t]{l}
                         \nu kt.\clet{m}{\enc{\pair{nt}{\pair{\fst{\dec{\fst{y}}{ke}}}{kt}}}{ke}} \\
                         \cout{c}{m, \mac{m, km}}
                         \end{array}  \\
 \texttt{else}\,\cout{c}{error}
\end{array}
\\
 \texttt{else}\,\cout{c}{error}
        \end{array}
\\\\
\Reader%
\triangleq&
\begin{array}[t]{l}
 c_k(x_k).\cout{c}{\getchallenge}.d(nt).\nu nr.\nu kr.
\\
 \clet{m}{\enc{\pair{nr}{\pair{nt}{kr}}}{\fst{x_k}}}
 \cout{c}{m, \mac{\pair{m}{\snd{x_k}}}}
\end{array}
\end{array}
\]

We can express the system and idealised specification, respectively, as follows.
\[
\Sys
\triangleq
 \mathopen{\nu c_k.}\left( !\Reader \cpar {!\nu ke. \nu km. !\MainUK} \right)
\]
%\SM{The ke,km message is swapped in comparison to the km,ke message in the MSC in figure 1.}
In the system above, the private channel $c_k$
is used to initiate a session between an ePassport and Reader. The use of a private channel in this way is a modelling trick to hide all information exchanged via OCR sessions that we assume cannot be intercepted using wireless technology.
Notice that the keys $ke$ and $km$ serve as the identity of each ePassport, since they are fixed for an ePassport when it is manufactured.
Thus the innermost replication in $!\nu ke. \nu km. !\MainUK$ models the fact that the same ePassport may be used across multiple sessions, while the outermost replication indicates that
there are many different ePassports, each employing distinct keys.

\subsection{Strong unlinkability and bisimilarity}

We formulate strong unlinkability as an equivalence problem by setting out to show that $\Sys$,
as defined above, is equivalent to an idealised specification of the system that trivially satisfies unlinkability.
The idealised specification models a more restricted variant of the system where each ePassport is used only once
-- as if, once an ePassport is read, it is destroyed and a new ePassport is issued for any future sessions.
We employ the following process to model the specification.
\[
\Spec \triangleq \mathopen{\nu c_k.}\left( !\Reader \cpar {!\nu ke. \nu km. \MainUK} \right)
\]
Notice the only difference between $\Sys$ and $\Spec$ is
the absence of replication after the generation of the key.
Thus, in $\Spec$, each new session is with a new ePassport with a freshly generated key.
Trivially, there is no way to link two sessions with the same ePassport in the above specification.

We specify unlinkability by stating that it holds whenever $\Sys$ and $\Spec$ are equivalent from the perspective of an attacker.
In principle, the idea is that, if an attacker cannot tell the difference
between a scenario where the same tag is allowed to be used in
multiple sessions and the scenario where each tag is really
used once, then you cannot link two uses of the same tag.

In formulations of strong unlinkability, when we say ``equivalent'', we mean equivalent with respect to a particular notion of bisimilarity called \textit{weak early bisimilarity}. We should  avoid potential confusion of terminology: ``strong'' in the context of unlinkability does not refer to the process equivalence, but instead the particular formulation of unlinkability as an equivalence problem, rather than as a property of traces used in earlier work on the topic~\cite{Deursen08}.
%We use a state-of-the-art approach to bisimilarity for the applied $\pi$-calculus employing the labelled transition system in Fig.~\ref{fig:lts}
In what follows, we briefly present a formulation of weak early
bisimilarity for the applied $\pi$-calculus. Our presentation makes
use of processes extended with the knowledge of the attacker and an early labelled
transition system which simplifies the analysis of bisimilarity problems.
%Note the presentation we adopt here makes it relatively easy to quickly discover our attack.
% on strong unlinkability.

We follow the convention that
labelled transitions are always defined directly on extended processes in
normal form.
Adopting normal forms removes
the need for several additional conditions that must be imposed
in older formulations
of bisimilarity for the applied $\pi$-calculus~\cite{Abadi2017}.

\begin{defi}[extended processes in normal form]\label{def:prelim}
Extended processes  $\mathopen{\nu \vec{x}.}\left(\sigma \cpar P\right)$
consist of a set of restricted names $\vec{x}$, a substitution $\sigma$ mapping variables to messages, and an applied $\pi$-calculus process $P$.
We write $\nu x_1.\nu x_2. \ldots \nu x_n.P$ as $\nu x_1, x_2, \ldots x_n.P$.
 The set of free variables for process terms are as standard,
where $\nu x.P$ and $M(x).P$ bind $x$ in $P$, and process terms are always treated modulo $\alpha$-conversion.
We say that a variable $x$ is fresh for a term $P$ (processes or messages) whenever the variable does not appear free in the term, i.e., $x \not\in \fv{P}$.
A variable $x$ is said to be fresh for a substitution $\sigma$ whenever $x\sigma = x$ and, for all $y$, either $x$ is fresh for $y\sigma$ or $x = y$, i.e., $\sigma$ does not change or use $x$ in any way.
Freshness extends in the obvious point-wise fashion to sets of variables, terms and substitutions.

In this work, we always assume extended processes are in normal form
meaning they are subject to the restriction
that the variables $\dom{\sigma}$ (i.e., those variables $z$ such that $z \neq z\sigma$) are fresh for $\vec{x}$,
$\fv{P}$ and $\fv{y\sigma}$, for all variables $y$ (i.e., $\sigma$ is
idempotent, and substitution $\sigma$ has already been applied to $P$).
The substitution in an extended process is referred to as an \textit{active substitution}.

We require the following definitions for composing extended processes
in parallel and with substitutions,
defined whenever $z$ is fresh for $B$ and $\rho$,
 and also $\dom{\sigma} \cap \dom{\theta} = \emptyset$.
% XXX Layout XXX
% Negative space
%\vspace{-.1cm}
\begin{gather*}
\sigma \cpar \theta \cpar Q
\triangleq
\sigma\cdot\theta \cpar Q
\qquad\qquad
(\sigma \cpar P) \cpar (\theta \cpar Q) \triangleq \sigma\cdot\theta \cpar (P \cpar Q)
\\[5pt]
\rho \cpar \nu z.A \triangleq \mathopen{\nu z.}\left(\rho \cpar A\right)
\qquad\quad
B \cpar \mathopen{\nu z.} A
\triangleq
\mathopen{\nu{z}.} \left( B \cpar A \right)
\qquad\quad
\mathopen{\nu z.} A \cpar B
\triangleq
\mathopen{\nu{z}.} \left( A \cpar B \right)
\end{gather*}
\end{defi}

We require a standard notion of static equivalence, which
checks two processes are indistinguishable in terms of the messages
output so far.
\begin{defi}[static equivalence]\label{def:static}
Extended processes in normal form
$\mathopen{\nu\vec{x}.}\left( \sigma \cpar P \right)$
and
$\mathopen{\nu\vec{y}.}\left( \theta \cpar Q \right)$
are statically equivalent
whenever,
for all messages $M$ and $N$ such that
$\vec{x} \cup \vec{y}$ are fresh for $M$ and $N$,
we have
$M\sigma \mathrel{=_{E}}N\sigma$ if and only if $M\theta \mathrel{=_{E}}N\theta$.
\end{defi}

%\vspace{-.1cm}
The above definitions are employed in our definition of ``early'' labelled transitions (Fig.~\ref{figure:active}), which are defined directly
on extended processes in normal form. Labels on transitions are
either: $\tau$ -- an internal communication; $\co{M}(z)$ -- an output on
channel $M$ binding the output message to variable $z$; or $M\,N$
-- an input on channel $M$ receiving message $N$.
Define the bound variables such that $\bn{\pi} = \left\{ x \right\}$ only
if $\pi = \co{M}(x)$ and $\bn{\pi} = \emptyset$ otherwise.
Define the free variables such that $\n{M\,N} = \fv{M} \cup \fv{N}$, $\n{\co{M}(x)}
= \fv{M} \cup \left\{x\right\}$ and $\n{\tau} = \emptyset$.
These sets are not disjoint, due to the context in which these definitions are used.

Notice, in this labelled transition system, \texttt{if-then-else}, match and mismatch inherit their actions from the processes they guard, which is traditional for the $\pi$-calculus. This contrasts to established reduction semantics~\cite{Abadi2017} for the applied $\pi$-calculus, where \texttt{if-then-else} statements perform additional $\tau$-transitions in order to resolve guards.
This design decision will enable us to provide genuine ``strong'' counterparts to the weak equivalences that we define.

\begin{figure}
\[
\begin{gathered}
\begin{array}{c}
\begin{prooftree}
M \mathrel{=_{E}} M'
\quad
N \mathrel{=_{E}} N'
\justifies
 \mathopen{\cin{M}{x}.}P \lts{M'\,N'} {P\sub{x}{N}}
\using
\mbox{\textsc{Inp}}
\end{prooftree}
\qquad
\begin{prooftree}
M \mathrel{=_{E}} M'
\quad
\mbox{$x$ is fresh for $M, N, M', P$}
%\cup \env
\justifies
 \cout{M}{N}.P \lts{\co{M'}(x)} {\sub{x}{N}} \cpar P
\using
\mbox{\textsc{Out}}
\end{prooftree}
\\[24pt]
\begin{prooftree}
 A \lts{\pi} B
\quad
x \not\in \mathrm{n}(\pi)
\justifies
  {{\nu x.A}\lts{\pi}{\nu x.B}}
\using
\mbox{\textsc{Res}}
\end{prooftree}
\qquad
\begin{prooftree}
  {P \lts{\pi\sigma} A}
\quad
\mbox{$\bn{\pi}$ is fresh for $\sigma$}
\justifies
  {{{\sigma} \cpar P} \lts{\pi} {{\sigma} \cpar A}}
\using
\mbox{\textsc{Alias}}
\end{prooftree}
\\[24pt]
\begin{prooftree}
  P \lts{\pi} A
\qquad
 M \mathrel{=_{E}} N
\justifies
  {\match{M = N}P \lts{\pi}{A}}
\using
\mbox{\textsc{Mat}}
\end{prooftree}
\qquad\qquad
\begin{prooftree}
  P \lts{\pi} A
\qquad
 M \mathrel{=_{E}} N
\justifies
  {\texttt{if}\,M = N\,\texttt{then}\,P\,\texttt{else}\,Q \lts{\pi}{A}}
\using
\mbox{\textsc{Then}}
\end{prooftree}
\\[24pt]
\begin{prooftree}
  P \lts{\pi} A
\qquad
 M \mathrel{\neq_{E}} N
\justifies
  {\match{M \neq N}P \lts{\pi}{A}}
\using
\mbox{\textsc{Mis}}
\end{prooftree}
\qquad\qquad
\begin{prooftree}
  Q \lts{\pi} A
\qquad
 M \mathrel{\neq_E} N
\justifies
  {\texttt{if}\,M = N\,\texttt{then}\,P\,\texttt{else}\,Q \lts{\pi}{A}}
\using
\mbox{\textsc{Else}}
\end{prooftree}
\\[24pt]
\qquad
\begin{prooftree}
  {P \lts{\pi} A}
\quad
\mbox{$\bn{\pi}$ is fresh for $Q$}
\justifies
  {{P \cpar Q} \lts{\pi} {A \cpar Q}}
\using
\mbox{\textsc{Par-l}}
\end{prooftree}
\qquad
\begin{prooftree}
 P \lts{\pi} A
\justifies
 \bang P \lts{\pi} A \cpar \bang P
\using
\mbox{\textsc{Rep-act}}
\end{prooftree}
\\[24pt]
\begin{prooftree}
  P \lts{\co{M}(x)} \nu \mathopen{\vec{z}.}\left({\sub{x}{N}} \cpar P'\right)
\qquad
  Q \lts{M\,N} Q'
\qquad
\mbox{$\left\{x \right\} \cup \vec{z}$ are fresh for $Q$}
\justifies
  {P \cpar Q}\lts{\tau}{\mathopen{\nu \vec{z}.}\left(P' \cpar Q'\right)}
\using
\mbox{\textsc{Close-l}}
\end{prooftree}
\\[24pt]
\begin{prooftree}
 P \lts{\co{M}(x)} \mathopen{\nu \vec{z}.}\left({\sub{x}{N}} \cpar Q\right)
\qquad
 P \lts{M\,N} R
\qquad
 \mbox{$\vec{z}$ are fresh for $P$}
\justifies
 \bang P \lts{\tau} \mathopen{\nu \vec{z}.}\left( Q \cpar R \cpar \bang P\right)
\using
\mbox{\textsc{Rep-close}}
\end{prooftree}
\end{array}
\end{gathered}
\]
\caption{An \textit{early} labelled transition system, plus symmetric rules for parallel composition.
%\textsc{Close-l}, \textsc{Par-l}, and \textsc{Close-l}.
}\label{figure:active}
\end{figure}

The early labelled transition system and
static equivalence together can be used to define weak early bisimilarity.
Since, initially, we employ a weak formulation of early bisimilarity,
we make use of weak transitions $A \Lts{\pi} B$ which allow zero or more $\tau$-transitions to occur before and after the transition $\pi$, or zero transitions if $\pi = \tau$.
\begin{defi}[weak early bisimilarity]\label{definition:early}
A %\textbf{open}
symmetric relation between extended processes
$\mathrel{\mathcal{R}}$ is a weak early bisimulation only if,
whenever $A \mathrel{\mathcal{R}} B$ the following hold:
\begin{itemize}
\item $A$ and $B$ are statically equivalent.
\item If $A \lts{\pi} A'$ there exists $B'$ such that $B \Lts{\pi}
B'$ and $A' \mathrel{\mathcal{R}} B'$.
\end{itemize}
Processes $P$ and $Q$ are weak early bisimilar, written $P \approx Q$, whenever there exists a weak early bisimulation $\mathcal{R}$ such that $P \mathrel{\mathcal{R}} Q$.
\end{defi}
Now we have the formal tools to express the theorem that confirms that strong unlinkability does not hold for the BAC protocol.
\begin{thm}\label{thm:BAC}
$\Sys \not\approx \Spec$.
\end{thm}
The above is the theorem rectifying the flawed claim, communicated in CSF'10~\cite{Arapinis2010}, that unlinkability holds for this formulation of the BAC protocol.
Much of the rest of the paper is dedicated to explaining the methodology we used to prove the above result, by constructing an attack strategy invalidating the claim in Sec.~\ref{sec:symbolic}.
Later in Sec.~\ref{sec:II} we will also make a case for adjusting the model and in Sec.~\ref{sec:pace} we will show how the analysis can be repeated for PACE\@.

%\section{Methodologically discovering an attack on unlinkability}\label{sec:api}

\newcommand{\Pass}{\textit{P}}
\newcommand{\R}{\textit{V}}

\subsection{Reducing weak to strong bisimilarity}

A challenge with the CSF'10~\cite{Arapinis2010} specification of unlinkability is that it is formulated using weak transitions, which are not image finite.
%notion of bisimilarity, and the weak transitions system generated by $\Sys$ and $\Spec$ is not image finite.
\begin{defi}
A labelled transition system, given by a relation say $\lts{}$, is image finite for a process $A$,
whenever for any label $\pi$ there are finitely many $B$ such that $A \lts{\pi} B$, up to $\alpha$-conversion.
\end{defi}
The strong labelled transition relation $\lts{}$ defined in Fig.~\ref{figure:active} is image finite for all extended processes; whereas its corresponding weak labelled transition relation $\Lts{}$ is only image finite for some extended processes.
In particular, $\Lts{}$ is not image finite for processes $\System$ and $\Spec$ that are used to specify the unlinkability problem.
To see this observe
%the transition system generated by processes $\Sys$ and $\Spec$ does not have the image-finiteness property, where image finiteness is defined, as follows, as a property of the labelled transition system and processes involved.
%For example,
there are infinitely many states reachable by $\tau$-transitions
from $\Spec$ of the following form, where $n$ sessions have started by communicating on the private channel $c_k$.
\[
 \Spec \Lts{}
\begin{array}[t]{l}
 \mathopen{\nu c_k, ke_1, km_1, \ldots ke_n, km_n.}\Big(
\begin{array}[t]{l}
     \R(ke_1,km_1) \cpar  \ldots \R(ke_n,km_n) \cpar !\Reader \cpar
\\
     \Pass(ke_1,km_1) \cpar  \ldots \Pass(ke_n,km_n) \cpar {!\nu ke. \nu km. \MainUK}~\Big)
\end{array}
\end{array}
\]
where
\[
\begin{array}{rl}
\Pass(ke,km) \triangleq&
\begin{array}[t]{l}
 \mathopen{d(x).}
 \match{ x = \getchallenge }\nu nt.\cout{c}{nt}.d(y). \\
     %& \clet{m_e,m_m}{\fst{y},\snd{y}} \\
 \texttt{if}\,\snd{y} = \mac{\fst{y}, km}\,\texttt{then} \\
\qquad
 \begin{array}[t]{l}
 \texttt{if}\,nt = \fst{\snd{\dec{\fst{y}}{ke}}}\,\texttt{then} \\
%      & \begin{array}[t]{ll}
%        \texttt{then}
\qquad
 \begin{array}[t]{l}
                         \nu kt.\clet{m}{\enc{\pair{nt}{\pair{\fst{\dec{\fst{y}}{ke}}}{kt}}}{ke}} \\
                         \cout{c}{m, \mac{m, km}}
                         \end{array}  \\
 \texttt{else}\,\cout{c}{error}
\end{array}
\\
 \texttt{else}\,\cout{c}{error}
        \end{array}
\\\\
\R(ke,km)
\triangleq&
\begin{array}[t]{l}
 \cout{c}{\getchallenge}.d(nt).\nu nr.\nu kr.
\\
 \clet{m}{\enc{\pair{nr}{\pair{nt}{kr}}}{ke}}
 \cout{c}{m, \mac{\pair{m}{km}}}
\end{array}
\end{array}
\]
When we do not have image finiteness we need to find a finite representation of infinitely many processes reachable by a transition, which can make verification challenging.
To simplify verification, we show that the problem of analysing the unlinkability of BAC
can be transformed into an equivalent problem where image finiteness does hold, thereby avoiding the need to explicitly deal with reasoning about transitions such as the above.
The procedure we employ involves removing the $\tau$-transition from the model -- a process known as saturation -- thereby reducing the unlinkability problem to a \textit{strong early bisimilarity} problem.
%\SM{``image finiteness'' plays an important role; should we give a definition?}

We define an alternative system ${\AltSys}$ and specification ${\AltSpec}$, as follows, in bold.
\[
{\AltSys} \triangleq \bang\nu ke. \nu km. \bang\left(\R(ke,km) \cpar \Pass(ke,km)\right)
\]
\[
{\AltSpec} \triangleq \mathopen{\bang\nu ke. \nu km.}\left(\R(ke,km) \cpar \Pass(ke,km)\right)
\]
In the above processes, the keys $ke$ and $km$ have been distributed in advance to the relevant parties; hence a $\tau$-transition is not required to initiate a reader and ePassport with the same keys.
%\SM{I also don't understand why there is no $\nu c_k$ or $\nu c$ in these processes (this probably has to do with the type of bisimulation, but I didn't get that argument then).}\RH{Is this clearer?}
We then show that each process above is bisimilar to the original system and specification, respectively. It is easier to establish a more general result, stated in the lemma below, which can be used to transform the unlinkability problem for related protocols into a form where we have image finiteness. We make use of the
term $a(x_1, x_2, \ldots x_n).P$ as an abbreviation
for $a(x).P\sub{x_1, x_2, \ldots x_n}{\texttt{proj}_1(x), \texttt{proj}_2(x), \ldots \texttt{proj}_n(x)}$,
where $\texttt{proj}_i$ is the obvious generalisation of $\fst{}$ and $\snd{}$ to $n$-tuples.

\begin{comment}
\SM{I don't understand this definition. Is it defining the $\equiv$
symbol? How is $P$ bound? Why is term $a(x_1, x_2, \ldots x_n).P$
defined, while it doesn't occur in the definition?}

\SM{It was not immediately clear that this definition is within the
Proof of the lemma. I wonder even if this proof and the definition of
equivariance should occur in the main body of the text. It distracts
the reader from the main line of reasoning and doesn't seem to play a
pivotal role in the story. It also enlarges the distance between the
lemmas 2.5 and 2.7, while it's essential for the reader to observe
their (dis)similarities. Once you see these two lemmas next to each
other, the following corollary is evident; if they are a full page
apart, understanding the corollary implies some extra work to be done
by the reader.}
\end{comment}

\begin{lem}\label{lemma:strong}
For any $P$ and $Q$ such that $c_k$ is fresh for $P$ and $Q$, we have
\[
\mathopen{\nu c_k.}\left( \bang c_k(\vec{k}).P \cpar \bang\mathopen{\nu \vec{k}.}\cout{c_k}{\vec{k}}.Q \right)
\approx
\bang\mathopen{\nu \vec{k}.}\left(P \cpar Q \right)
\]
\end{lem}
Notice the proof, provided in Appendix~\ref{sec:app1}, is just a sketch. To go through all details would be cumbersome, indicating the amount of work that is saved when applying the above lemma to reduce the complexity of the unlinkability problem we aim to solve.

By a similar argument, we can also establish the following lemma.
\begin{lem}\label{lemma:strong2}
For any $P$ and $Q$ such that $c_k$ is fresh for $P$ and $Q$, we have
\[
\mathopen{\nu c_k.}\left( \bang c_k(\vec{k}).P \cpar \bang \mathopen{\nu \vec{k}.}\bang \cout{c_k}{\vec{k}}.Q \right)
\approx
\bang \mathopen{\nu \vec{k}.} \bang \left(P \cpar Q \right)
\]
\end{lem}

As an immediate consequence of Lemma~\ref{lemma:strong} and Lemma~\ref{lemma:strong2} we obtain.
\begin{prop}\label{cor:strong}
${\Sys} \approx {\Spec}$
if and only if
${\AltSys} \approx {\AltSpec}$.
\end{prop}
Since, in this model of the BAC protocol,
all communications on public channels use channel $c$ for outputs and $d$ for inputs, there are no $\tau$-transitions in ${\AltSys}$ or ${\AltSpec}$.
Thereby, we have reduced the problem to a form where we can apply the strong variant of early bisimilarity, defined as follows.
%\SM{The terminology is a bit confusing. So we have defined ``early bisimilarity'', which is in fact ``weak early bisimilarity'' and now we're defining ``strong bisimilarity'', which is in fact ``strong early bisimilarity''. That doesn't sound like a consistent naming strategy.}
\begin{defi}[strong early bisimilarity]\label{def:strong}
A %\textbf{open}
symmetric relation between extended processes
$\mathrel{\mathcal{R}}$ is a strong early bisimulation only if,
whenever $A \mathrel{\mathcal{R}} B$ the following hold:
\begin{itemize}
\item $A$ and $B$ are statically equivalent.
\item If $A \lts{\pi} A'$ there exists $B'$ such that $B \lts{\pi}
B'$ and $A' \mathrel{\mathcal{R}} B'$.
\end{itemize}
Processes $P$ and $Q$ are strong early bisimilar, written $P \sim Q$, whenever there exists a strong early bisimulation $\mathcal{R}$ such that $P \mathrel{\mathcal{R}} Q$.
\end{defi}
Notice the only difference compared to Def.~\ref{definition:early} is that, in clause two of the definition above, every transition is matched by a single transition -- extra $\tau$-transitions are not permitted.
The following theorem summarises the correctness of the transformation of the unlinkability problem described in this section, which is an immediate consequence of Proposition~\ref{cor:strong} and the absence of $\tau$-transitions in $\AltSys$ and $\AltSpec$.
\begin{thm}\label{thm:strong}
${\Sys} \approx {\Spec}$
if and only if
${\AltSys} \sim {\AltSpec}$.
\end{thm}

%\hspace{-.63cm}
%\fbox{\begin{minipage}{0.98\textwidth}
\subsection{A tribute to Jos Baeten in the language of process equivalences}\label{sec:jos}
%The various terms used to identify process equivalences can become confusing -- ``weak''/``strong'', ``early''/``late''/``open''\ldots
% ``interleaving''/``non-interleaving'', ``branching-time''/``linear-time''\ldots

%\textbf{From weak to strong.}
In this work, we make use of both ``weak'' and ``strong'' equivalences, since the original formulation of strong unlinkability
 was in terms of a weak equivalence, but strong equivalences are
 easier to work with. Indeed, the authors were inspired by a panel
 discussion during the $20^{th}$ edition of CONCUR chaired by Jos
 Baeten, to whom this paper is dedicated on the occasion of his
 retirement. The idea that strong equivalences are easier to work with
 was a point of view raised by Jos Baeten during that panel session.

The above mentioned panel session, during the $20^{th}$ edition of CONCUR, ended with a question from the audience, ``but what can you do with all these process equivalences?'' The response from a  panellist was one word: ``security.'' Indeed, this paper embodies that panel session, since we go deeper into the spectrum of process equivalences in several dimensions in order to obtain results in the security domain.

%\textbf{From early to open.}
Beyond the ``weak'' v.s.~``strong'' dimension, another dimension we exploit
in this work is the distinction between ``early'' and ``open''
equivalences. In the next section, we introduce a notion of ``strong
open'' bisimilarity, which is described in terms of an ``open late''
labelled transition system. Traditionally, the applied $\pi$-calculus
has been endowed with a notion of bisimilarity called ``labelled
bisimilarity,'' which, is little more than an alias for ``weak early''
bisimilarity (Def.~\ref{definition:early}). Our primary reason for
moving from  ``early'' to ``open'' is that the open setting enables
symbolic methods to be directly applied hence is easier to check
systematically. Open bisimilarity should however be applied carefully,
since it is strictly finer than early bisimilarity; indeed, open
bisimilarity is intuitionistic whereas early bisimilarity is
classical~\cite{Ahn2017}. The significance of this insight was
emphasised by Jos Baeten himself at the $28^{th}$ edition of CONCUR during the best paper award ceremony, indicating another way in which his leadership has influenced this paper.
%Note an additional benefit of open bisimilarity is that it is a congruence for open process terms, i.e., those containing free variables, which makes it a robust method to apply.

One might say, ``well, if it's easier to check, why not just fix open bisimilarity as the target equivalence?'' This view doesn't hold up for two reasons. Firstly, the security community are used to weak early bisimilarity, so confidence is increased if we can verify which attacks discovered using open bisimilarity are also valid for weak early bisimilarity.
%We will bridge that gap between early and open bisimilarity, using a methodology illustrated in the next section that makes use of modal logic formulae.
Secondly, taking a fresh position,
open bisimilarity is a little too fine for proving some security and privacy properties, so a better target equivalence for open processes (those containing free variables) would be ``quasi-open'' bisimilarity which balances the qualities of early bisimilarity and open bisimilarity -- a discussion on this appears in a companion report~\cite{ArXiv}. Thus when checking bisimilarity, we require both a notion of open bisimilarity which is easier to explore symbolically, and also a coarser equivalence such as early bisimilarity (or quasi-open bisimilarity) that serves as our target notion of bisimilarity; and, during the search for a proof or a counterexample (an attack), we play a game where we move between these equivalences. This methodology we illustrate in the next section.
For the above reasons, we should be aware of
%\SM{Grammatical problem.}
how to move between ``weak early'', ``strong early'' and ``strong open'' variants of bisimilarity, since they come together to form a methodology for solving unlinkability problems.

%\textbf{Further dimensions and attacker models.}
Going further, we could exploit further dimensions in the spectrum of process equivalences -- a point we return to in Section~\ref{sec:sim}. In particular, we can move along the ``linear-time''/ ``branching-time'' spectrum~\cite{Glabbeek2001} to pick out coarser equivalences than bisimilarity, which can be connected with a spectrum of attacker threat models.
In short, the choice of equivalence can control the testing capabilities of an attacker, which can restrict the space of attacks that we range over when we verify a security or privacy property.
These intermediate definitions can be obtained by taking any of the above mentioned notions of bisimilarity and restricting them in various ways.
 Indeed, the linear-time/branching-time spectrum was the main topic of
 the aforementioned panel discussion chaired by Jos Baeten, and has been a running theme throughout his work~\cite{Baeten87,Andova06,Markovski2012}.
Looking beyond the current paper, there are further uncharted depths to be explored in
terms of exploiting the spectrum of process equivalences to both understand attacker/threat models and to enable new methodologies for verification in the security domain. For example, all equivalences in this work ``interleave'' actions, but there is a spectrum of ``non-interleaving'' or ``truly concurrent'' equivalences that make explicit subtle distinctions that occur when there may be multiple attackers that are not co-located or where the duration of events is significant~\cite{Baeten91,Baeten93,Baeten98}.
This line of inspiration, assimilated into this paper, runs back to the days when the second author was supervised by
Jos Baeten at the University of Amsterdam~\cite{BaetenMauw}, during
which time the inter-personal style of Jos Baeten set a benchmark for the career of the second author.

%Weak/progressing/strong, Early/Late/Qusi-Open/Open, interleaving/non-interleaving, branching-time/linear-time\ldots}
%The language of process equivalences can become confusing.

%\end{minipage}}

%% file: symbolic.tex
\newcommand{\varnt}{nt}
\newcommand{\varntone}{nt_1}

\section{Searching for a bisimulation symbolically}\label{sec:symbolic}

Having reduced unlinkability to a strong bisimilarity problem,
we now aim to prove or disprove ${\AltSys} \sim {\AltSpec}$.
To do so, we attempt to construct a strong early bisimulation $\mathcal{R}$ (Def.~\ref{def:strong}) such that
${\AltSys} \mathrel{\mathcal{R}} {\AltSpec}$.
However, na{\"{\i}}vely searching for a bisimulation using the early labelled transition system in
Fig.~\ref{figure:active} is challenging, since we must consider an infinite number of messages which can be received for every input. And, although it has been shown that checking a bounded number of such messages suffices for message theories such as the one we employ, the bound on the number of messages to check is hyper-exponential~\cite{Huttel2003}.

%\SM{After having introduced early and strong bisimilarity, we now
%introduce open bisimilarity. We should somehow prepare the reader for
%these steps. At the moment it looks like a sequence of magic tricks
%that we introduce at the point where we need them. Maybe provide the
%general ideas of these steps already in the beginning of Section 3 and
%come back to it every time we go to the next representation.}
For this reason, it makes sense to approach the problem using symbolic methods, for which we apply open bisimilarity which is an under-approximation of early bisimilarity --
that is, if two processes are open bisimilar then they are early bisimilar, but not necessarily vice-versa.
%\SM{We need to hint already at this point on the consequences of using an under-approximation. I assume that it means that if two process are open bisimilar they are also early bisimilar.}
Open bisimilarity is suited to symbolic methods, since it uses a call-by-need approach to instantiating inputs where variables representing inputs are only instantiated when they are needed in order to enable a transition.
Due to the fact that open bisimilarity is an under-approximation, care must be taken, since open bisimilarity however may discover certain spurious attacks for the BAC unlinkability problem.\footnote{
The spurious counterexamples arise due to the fact that guards in $\texttt{if-then-else}$
 statements are treated intuitionistically. We leave it to related work to explain why open bisimilarity is intuitionistic~\cite{Ahn2017,Horne2018}, and what spurious examples may arise~\cite{ArXiv}. We will focus here on a counterexample that is not spurious.}
Hence the use of open bisimilarity must be complemented by a methodology for verifying whether an attack discovered using symbolic methods is a real attack or not.

\newcommand{\sstar}[2]{\mathclose{\left(#1\circ#2\right)^*}}

\subsection{Open bisimilarity as a symbolic bisimilarity}

Open bisimilarity is suited to symbolic analysis of protocols, since it permits inputs to be lazily instantiated.
Previously, open bisimilarity has been defined for a slightly less abstract cryptographic calculus, called the \textit{spi-calculus}~\cite{briais06entcs,tiu07aplas,abadi99ic}.
The spi-calculus is less general since it is hard wired with mechanisms for implementing specific equational theories which are abstracted away in the applied $\pi$-calculus, making the applied $\pi$-calculus more concise and allowing it to be instantiated with more theories.

For analysing the unlinkability of ePassports we require the additional power of the applied $\pi$-calculus, hence introduce a notion of open bisimilarity for the applied $\pi$-calculus. This definition of open bisimilarity is a contribution of this paper.
The definitions we provide do have many features in common with notions of symbolic bisimilarity~\cite{Hennessy1995} for the applied $\pi$-calculus, particularly the work of Liu and Lin~\cite{Liu}. We should clarify that such notions of symbolic bisimilarity were never intended to capture open bisimilarity, due to their classical interpretation of constraints; their objective was to directly implement early bisimilarity (or early congruence -- the largest congruence relation contained in early bisimilarity).

%We require less machinery to define open bisimilarity for the applied $\pi$-calculus, since the applied $\pi$-calculus abstracts away from details that can be regarded as implementation concerns.

%\SM{Here we introduce ``late'' in comparison to ``early'' before. The explanation of ``late'' is only given in technical terms that don't help to understand the difference. What is one called ``early'' and the other ``late''? What is the intuition between the difference?}
%\SM{Now we have three types of bisimulation: two which have to do with ``early'', which are called ``early bisimulation'' and ``strong bisimulation'' and one which has to do with ``late'' and which is called ``open''. We put the burden on the reader to see their correspondence.}
Open bisimilarity is defined in terms of an \textit{open late labelled transition system}, presented in Fig.~\ref{fig:late}, where,
like the early labelled transition system in Fig.~\ref{figure:active},
the rules are only well defined for extended processes in normal form.
We firstly explain the ``open late'' terminology (in comparison to ``closed early'', where ``closed'' is the antonym for ``open'' in this setting).
%In addition to definitions in Def.~\ref{def:prelim}, for extended processes in normal form we allow substitutions to be applied as follows.

\paragraph{Late v.s.\ early.} A key difference between these labelled transition systems is that, in a \textit{late} labelled transition system, the input labels are of the form $M(x)$, where $M$ is a message representing a recipe for producing a channel and $x$ is variable which acts as a placeholder for some input message.
%The key feature of such a \textit{late} labelled transition system compared to an \textit{early} labelled transition system, is that way input transitions a handled.
Notice in rule $\textsc{Inp}$ in Fig.~\ref{figure:active} the message input is chosen immediately (from infinitely many possible messages) and hence the input message appears on the input label;
whereas, in rule $\textsc{oInp}$ in Fig.~\ref{fig:late} the input message appears as a variable. The use of a variable means that we do not need to decide immediately which messages should be chosen as inputs; instead, we can instantiate the variable later in a called-by-need fashion, possibly after several steps (subject to some constraints as we will explain below).
The key rules that change to accommodate a late approach to inputs compared to the early approach are the rules \textsc{oInp}, \textsc{oClose-l}, \textsc{oRep-close}.

In order to accommodate the late input labels, we must change slightly the definition of the bound names and free names of a label, compared to the corresponding definition for the early labelled transition system, as follows:
the bound variables are such that $\bn{M(x)} = \bn{\co{M}(x)} = \left\{ x \right\}$ and $\bn{\tau} = \emptyset$; while
the free variables are such that $\n{M(x)} = \n{\co{M}(x)} =  \fv{M}\cup\left\{x\right\}$
and $\n{\tau} = \emptyset$.
These definitions are used in the rules of Fig.~\ref{fig:late}.

\begin{figure*}[ht]
\Small%
\[
\begin{gathered}
\begin{array}{c}
\begin{prooftree}
M \mathrel{=_{E}} M'\theta
\quad
\mbox{$x$ is fresh for $M, M', h, \mathcal{D}, \theta$}
\justifies
h, \mathcal{D}  \colon
\theta
\cpar
\mathopen{\cin{M}{x}.}P \lts{M'(x)}
\theta
\cpar
P
\using
\mbox{\textsc{oInp}}
\end{prooftree}
\\[24pt]
\begin{prooftree}
M \mathrel{=_{E}} M'\theta
\quad
\mbox{$x$ is fresh for $M, M', N, P, h, \mathcal{D}, \theta$}
%\cup \env
\justifies
h, \mathcal{D} \colon
\theta
\cpar
 \cout{M}{N}.P \lts{\co{M'}(x)}
\theta
\cpar
{\sub{x}{N}} \cpar P
\using
\mbox{\textsc{oOut}}
\end{prooftree}
\\[24pt]
\begin{prooftree}
h, \mathcal{D} \colon  \theta
\cpar
P \lts{\pi}
A
\qquad
M \mathrel{=_{E}} N
\justifies
h, \mathcal{D} \colon
\theta
\cpar
\match{M = N} P
\lts{\pi}{
A}
\using
\mbox{\textsc{oMat}}
\end{prooftree}
\qquad
\begin{prooftree}
h, \mathcal{D} \colon  \theta
\cpar
P \lts{\pi}
A
\qquad
M \mathrel{=_{E}} N
\justifies
h, \mathcal{D} \colon
\theta
\cpar
\texttt{if}\,M = N\,\texttt{then}\,P\,\texttt{else}\,Q
\lts{\pi}{
A}
\using
\mbox{\textsc{oThen}}
\end{prooftree}
\\[24pt]
\begin{prooftree}
h, \mathcal{D} \colon
\theta \cpar Q \lts{\pi} A
\quad
h, \mathcal{D}, \theta \vDash M \neq N
\justifies
  {h, \mathcal{D} \colon \theta \cpar \match{M \neq N} Q \lts{\pi}{A}}
\using
\mbox{\textsc{oMis}}
\end{prooftree}
\quad
\begin{prooftree}
h, \mathcal{D} \colon   \theta \cpar Q \lts{\pi} A
\quad
h, \mathcal{D}, \theta \vDash M \neq N
\justifies
  {h, \mathcal{D} \colon \theta \cpar \texttt{if}\,M = N\,\texttt{then}\,P\,\texttt{else}\,Q \lts{\pi}{A}}
\using
\mbox{\textsc{oElse}}
\end{prooftree}
\\[24pt]
\begin{prooftree}
h \cdot x^o, \mathcal{D} \colon
 A \lts{\pi} B
\quad
\mbox{$x$ is fresh for $\pi$, $h$, $\mathcal{D}$}
\justifies
h, \mathcal{D}  \colon {{\nu x.A}\lts{\pi}{\nu x.B}}
\using
\mbox{\textsc{oRes}}
\end{prooftree}
\\[24pt]
\begin{prooftree}
h, \mathcal{D}  \colon {\theta \cpar P \lts{\pi} A}
\quad
\mbox{$\bn{\pi}$ is fresh for $Q$}
\justifies
h, \mathcal{D}  \colon {\theta \cpar {P \cpar Q} \lts{\pi} {A \cpar Q}}
\using
\mbox{\textsc{oPar-l}}
\end{prooftree}
\qquad
\begin{prooftree}
h, \mathcal{D}  \colon \theta \cpar P \lts{\pi} A
\justifies
h, \mathcal{D}  \colon \theta \cpar \bang P \lts{\pi} A \cpar \bang P
\using
\mbox{\textsc{oRep-act}}
\end{prooftree}
%\qquad
%\begin{prooftree}
%h, \mathcal{D}\sigma  \colon {P\sigma \lts{\pi\sigma} A\sigma}
%\quad
% \bn{\pi} \cap \fv{\sigma} = \emptyset
%\justifies
%h, \mathcal{D}  \colon {{\sigma \cpar P} \lts{\pi} {\sigma \cpar A}}
%\using
%\mbox{\textsc{oAlias}}
%\end{prooftree}
%\qquad
\\[24pt]
\begin{prooftree}
h, \mathcal{D}  \colon
\theta \cpar P \lts{\co{M}(x)} \theta \cpar \nu \mathopen{\vec{z}.}\left(\sub{x}{N} \cpar P'\right)
\quad
h, \mathcal{D}  \colon
\theta \cpar Q \lts{M(x)} \theta \cpar Q'
\quad
\begin{array}[b]{l}
 \mbox{$x$ is fresh for $h$, $\mathcal{D}$, $\vec{z}$}
\\
 \mbox{$\vec{z}$ are fresh for $Q$}
\end{array}
\justifies
h, \mathcal{D}  \colon
 {\theta \cpar P \cpar Q}\lts{\tau}{\theta \cpar \nu \vec{z}.\left(P' \cpar Q'\sub{x}{N}\right)} % chktex 1
\using
\mbox{\textsc{oClose-l}}
\end{prooftree}
\\[24pt]
\begin{prooftree}
h, \mathcal{D}  \colon
\theta \cpar  P \lts{\co{M}(x)} \mathopen{\nu \vec{z}.}\left(\sub{x}{N} \cpar Q\right)
\qquad
h, \mathcal{D}  \colon
\theta \cpar  P \lts{M(x)} R
\qquad
% x \not\in \fv{h} \cup \fv{\mathcal{D}} \cup \vec{z}
%\quad
 \vec{z} \cap \fv{P} = \emptyset
\justifies
h, \mathcal{D}  \colon
\theta \cpar  \bang P \lts{\tau} \mathopen{\nu \vec{z}.}\left( Q \cpar R\sub{x}{N}  \cpar \bang P \right)
\using
\mbox{\textsc{oRep-close}}
\end{prooftree}
%\\[5pt]
%\clet{\pair{x}{y}}{\pair{M}{N}} P \reduce P\sub{x,y}{M,N}
%\qquad\mathopen{[ \checksign{\sign{M}{K}}{M, \pk{K}} ]}P \reduce P
\end{array}
\end{gathered}
\]
\caption{
An open late labelled transition system, plus symmetric rules for parallel composition.
%Note, for the fragment without mismatch, transition rules do not carry an environment.
%The equational theory over message terms can be applied at any point in a derivation.
}\label{fig:late}
\end{figure*}

%In our companion paper~\cite{Ahn2018} we show, in the setting of the $\pi$-calculus, how to define open bisimilarity such that privacy properties can be verified.
%This section lifts this results to the setting of the applied $\pi$-calculus, justifies the definition with respect to the literature, and

\paragraph{Open v.s.\ closed.} The keyword \textit{open} in the term \textit{open late labelled transition system} refers to the fact that we allow free variables to appear. Due to the presence of free variables, we must keep track of a constraint system that determines what messages are allowed to be substituted for each free variable.
We succinctly represent these constraints by keeping track of a \textit{history} which records the order in which inputs and outputs occurred,  which allows us to determine which messages had already been output before each input occurs and hence were available to use when performing an input.
 This avoids the possibility of a variable representing an input making use of knowledge from the future. In our representation of constraints, we also employ a set of inequalities between messages $\mathcal{D} = \left\{ M_1 \not= N_1, M_2 \not= N_2, \hdots \right\}$, called a \textit{distinction}. Distinctions are used to symbolically handle inequality constraints that typically arise due to the presence of \texttt{else} branches.

Histories are defined by grammar $h \Coloneqq \epsilon \mid h \cdot x^o \mid h \cdot M^i$,
representing the order in which messages are sent and received. An annotated variable $x^o$ means some output occurred which we refer to indirectly using an alias $x$ where $x$ appears in the domain of some active substitution $\theta$ which is associated with some extended process of the form $\mathopen{\nu \vec{y}.}\left( {\theta} \cpar P\right)$; thus $x\theta$ is the message term that is output, which possibly contains private names, i.e., variables $\vec{y}$ bound by the $\nu$ binder.
The annotated variable $M^i$ represents a larger message that has been input, which, initially is a variable, but may later be a message when the input variable is lazily instantiated.
%\SM{The grammar is not sufficient to understand this notion of ``history''. An example with short explanation would help.}
Notice, in Fig.~\ref{fig:late}, each rule carries a history and distinction that may be used to resolve the \textsc{oElse} rule, by providing sufficient evidence that two messages are not equal (i.e., negation is treated intuitionistically).
Another key differences compared to Fig.~\ref{figure:active} are the updating of the history in rule \textsc{oRes}, which has the effect of further constraining free variables such that none of them may directly refer to any private name. That is, when instantiating inputs, we may not use directly the variables $\vec{x}$ in an extended process of the form $\mathopen{\nu \vec{x}.}\left( {\theta} \cpar P\right)$; we may only refer to messages containing those variables indirectly via the variables in $\dom{\theta}$ that are used as aliases for outputs.

%\SM{We need to give some intuition on what it means for a substitution to respect history. I find the definition hard to follow. Especially since we claim that this is an original contribution of the paper, we should make everything fully clear at both intuitive and technical level.}

\paragraph{The definitions.}
The effect of histories on restricting the substitutions that may be applied, as described above, is captured formally in the following definition.
Substitutions respecting histories, are key to the lazy approach of open bisimilarity.
\begin{defi}[respects]\label{def:open-respects}
Substitution $\sigma$ respects $h$, where $h$ is a history,
whenever for all $h'$ and $h''$
such that $h = h' \cdot x^o \cdot h''$,
we have $x\sigma = x$,
and $y \in \fv{h'}$ implies $x \not\in y\sigma$ (i.e., $x$ is fresh for $h'\sigma$).
%and furthermore for all $M \not= N \in \mathcal{D}$ we have $M\sigma \not=_E N\sigma$.
In the above, $\fv{h'}$ refers to the set of all variables appearing in any term in $h'$.
\end{defi}
For an example, consider the following substitutions and history.
\[
 \sigma = \sub{x,\ y,\ z}{u_1,\ u_2,\ u_3}
\qquad
 \sigma' = \sub{x,\ y,\ z}{u_3,\ u_2,\ u_1}
\qquad
 u_1^o \cdot x^i \cdot u_2^o \cdot y_i \cdot u_3^o \cdot z^i
\]
Observe that,
$\sigma$ respects $h$.
In contrast, $\sigma'$ does not respect $h$, since $x\sigma' = u_3$, which is forbidden since $u_3$, represents an output, which, according to the history, did not occur until after the input $x$.

When applying a respectful substitution $\sigma$
to an extended processes in normal form, with active substitution $\theta$, we must iteratively apply the two substitutions together in order to recover an idempotent substitution, which is a requirement for normal forms. For example, consider $\sigma$ defined above and $\theta$ defined as $\sub{u_1,\ u_2,\ u_3}{n,\ \enc{x}{a},\ \enc{y}{b}}$.
Notice $u_2$ and $u_3$ in the domain of active substitution $\theta$ represent aliases for messages that have been output, $\enc{x}{a}$ and $\enc{y}{b}$ respectively, where each of these messages contain variables, $x$ and $y$ respectively, representing inputs. Thus to find the value of $z$ we must apply $\sigma$ and $\theta$ thrice, that is $z\sigma\theta\sigma\theta\sigma\theta = \enc{\enc{n}{a}}{c}$. When $\sigma$ and $\theta$ are such a respectful-active substitution pair, we always obtain an idempotent substitution after applying at most as many iterations as there are inputs in the history.
The above observations explain why we require the following standard machinery for defining substitutions.
\begin{defi}[substitutions]
Given substitutions $\sigma$ and $\theta$ define $\sigma \circ \theta$ to be the standard  composition of substitutions (i.e., composition of functions).
Acyclic substitutions $\sigma$ are those for which there exists a strict partial order $\sqsubset_\sigma$ over variables such that if $y \in \fv{x\sigma}$ then $x \sqsubset_\sigma y$.
For acyclic substitutions $\sigma$, define  $\sigma^*$ to be the substitution obtained by iteratively composing $\sigma$ with itself until it stabilises, i.e., if $\sigma^0 = id$ (the identity substitution) and $\sigma^{n+1} = \sigma^n \circ \sigma$, then $\sigma^*$ is $\sigma^m$ for some $m$ such that $\sigma^m \circ \sigma = \sigma^m$.
\end{defi}

%\begin{array}
%The application of a substitution $\sigma$ to an extended process $\mathopen{\nu{\vec{x}}.}\left( \theta \cpar P \right)$ is defined as follows, whenever $\vec{x}$ are fresh for $\sigma$
%\[
%\left(\mathopen{\nu{\vec{x}}.}\left( \theta \cpar P \right)\right)\mathclose{\sigma} = \mathopen{\nu{\vec{x}}.}\left( \sstar{\sigma}{\theta}\mathclose{\restriction_{\dom{\theta}}} \cpar P\sstar{\sigma}{\theta} \right)
%\]
%\end{defi}

Given an active substitution $\theta$ and respectful substitution $\sigma$, we can use $\sstar{\sigma}{\theta}$ to obtain a new active substitution.
This trick is used in the following definition of satisfaction, which is used to resolve inequalities in the labelled transition system in Fig.~\ref{fig:late}.
Defining satisfaction is the reason for carrying around constraints, consisting of a history and distinction, at every step in the labelled transition system, since, for some pairs of messages, we can only determine whether they are not equal by observing that the constraints on their variables forbid them from being made equal.
\begin{defi}[satisfaction]
%Satisfaction $h, \mathcal{D}, \theta \vDash M \not= N$ holds whenever there does not exists substitution $\sigma$ respecting $h$ such that $\mathcal{D}\left(\sigma \circ \theta\right)^*$ such that $M\sstar{\sigma}{\theta} \mathrel{=_E} N\sstar{\sigma}{\theta}$.
Satisfaction $h, \mathcal{D}, \theta \vDash M \not= N$ holds whenever there does \textbf{not} exist substitution $\sigma$ respecting $h$ such that:
 \begin{itemize}
\item for all $K \neq L \in \mathcal{D}$, $K\sstar{\sigma}{\theta} \mathrel{\neq_{E}} L\sstar{\sigma}{\theta}$
\item
and $M\sstar{\sigma}{\theta} \mathrel{=_E} N\sstar{\sigma}{\theta}$.
\end{itemize}
\end{defi}

\noindent
Entailment defines a notion of intuitionistic negation, which could be extracted from a Kripke semantics~\cite{Kripke1965}, where the ``reachable worlds'' are those which can be reached by applying substitutions satisfying our constraints (or, equivalently, adding equalities).
What is happening is that, since variables subject to constraints may occur in messages compared using equality or inequality, it is possible that we don't yet have enough information to determine whether or not two messages are equal. In general, two messages may be equal under one substitution of variables but not equal under another substitution. Hence it is useful, in this setting, to say that neither holds yet until we have more information, i.e., we do not assume the law of excluded middle.

For an example of a scenario where the law of excluded middle is violated consider, entailment $u^o\cdot y^i \cdot x_o, {\sub{u}{x}} \vDash x \not= y$. This entailment does \textbf{not} hold yet, since $\sub{y}{u}$ respects history $u^o\cdot y^i \cdot x^o$, and $y\sub{y}{u}\sub{u}{x} = x$, thus there exists a respectful substitution under which these messages are equal, and other substitutions that distinguish them.
Observe also $x \mathrel{=_{E}}y$ also does \textbf{not} hold yet.
Thus, clearly, the law of excluded middle is violated.

In contrast to the above example, consider $y^i \cdot u^o \cdot x^o, {\sub{u}{x}} \vDash x \not= y$. This entailment holds, since there is no substitution $\sigma$ respecting $y^i \cdot u^o \cdot x^o$ such that $x\sigma\theta = y\sigma\theta$, i.e., it is impossible for $x$ and $y$ to be made equal under any permitted substitution.
In other words, it is impossible for an attacker who manufactures input $y$, using their knowledge at the time when $y$ was input, to set $y$ to be equal to private name $x$.

Now consider the following entailments, which make use of distinctions.
\[
\begin{array}{l}
u^o \cdot v^o \cdot x^i \cdot y^o \cdot z^o, x \neq u, {\sub{u,v}{y,z}} \vDash x \neq y
\\
u^o \cdot v^o \cdot x^i \cdot y^o \cdot z^o, x \neq u, {\sub{u,v}{y,z}} \not\vDash x \neq z
\end{array}
\]
The former entailment above holds since the most general substitution $\sigma$ respecting history $u^o \cdot v^o \cdot x^i \cdot y^o \cdot z^o$ such that $x\sigma\sub{u,v}{y,z} \mathrel{=_{E}} y\sigma\sub{u,v}{y,z}$ is $\sigma = \sub{x}{u}$, but that substitution violates the inequality $x \neq u$.
The latter entailment above does not hold since there exists substitution $\sub{x}{v}$ respecting both the history and the distinction such that $x\sub{x}{v}\sub{u,v}{y,z} = z$. Thus under the given history and distinction neither $x \neq z$ nor $x = z$ hold, i.e., the law of excluded middle is violated.

We find it insightful to present an explicit definition of reachability with respect to some substitution. This gives all the extended processes that are reachable from some extended process by applying some substitution, subject to constraints given by histories and distinctions.
%The use of substitution restrictions is used to ensure that the domain of an active substitution is not enlarged by combining it with a respectful substitution.
\begin{defi}[reachability]\label{def:reach}
For a set of variables $V$, let $\sigma \mathclose{\restriction_V}$ be the substitution restricted to the variables in $V$, i.e., if $x \in V$, $x \sigma \mathclose{\restriction_V} = x\sigma$, otherwise $x \sigma \mathclose{\restriction_V} = x$.

Reachability $\leq$
is such that, for history $h$ and $h'$, distinction $\mathcal{D}$ and $\mathcal{D'}$
and extended processes in normal form $A$ and $B$,
we have $h, \mathcal{D}, A
\leq_{\sigma}
h', \mathcal{D}', B$
whenever the following hold:
\begin{itemize}
\item $A = \mathopen{\nu \vec{y}.}\left( P \cpar \theta\right)$;

\item $\sigma$ respects $h$ and $h' = h\sigma$;

\item
For some distinction $\mathcal{E}$,
we have $\mathcal{D}' = \mathcal{D}\sigma  \cup \mathcal{E}\sigma$;

\item
for all $K \neq L \in \mathcal{D} \cup \mathcal{E}$,
we have $K\sstar{\sigma}{\theta} \mathrel{\neq_{E}} L\sstar{\sigma}{\theta}$;

\item $\vec{y}$ are fresh for $\sigma$, $h$, $\mathcal{D}$ and $\mathcal{E}$;

\item
and
$B = \mathopen{\nu{\vec{y}}.}\left( \sstar{\sigma}{\theta}\mathclose{\restriction_{\dom{\theta}}} \cpar P\sstar{\sigma}{\theta} \right)$.
\end{itemize}
\end{defi}

\noindent
Of course, the above is only well defined if $\sigma \circ \theta$ is acyclic.
However, acyclicity of $\sigma \circ \theta$ is an invariant property of reachability, assuming that we start $\theta$ being $id$ and, then generate $\theta$ and $h$ from transitions of our labelled transitions system by recording inputs and outputs as they occur (to be made formal in the definition of open bisimilarity below).
%, hence idempotent such that $\dom{\theta}$ coincides with the outputs in $h$.

Open bisimilarity $\sim_o$ can now be defined as follows, as the largest relation between processes such that there exists an open bisimulation containing the two processes, where all the free variables are treated as initial inputs. Notice this is the strong formulation of open bisimilarity.

%\begin{defi}
%A history $h'$ $\mathcal{D}\cup \mathcal{D}'$ and extended processes $B$ is reachable from $h, \mathcal{D}$ and $A$ by substitution $\sigma$ and set of inequalities $\mathcal{D}'$, written $\reachable{h, \mathcal{D}, A}{\sigma, \mathcal{D}'}{h', \mathcal{D}\cup\mathcal{D}', B}$,
%whenever:
%$\sigma$ respects $h, \mathcal{D}$
% variables $\vec{x}$ are fresh, for some $h''$, $h' = h''\cdot h\sigma$ and
% $B = \mathopen{\nu \vec{z}.}\left(\sub{\vec{x}}{\vec{z}} \cpar A\sigma\right)$
%(notice there is no constraint on the inequalities introduced).
%\end{defi}
%For clarity, we restrict to the fragment without mismatch.
%Hence the late labelled transitions in Fig.~\ref{fig:late} do not need to carry around environment information to resolve mismatches.
%This fragment is adequate for this discussion on related work, since open bisimilarity for the spi-calculus as only previously defined without mismatch.
%Note, to extend with mismatch we must  assumed at a certain point during execution.

\begin{defi}[open bisimilarity]\label{def:open}
A symmetric relation  $\mathrel{\mathcal{R}}$ indexed by a history and distinction is an open bisimulation whenever:
if $A \mathrel{\mathcal{R}}^{h, \mathcal{D}} B$ the following hold, for $x$ fresh for $A$, $B$, $h$, $\mathcal{D}$:
\begin{itemize}
\item
whenever
$h, \mathcal{D}, A \leq_{\sigma} h', \mathcal{D'}, A'$
and
$h, \mathcal{D}, B \leq_{\sigma} h', \mathcal{D'}, B'$,
we have
$A' \mathrel{\mathcal{R}}^{h', \mathcal{D'}} B'$.

\item $A$ and $B$ are statically equivalent.
\item If $h, \mathcal{D} \colon A \lts{\tau} A'$ there exists $B'$ such that $h, \mathcal{D} \colon  B \lts{\tau} B'$ and $A' \mathrel{\mathcal{R}}^{h, \mathcal{D}} B'$.
\item If $h, \mathcal{D} \colon A \lts{\co{M}(x)} A'$, for some $B'$, we have $h, \mathcal{D} \colon  B \lts{\co{M}(x)} B'$
and $A' \mathrel{\mathcal{R}}^{h \cdot x^o, \mathcal{D}} B'$.
\item If $h, \mathcal{D} \colon A \lts{{M}(x)} A'$, for some $B'$, we have $h, \mathcal{D} \colon B \lts{{M}(x)} B'$
and $A' \mathrel{\mathcal{R}}^{h \cdot x^i, \mathcal{D}} B'$.
\end{itemize}
Open bisimilarity $\lsim_o$ is a binary relation over processes defined such that $P \mathrel{\lsim_o} Q$ holds whenever there exists open bisimulation $\mathcal{R}$
such that $P \mathrel\mathcal{R}^{x_1^i\cdot \hdots x_n^i} Q$ holds, where $\fv{P}\cup\fv{Q} \subseteq \left\{x_1, \hdots x_n\right\}$.
\end{defi}
%\SM{I don't think this is sufficient for me to understand open bisimulation.}

The second clause checks static equivalence, as in Def.~\ref{def:static}; but, in contrast to early bisimilarity, due to the first clause we must check static equivalence holds under all substitutions respecting the current history and distinctions, as defined by reachability.
Similarly, the clauses for transitions must be checked under all substitutions permitted by reachability.
The input and output transitions update the history in order to remember which outputs were available at each moment when an input occurs, thereby constraining the permitted substitutions.

\begin{rem}[practical benefits] At first sight, it may appear that closing under all substitutions makes open bisimilarity more difficult to check than early bisimilarity; however, the opposite is true. For many useful equational theories, such as the one featuring basic symmetric encryption used in our model of the BAC protocol, we can calculate a finite set of most general substitutions (and inequalities) that are sufficient to check in order to cover all solutions. This complexity is hidden in the definition of early bisimilarity in the use of early input transitions, where early inputs implicitly ask for all such substitutions and induced inequalities to be checked up front\ldots but we rarely know which to check at the point such inputs occur; hence when checking early bisimilarity we require backtracking that is avoided entirely for open bisimilarity.
The feature of intuitionistic logic that is being exploited here is the fact that intuitionistic constraint systems are monotonic, allowing us to progressively close down the set of constraints without missing anything, whereas classical negation violates this monotonicity property.
\end{rem}

\subsection{Discovering unlinkability attacks by calculation}

We demonstrate our methodology, by showing how attacks on unlinkability can be discovered with minimal heuristics simply as a calculation using open bisimilarity.

The steps illustrated in the following subsections are:
%\SM{This listing helps to understand the structure of what follows. But more information is required to understand the texts explaining these four items and to understand why executing these four solves our problem.}
\begin{itemize}[left=1cm]
\item[3.2.1.] The initialisation of two readers and an ePassport, all with the same keys, w.r.t.\ the system.
\item[3.2.2.] The use of respectful substitutions to refine an input to pass a simple guard, ignoring infinitely many other inputs.
\item[3.2.3.] Exploiting the game behind this bisimilarity problem, to expose a distinguishing strategy.
\item[3.2.4.] Symbolically reasoning about larger messages using the sequent calculus.
%\item[4.2.5.] Symbolically checking static equivalence under respectful substitutions.
\end{itemize}

\noindent
Heuristics are required only for selecting which actions to perform (points 3.2.1.\ and 3.2.3.\ above). The other steps above are calculations that could be formulated as a decision procedure, building on decision procedures for the spi-calculus~\cite{Tiu2010CSF}.
Here we begin by starting up two readers, although a more general heuristic searching for a proof would probably start by starting up $n$ readers in order to eventually construct an inductive definition of an open bisimulation covering the whole state space. We provide two reader sessions, since two suffice for the discovery of the particular attack highlighted.

\subsubsection{Initiate two readers with the same ePassport.}

\newcommand{\varuone}{u_1}
\newcommand{\varutwo}{u_2}

%\SM{It's probably part of our heuristic to choose this scenario. But why do we choose this scenario? How do we know that two readers and one ePassport would give us an attack? What if it doesn't give us an attack?}
Our system, ${\AltSys}$, makes the first moves by starting two reader sessions,
both of which are loaded with the key information of the same ePassport.
This can be achieved by triggering two outputs, which must be $\getchallenge$ messages from readers,
and then sending an input to an ePassport, as performed by the following three transitions.
\[
h_0 \colon {\AltSys} \lts{\co{c}({\color{blue}\varuone})}\lts{\co{c}({\color{blue}\varutwo})}\lts{d({\color{red} x})} \SystemUKone
\]
In the above, $h_0 \triangleq \error^i \cdot \getchallenge^i \cdot c^i\cdot d^i$ is the initial history, which constrains the initial free variables so that they may not be instantiated with private messages that are output later during execution.
We also have $\SystemUKone$ defined as follows (employing abbreviations in Fig.~\ref{fig:abb}), where $\theta_1 = {\sub{{\color{blue} \varuone,\,\varutwo}}{\getchallenge,\,\getchallenge}}$:
\[
      \nu ke_1, km_1. \Big(\theta_1 \cpar \begin{array}[t]{l}
          \Rone(ke_1,km_1) \cpar \UKone(ke_1,km_1,{\color{red} x}) \cpar
          \Rone(ke_1,km_1) \cpar \Pass(ke_1,km_1) \cpar \\
          {\bang\left(\R(ke_1, km_1) \cpar \Pass(ke_1, km_1)\right)} \cpar
          {\bang\nu ke.\nu km.\bang\left(\R(ke,km) \cpar \Pass(ke,km)\right)} \ \Big)
          \end{array}
\]
$\AltSpec$ can only follow these actions, by starting two reader sessions with different ePassports.
\[
h_0 \colon {\AltSpec} \lts{\co{c}({\color{blue}\varuone})}\lts{\co{c}({\color{blue}\varutwo})}\lts{d({\color{red} x})}
\SpecUKone
\]
where $\SpecUKone$ is defined as follows:
\[
      \nu ke_1, km_1, ke_2, km_2. \Big( \theta_1 \cpar\hspace{-1mm} \begin{array}[t]{l}
          \Rone(ke_1,km_1) \cpar \UKone(ke_1,km_1,{\color{red} x}) \cpar \\
          \Rone(ke_2,km_2) \cpar \Pass(ke_2,km_2) \cpar
          {\bang\nu ke.\nu km.\left(\R(ke,km) \cpar \Pass(ke,km)\right)} \ \Big)
	  \end{array}
\]
Note, since open bisimilarity is preserved by associativity and commutativity of parallel composition and equivariance,
we have already also covered the case where, in the specification, the input is received by the ePassport with keys $ke_2$ and $km_2$.
Note there is a third possible response by the specification, where a third session with yet another set of keys is started, but the distinguishing strategy in that branch is subsumed by the distinguishing strategy in the cases we explain here.

The updated history, tracking constraints on variables after these initial three transitions, is:
\[
h_1 = h_0 \cdot {\color{blue} \varuone}^o \cdot {\color{blue} \varutwo}^o \cdot {\color{red} x}^i
\]

\begin{figure}
\begin{gather*}
\mbox{The ePassport (or prover):}
\qquad
\Pass(ke,km)
\triangleq
d(x).\UKone(ke,km,x)
\\[10pt]
\UKone(ke,km,x)
\triangleq
 \match{ x = \getchallenge }\nu nt.\cout{c}{nt}.\UKtwo(ke,km,nt)
\\[10pt]
\UKtwo(ke,km,nt)
\triangleq
d(y).\UKthree(ke,km,nt,y)
\\[10pt]
\UKthree(ke,km,nt,y)
\triangleq
\begin{array}[t]{l}
     %& \clet{m_e,m_m}{\fst{y},\snd{y}} \\
 \texttt{if}\,\snd{y} = \mac{\fst{y}, km}\,\texttt{then} \\
\qquad
 \begin{array}[t]{l}
 \texttt{if}\,nt = \fst{\snd{\dec{\fst{y}}{ke}}}\,\texttt{then} \\
%      & \begin{array}[t]{ll}
%        \texttt{then}
\qquad
 \begin{array}[t]{l}
                         \nu kt.\clet{m}{\enc{\pair{nt}{\pair{\fst{\dec{\fst{y}}{ke}}}{kt}}}{ke}} \\
                         \cout{c}{m, \mac{m, km}}
                         \end{array}  \\
 \texttt{else}\,\cout{c}{error}
\end{array}
\\
 \texttt{else}\,\cout{c}{error}
        \end{array}
\end{gather*}
\begin{gather*}
\mbox{The Reader (or verifier):}
\qquad
\R(ke,km) \triangleq \cout{c}{\getchallenge}.\Rone(ke,km)
\\[10pt]
\Rone(ke,km) \triangleq d(nt).\Rtwo(ke,km,nt)
\\[10pt]
\Rtwo(ke,km,nt)
\triangleq
\begin{array}[t]{l}
 \nu nr.\nu kr.
 \clet{m}{\enc{\pair{nr}{\pair{nt}{kr}}}{ke}}
\\
 \cout{c}{m, \mac{\pair{m}{km}}}
\end{array}
\end{gather*}

\caption{Abbreviations for process used throughout this symbolic analysis.}\label{fig:abb}
\end{figure}

\subsubsection{Applying respectful substitutions.}

Since we are reasoning symbolically, the first input, performed above, is initially a variable ${\color{red} x}$.
When unfolding the rules of the labelled transition in Fig.~\ref{fig:late}, we find that the following transition
is enabled for substitutions $\sigma$ respecting history $h_1$
equating the messages in the guard ${\color{red} x} = \getchallenge$, where a most general unifier is clearly $\sub{{\color{red} x}}{\getchallenge}$.
{\small
\[
\begin{array}{l}
h_1\sub{{\color{red} x}}{\getchallenge} \cdot ke_1^o \cdot km^o_1
\colon
\theta_1 \cpar
 \match{ {\color{red} x}\sub{{\color{red} x}}{\getchallenge} = \getchallenge }\nu nt.\cout{c}{nt}.\UKtwo(ke,km,nt)
\\
\qquad\qquad\qquad\qquad\qquad\qquad\qquad\qquad\qquad\qquad\qquad
\lts{\co{c}({\color{blue} \varv})}
 \nu {\color{blue}\varntone}.\left( \theta_1 \cpar {\sub{{\color{blue}\varv}}{{\color{blue}\varntone}}} \cpar \UKtwo(ke,km,{\color{blue}\varntone}) \right)
\end{array}
\]}
The above transition is valid since, the unifier $\sub{{\color{red} x}}{\getchallenge}$ respects history $h_1 \cdot ke_1^o \cdot km^o_1$,
% and $ke_1, km_1$ are fresh for $h_1\sub{{\color{red} x}}{\getchallenge}$,
which is trivially the case since there are no constraints on unifying variables such as
$\getchallenge$ and $x$.

Using the above transition we induce the following transitions for the system and specification.
\[
h_1 \colon
\SystemUKone\sub{{\color{red} x}}{\getchallenge} \lts{\co{c}({\color{blue}\varv})}
\SystemUKtwo
\qquad\qquad
h_1 \colon
\SpecUKone\sub{{\color{red} x}}{\getchallenge} \lts{\co{c}({\color{blue}\varv})}
\SpecUKtwo
\]
where $\SystemUKtwo$ is defined as follows (employing abbreviations in Fig.~\ref{fig:abb}),
\[
      \nu ke_1, km_1, {\color{blue} \varntone}. \Big(
%          {\sub{{\varuone,\,\varutwo}}{\getchallenge,\,\getchallenge}} \cpar
\theta_1 \cpar
	  {\sub{{\color{blue}\varv}}{{\color{blue} \varntone}}} \cpar
\begin{array}[t]{l}
          \Rone(ke_1,km_1) \cpar \UKtwo(ke_1,km_1,{\color{blue} \varntone}) \cpar \\
          \Rone(ke_1,km_1) \cpar \Pass(ke_1,km_1) \cpar  \\
          {\bang\left(\R(ke_1, km_1) \cpar \Pass(ke_1, km_1)\right)} \cpar \\
          {\bang\nu ke.\nu km.\bang\left(\R(ke,km) \cpar \Pass(ke,km)\right)} \ \Big)
          \end{array}
\]
and $\SpecUKtwo$ is defined as follows.
\[
      \nu ke_1, km_1, ke_2, km_2, {\color{blue} \varntone}. \Big(
	  \theta_1 \cpar
{\sub{\varv}{{\color{blue} \varntone}}} \cpar
\begin{array}[t]{l}
          %{\sub{{\varuone,\,\varutwo}}{\getchallenge,\,\getchallenge}} \cpar
          \Rone(ke_1, km_1) \cpar \UKone(ke_1, km_1, {\color{blue} \varntone}) \cpar \\
          \Rone(ke_2, km_2) \cpar \Pass(ke_2, km_2) \cpar \\
          {\bang\nu ke.\nu km.\left(\R(ke, km) \cpar \Pass(ke, km)\right)} \ \Big)
	  \end{array}
\]
The updated history at this point is $h_2 \triangleq h_1\sub{{\color{red} x}}{\getchallenge} \cdot {\color{blue}\varv}^o$, where the substitution records that the most recent input $x$ was a $\getchallenge$ message. In full, we have at this point:
\[
h_2 \triangleq \error^i \cdot \getchallenge^i \cdot c^i\cdot d^i \cdot \varuone^o  \cdot \varutwo^o \cdot  {\getchallenge}^i \cdot {\color{blue}\varv}^o
\]

\newcommand{\varz}{z}
\newcommand{\nrtwo}{nr_2}
\newcommand{\krtwo}{kr_2}

\subsubsection{Alternating play in the distinguishing game.}
We now appeal to the symmetry of bisimilarity,
allowing the specification $\SystemUKtwo$ to lead with one input.
That strategy allows us to trigger the reader which does not have the same keys as the ePassport that outputs a nonce in the previous step.
That approach leads to a distinguishing game that quite accurately describes a practical strategy, which can be implemented using NFC enabled phones running a modified ePassport reader app, as first reported in the conference version of this paper~\cite{ESORICS}. In this strategy, the attacker deliberately selects the reader that should fail to authenticate an ePassport if unlinkability really holds as modelled by the idealised specification.

The flow is as follows, where annotation $(\dagger)$ indicates the player (the system or specification) that leads at each point in the game. Note the system has always been leading up to now, in order to trigger the scenario where two sessions with the same ePassport really started.
\[
h_2 \colon \SystemUKtwo \lts{d({\color{red} \varnt})}  (\dagger) \SystemUKthree \lts{\co{c}({\color{blue}\varw})}\lts{d({\color{red}\vary})} \SystemUKfour
\]
\[
h_2 \colon
(\dagger) \SpecUKtwo \lts{d({\color{red} \varnt})} \SpecUKthree \lts{\co{c}({\color{blue}\varw})}\lts{d({\color{red}\vary})}  \SpecUKfour
\]
where $\SystemUKfour$ is defined as follows (employing abbreviations in Fig.~\ref{fig:abb}):
\[
\begin{array}{l}
      \nu ke_1, km_1, \varntone, \nrtwo, \krtwo. \Big(
\\
\qquad\begin{array}[t]{l}
          \theta_1 \cpar
	  {\sub{{\varv}}{{\varntone}}} \cpar
 {\sub{{\color{blue} \varw}}{{ \pair{\enc{\pair{\nrtwo}{\pair{{\color{red} \varnt}}{\krtwo}}}{ke_1}}{\mac{   \enc{\pair{\nrtwo}{\pair{{\color{red} \varnt}}{\krtwo}}}{ke_1} , km_1}}}
 }}
\cpar \\
          \Rone(ke_1, km_1) \cpar \UKthree(ke_1, km_1,{\color{red}\varnt},{\color{red}\vary}) \cpar
          0 \cpar \Pass(ke_1, km_1) \cpar \\
          {\bang\left(\R(ke_1, km_1) \cpar \Pass(ke_1, km_1)\right)} \cpar
          {\bang\nu ke.\nu km.\bang\left(\R(ke, km) \cpar \Pass(ke, km)\right)} \ \Big)
          \end{array}
          \end{array}
\]
and $\SpecUKfour$ is defined as follows
\[
\begin{array}{l}
      \nu ke_1, km_1, ke_2, km_2, \varntone, \nrtwo, \krtwo. \Big( \\
\qquad\begin{array}[t]{l}
          \theta_1 \cpar
	  {\sub{\varv}{{\varntone}}} \cpar
 {\sub{{\color{blue} \varw}}{{ \pair{\enc{\pair{\nrtwo}{\pair{{\color{red} \varnt}}{\krtwo}}}{ke_2}}{\mac{   \enc{\pair{\nrtwo}{\pair{{\color{red} \varnt}}{\krtwo}}}{ke_2} , km_2}}}
 }}
\cpar
\\
          \Rone(ke_1, km_1) \cpar \UKthree(ke_1, km_1,{\color{red}\varnt},{\color{red}\vary}) \cpar
          0 \cpar \Pass(ke_2, km_2) \cpar \\
          {\bang\nu ke.\nu km.\left(\R(ke, km) \cpar \Pass(ke, km)\right)} \ \Big)
	  \end{array}	  \end{array}
\]
Observe that all the above transitions proceed lazily without instantiating the input variable ${\color{red} \varnt}$.
In the specification, the reader with keys $ke_2, km_2$ is used up entirely, without determining yet what challenge was received.
Observe also that $\SystemUKtwo$ has only one option, up to structural rules such as commutativity of parallel composition, for following the specification (without being immediately distinguishable), which is to continue a session with keys $ke_1, km_1$.

The updated history at this point records the two inputs and the output in the order they occurred in the above transitions, as follows.
\[
h_4 \triangleq h_2  \cdot {\color{red} \varnt}^i
\cdot {\varw}^o \cdot {\color{red}\vary}^i
\]

\begin{rem}[playing this strategy]
A question arising at this point is whether the change of player at this point is meaningful in terms of attacker models.
In general, to answer such a question we require domain specific knowledge.
Observe that the input action, where the specification leads, selects a specific reader which should ideally behave as if it has different keys from the ePassport issuing the challenge nonce.
In reality, the attacker does indeed have the power to choose which reader will receive an input,
and so can indeed choose the reader that, according to the specification of unlinkability, should not
successfully authenticate with the ePassport, i.e., the reader that is not located next to an ePassport that has just engaged in an OCR session with it. Thus the need for a game at this point is partly due to under-specification in the model where there are insufficient observables to determine that the reader is not in proximity to the ePassport issuing a challenge.
Note this is far from being the only distinguishing strategy;
other distinguishing strategies may require a different domain-specific explanation.
%\SM{The previous sentence has some grammatical problems.}
\end{rem}

\subsubsection{Calculating inputs using the sequent calculus.}

Now consider whether
$\UKthree(ke_1,km_1,{\color{red} \varnt},{\color{red} \vary})$, which is a subprocess of $\SystemUKfour$ shown in expanded form below, can make progress.
\[
\begin{array}[t]{l}
     %& \clet{m_e,m_m}{\fst{y},\snd{y}} \\
 \texttt{if}\,\snd{{\color{red} \vary}} = \mac{\fst{{\color{red} \vary}},\,km_1}\,\texttt{then} \\
\qquad
 \begin{array}[t]{l}
 \texttt{if}\,{\color{red} \varnt} = \fst{\snd{\dec{\fst{{\color{red} \vary}}}{ke_1}}}\,\texttt{then} \\
%      & \begin{array}[t]{ll}
%        \texttt{then}
\qquad
 \begin{array}[t]{l}
                         \nu kt.\clet{m}{\enc{\pair{{\color{red} \varnt}}{\pair{\fst{\dec{\fst{{\color{red} \vary}}}{ke_1}}}{kt}}}{ke_1}} \\
                         \cout{c}{m, \mac{m, km_1}}
                         \end{array}  \\
 \texttt{else}\,\cout{c}{error}
\end{array}
\end{array}
\]
In what follows, we must take into account the active substitution of $\SystemUKfour$, which we recall below and denote by $\theta_4$:
\[
\theta_4
\triangleq
{\sub{{\varuone,\ \varutwo,\ \varv,\ \varw}}{\getchallenge,\ \getchallenge,\ \varntone,\ \pair{\enc{\pair{\nrtwo}{\pair{{\color{red} \varnt}}{\krtwo}}}{ke_2}}{\mac{   \enc{\pair{\nrtwo}{\pair{{\color{red} \varnt}}{\krtwo}}}{ke_2} ,\, km_2}}}}
\]

By the rules in Fig.~\ref{fig:late}, the two \texttt{then} branches of the \texttt{if-then-else} statements above, which result in a non-error output, can only be triggered
for particular substitutions $\sigma$. In particular, we are interested in whether there are substitutions $\sigma$ respecting history $h_4$ such that the two equations below hold,
and, furthermore, $\sigma$ is fresh for the bound variables $\varntone, ke_1, km_1, \nrtwo, \krtwo$ (a constraint enforced by the \textsc{oRes} rule).
\[
\begin{array}{rl}
\snd{{\color{red} \vary}}\sstar{\sigma}{\theta_4}
\!\!\!\!
&=_E \mac{\fst{{\color{red} \vary}},\,km_1}\sstar{\sigma}{\theta_4}
\\
\varntone \sstar{\sigma}{\theta_4}
\!\!\!\!
&=_E \fst{\snd{\dec{\fst{{\color{red} \vary}}}{ke_1}}}\sstar{\sigma}{\theta_4}
\end{array}
\]
It is convenient to select fresh variable $\vary'$ to represent the local view of messages that $\vary$ can be mapped to by the relevant substitution $\sstar{\sigma}{\theta_4}$, i.e., for some substitution $\sigma'$, instantiating $\vary'$ we have $\vary\sstar{\sigma}{\theta_4} = \vary'\sigma'$.
This represents the fact that $\vary$ represents the external view of an observer or attacker when they inject inputs, while $\vary'$ exposes more of the internal structure of messages that cannot be observed by an attacker.
While such additional structure may contain more private information than the attacker is immediately aware of (e.g., because the message represent a cyphertext),
 that information may be required, internally by the process, in order to enable guards such as the guard in the above \texttt{if-then-else} statements.
This leads us to the following equations.
%   is being mapped to their  under the such that $\vary\sigma\theta_4 = y'\sigma$ and $\varnt\theta_4 = \varnt'$. Thus we are interested in the following equations.
\[
\snd{{\color{red} \vary'}} =_E \mac{\fst{{\color{red} \vary'}},\, km_1}
\quad
\mbox{and}
\quad
\varntone =_E \fst{\snd{\dec{\fst{{\color{red} \vary'}}}{ke_1}}}
\]

We show how to calculate a most general unifier for the above equations.
Firstly, we remove destructors $\fst{\cdot}$, $\snd{\cdot}$ and $\dec{\cdot}{\cdot}$ by instantiating variables to which they are applied with the most general form of the constructor to which the destructor is applied.
This yields the following substitution, where ${\color{red} y_1}$ and ${\color{red} y_2}$ are fresh variables.
\[
\sub{\color{red} \vary'}{\pair{
      \enc{\pair{{\color{red}y_1}}{\pair{\varntone}{{\color{red} y_2}}}}{ke_1}
    }{
      \mac{\enc{\pair{{\color{red}y_1}}{\pair{\varntone}{{\color{red}y_2}}}}{ke_1},\, km_1}
    }
}
\]
The problem now is to calculate the most general form of ${\color{red} y_1}$ and ${\color{red} y_2}$, refining the above substitution taking into account the history $h_4$, active substitution $\theta_4$ and bound variables $\varntone, ke_1, km_1, \nrtwo, \krtwo$, as described above.
This question can be formulated as the problem of calculating the most general solutions to a system of \textit{deducibility constraints} which are generated from the above mentioned constraints and active substitution.
%  determine which messages can be used to instantiate each input variable.
%The relevant system of deducibility constraints is generated from history $h_4$, while taking into account the above substitution for $\vary'$ and the active substitution of $\SystemUKfour$, i.e., $\theta_4$.

The first step in this calculation is to generate an intermediate constraint system to solve.
The following is an alternative representation of a history, where the names to the left of a turnstile represent the knowledge of the attacker at the moment when the input message to the right of that turnstile is performed.
\begin{gather*}
\vdash \error
\quad\
\vdash \getchallenge
\quad\
\vdash c
\quad\
\vdash d
\\[5pt]
\varuone, \varutwo \vdash \getchallenge
\quad\
\varuone, \varutwo, \varv \vdash {\color{red} \varnt}
\quad\
\varuone, \varutwo, \varv, \varw \vdash {\color{red} \vary}
\end{gather*}
In this case, it is sufficient to focus on the final two intermediate constraints, although, in general, the initial constraints are essential for ensuring no private information from outputs during execution are used to instantiate the initial knowledge. Also, $\varuone$ and $\varutwo$ provide no new information so can be safely removed from the constraints in order to focus on the essential aspects of the problem.
From the intermediate constraints $
\varv \vdash {\color{red} \varnt}$
and $\varv, \varw \vdash {\color{red} \vary}$, annotated with messages generated by applying the active substitution $\theta_4$ to each of the variables on the left of the turnstile. We also apply the substitution generated for ${\color{red} \vary'}$ above, resulting in two deducibility constraints described below.
% by annotation variables $\varv$ and $\varw$ are annotated with their respective messages from the active substitution in $\SystemUKfour$.

The first deducibility constraint generated is as follows,
where ${\color{red} \varnt'}$ is a fresh variable, which is introduced for the same reason as we introduced ${\color{red} \vary'}$, as explained above.
% i.e., we aim for some $\sigma$ respecting $h_4$ and idempotent substitution $\sigma'$ fresh for $\dom{\theta_4}$ such that $\varnt\sstar{\sigma}{\theta_4} = \varnt'\sigma'$.
\begin{equation}\label{ded1}
\varv \colon \varntone \vdash {\color{red} \varnt} \colon {\color{red} \varnt'}
\end{equation}
Thus, ${\color{red} \varnt' }$ represents any message such that for some suitable substitutions $\sigma$ and $\sigma'$ we have $\varnt\sstar{\sigma}{\theta_4} = {\color{red}\varnt'}\sigma'$, where $\dom{\theta_4}$ are fresh for $\sigma'$. Thus, the difference is that ${\color{red}\varnt}\sigma$ may not refer directly the private names representing various keys and nonces, whereas ${\color{red}\varnt'}\sigma'$ can.

Such deducibility constraint of the form $\Gamma \vdash x \colon x'$, where $x, x'$ are variables, are said to be in \textit{solved form}. This means that $x, x'$ can be any messages produced using information in $\Gamma$ plus some fresh variables, and $x'$ is the local view of $x$ taking into account the current active substitution, following the principles used to explain the use of $\vary'$ and $\varnt'$.
%In this worked example, all leaves are axioms so we do not dwell on the notion of a solved form; except that, notice, without the second constraint
Thus the constraint (\ref{ded1}) generated above is already in solved form, hence, by itself, does not require further analysis.

The second deducibility constraint, generated from intermediate constraint $\varv, \varw \vdash {\color{red} \vary}$ by annotating variables with messages given by the active substitution $\theta_4$, is as follows.
\begin{equation}\label{ded2}
\begin{array}{l}
\varv \colon \varntone,\
\varw \colon
\pair{\enc{\pair{\nrtwo}{\pair{{\color{red}\varnt'}}{\krtwo}}}{ke_1}}{\mac{ \enc{\pair{\nrtwo}{\pair{{\color{red} \varnt'}}{\krtwo}}}{ke_1} , km_1}}
\\
\qquad\qquad\qquad\qquad
 \vdash {\color{red} \vary} \colon {
\pair{
      \enc{\pair{\color{red}y_1}{\pair{\varntone}{\color{red}y_2}}}{ke_1}
    }{
      \mac{\enc{\pair{\color{red}y_1}{\pair{\varntone}{\color{red}y_2}}}{ke_1},\, km_1}
    }
}
\end{array}
\end{equation}
%Also, $r$ and $s$ are fresh variables representing the recipe used to produce ${\color{red} nt}$ and ${\color{red} y}$ respectively, where the annotation associated with $s$ is obtained from the substitution for ${\color{red} y}$ generated above.

We find all solutions to the system consisting of the above deducibility constraints (\ref{ded1}) and (\ref{ded2}), by calculating the most general substitutions such that
there is a proof tree using the sequent calculus rules in Fig.~\ref{fig:constraints}, where the leaves of each proof are either \textit{axioms} or are in \textit{solved form}.
Fig.~\ref{fig:constraints} extends an existing sequent calculus presentation of deducibility constraints~\cite{Tiu2010LMCS} with annotations to the left of a colon representing \textit{recipes} for how a message is deduced.
%The idea of using such annotations is due to, thus far, unpublished work of Alwen Tiu.

%Furthermore, we assume that the names representing outputs ($\varv$ and $\varw$) are atomic names, while $\vary$, $\varnt$, $y_1$, $y_2$, $\varnt'$ are open variables.
%while names  $\varntone$, $ke_1$, $km_1$, $\nrtwo$, $\krtwo$ are constrained to be atomic names by the history.
\begin{figure*}
\small
\begin{gather*}
\begin{prooftree}
\justifies
\Gamma, R \colon M \vdash R \colon M
\using
\textit{axiom}
\end{prooftree}
\qquad\qquad\qquad
\begin{prooftree}
\Gamma \vdash R_1 \colon K_1
\qquad\hdots\quad
\Gamma \vdash R_n \colon K_n
\justifies
\Gamma \vdash f(R_1, \hdots R_n) \colon f(K_1, \hdots K_n)
\using
\textit{intro} \end{prooftree}
\\
\qquad\qquad\qquad\qquad\quad\qquad\qquad\qquad\quad\mbox{where $f \in \left\{\pair{ \cdot}{\cdot}, \enc{\cdot}{\cdot}, \mac{\cdot, \cdot}, \dec{\cdot}{\cdot}, \fst{\cdot},\snd{\cdot}, \right\}$}
\\[5pt]
\begin{prooftree}
\Gamma, \fst{R} \colon M, \snd{Rr} \colon N \vdash S \colon K
\justifies
\Gamma, R \colon \pair{M}{N} \vdash S \colon K
\using
\textit{pair-elim}
\end{prooftree}
\quad\quad
\begin{prooftree}
\Gamma \vdash T \colon K
\quad
\Gamma, \dec{R}{T} \colon M \vdash S \colon K
\justifies
\Gamma, R \colon \enc{M}{K} \vdash S \colon N
\using
\textit{enc-elim}
\quad\qquad
\end{prooftree}
\end{gather*}
\caption{Deducibility constraints, in sequent calculus style, annotated with messages representing recipes for producing messages to the left of each colon.
}%
\label{fig:constraints}
\end{figure*}

For this system of constraints,
the only possibility is to apply the axiom in Fig.~\ref{fig:constraints}.
This is achieved by unifying the following messages (recall that $\varntone$, $ke_1$, $km_1$, $\nrtwo$, $\krtwo$ are private names hence cannot be unified with other messages):
\[
\begin{array}{l}
\pair{\enc{\pair{\nrtwo}{\pair{{\color{red}\varnt'}}{\krtwo}}}{ke_1}}{\mac{ \enc{\pair{\nrtwo}{\pair{{\color{red}\varnt'}}{\krtwo}}}{ke_1} , km_1}}
%\\
%\qquad\qquad\qquad\qquad\qquad\qquad\qquad\qquad
=
\mac{\enc{\pair{\color{red}y_1}{\pair{\varntone}{\color{red}y_2}}}{ke_1},\, km_1}
\end{array}
\]
We now use the most general unifier for the above problem,
$\sigma' = \sub{\color{red} \varnt',\,y_1,\,y_2}{\varntone,\,\nrtwo,\,\krtwo}$
thereby allowing the axiom in Fig.~\ref{fig:constraints} to be applied to both deducibility constraints generated above (firstly ignoring the recipes on the left of each colon).
Now, taking into account the recipe on the left of each colon, each of the deducibility constraints is an axiom only if we have $\varv = {\color{red}\varnt}$ and $\varw = {\color{red}\vary}$, which leads us to the substitution $\sigma = \sub{\color{red} \varnt,\,\vary}{\varv,\,\varw}$. Notice the domain of this substitution must be $\left\{ {\color{red}\varnt}, {\color{red}\vary} \right\}$, since $\varv$ and $\varw$ are treated as names rather then free variables (this is enforced by the constraints on output variables in the notion of a respectful substitution).

Thereby, from deducibility constraints (\ref{ded1}) and (\ref{ded2}) where $\sigma$ is applied to the left of each colon and $\sigma'$ is applied to the right of each colon, we obtain the following two proofs. Each proof consists of a single axiom, where a proof is a proof tree where all leaves are axioms (hence the set of premises are empty and hence vacuously in solved form).
\[
\begin{prooftree}
\justifies
\varv \colon \varntone \vdash \varv \colon \varntone
%\using\mbox{axiom}
\end{prooftree}
%\]
%\[\small
\qquad
\begin{prooftree}
\justifies
\begin{array}{l}
\varv \colon \varntone,\
\varw \colon
\pair{\enc{\pair{\nrtwo}{\pair{{\varntone}}{\krtwo}}}{ke_1}}{\mac{ \enc{\pair{\nrtwo}{\pair{{\varntone}}{\krtwo}}}{ke_1} , km_1}}
\\
\qquad
%\qquad\qquad\qquad\qquad\qquad
\qquad
 \vdash \varw \colon {
\pair{\enc{\pair{\nrtwo}{\pair{{\varntone}}{\krtwo}}}{ke_1}}{\mac{ \enc{\pair{\nrtwo}{\pair{{\varntone}}{\krtwo}}}{ke_1} , km_1}}}
\end{array}
%\using\mbox{axiom}
\end{prooftree}
\]
%Since we were forced to set $\varnt = \varv$ and $\vary = \varw$, we obtain our
Thereby we have calculated the most general respectful substitution $\sub{\color{red} \varnt,\,\vary}{\varv,\,\varw}$, enabling the following transition.
\[
h_4\sub{\color{red} \varnt,\,\vary}{\varv,\,\varw} \colon
\SystemUKfour\sub{\color{red} \varnt,\,\vary}{\varv,\,\varw}
\lts{\co{c}({\color{blue}z})} \SystemUKfive
\]
where the frame of $\SystemUKfive$ (ignoring the process) is as follows.
\[
\begin{array}{l}
      \nu ke_1, km_1, \varntone, \nrtwo, \krtwo, kt_1. \Big( \\
      \qquad\begin{array}[t]{l}
      {\sub{{\varuone,\,\varutwo}}{\getchallenge,\,\getchallenge}} \cpar
	  {\sub{{\varv}}{{\varntone}}} \cpar \\
      {\sub{{ \varw}}{{ \pair{\enc{\pair{\nrtwo}{\pair{{\varntone}}{\krtwo}}}{ke_1}}{\mac{   \enc{\pair{\nrtwo}{\pair{{\varntone}}{\krtwo}}}{ke_1} , km_1}}} }} \cpar \\
      {\sub{{\color{blue} \varz}}{{ \pair{\enc{\pair{\varntone}{\pair{{\nrtwo}}{kt_1}}}{ke_1}}{\mac{   \enc{\pair{\varntone}{\pair{{\nrtwo}}{kt_1}}}{ke_1} , km_1}}}}} \cpar \ldots \Big)
      \end{array}
\end{array}
\]

Observe that the specification, $\SpecUKfour\sub{\color{red} \varnt, \vary}{\varv, \varw}$ can also perform an output, either starting a new session,
or triggering an error message.
In either case, we reach a state that is distinguishable by static equivalence witnessed by the test ${\color{blue} \varz} = \error$ or ${\color{blue} \varz} = \getchallenge$ respectively.

Notice that in the specification, the \texttt{else} branch in which an error message is output
is enabled by the \textsc{oElse} rule  in Fig.~\ref{fig:late}.
That rule is enabled only when the following inequality holds,
where $h'_4$ is the current history extended with the private names by using the $\textsc{oRes}$ rule as follows, i.e., $h'_4 = h_4\sub{\color{red} \varnt, \vary}{\varv, \varw} \cdot ke_1^o \cdot km_1^o \cdot ke_2^o \cdot km_2^o \cdot \varntone^o \cdot \nrtwo^o \cdot \krtwo^o\cdot kt_1^o$ and $\rho_4$ is the active substitution of $\SpecUKfour\sub{\color{red} \varnt, \vary}{\varv, \varw}$ at this point.
\[
h'_4,
\rho_4
\vDash
{ \mac{\enc{\pair{\nrtwo}{\pair{\varntone}{\krtwo}}}{ke_2}, {km_2} }
                       \not= \mac{ \enc{\pair{\nrtwo}{\pair{\varntone}{\krtwo}}}{ke_2}, {km_1} } }
\]

%\error^i \cdot \getchallenge^i \cdot c^i\cdot d^i \cdot \varuone^o  \cdot \varutwo^o \cdot {\getchallenge}^i  \cdot {\varv}^o \cdot { \varv}^i
%\]
%\noindent
The above inequality is satisfied, since there is no substitution respecting $h'_4$
equating the two terms in the above inequality, i.e., it holds even under intuitionistic assumptions.
% (under the active substitution).
To see why, observe that any unifier for the above message equates ${km_1}$ and ${km_2}$, which must be kept distinct by any substitution respecting the above history, since both $km_1^o$ and $km_2^o$ appear in the history.

\subsection{Constructing a distinguishing formula from the distinguishing strategy.}\label{sec:dist}~\\
Firstly, we briefly summarise the distinguishing strategy calculated in the previous subsections.
\begin{enumerate}
\item $\AltSys$ leads with transitions labelled $\co{c}(\varuone)$ then $\co{c}(\varutwo)$ and then $d(x)$ (thereby reaching $\SystemUKone$ in which sessions have started with two readers and an ePassport, all using the same keys).
%\item Substitution $\sub{x}{\getchallenge}$ is applied.
\item $\SystemUKone\sub{x}{\getchallenge}$ leads with transition labelled $\co{c}(\varv)$. If $\SpecUKone\sub{x}{\getchallenge}$ follows with $\varv = \getchallenge$, we are done, otherwise continue.
\item $\SpecUKtwo$ leads with transition labelled $\co{d}(\varnt)$ starting up the wrong reader. If $\SystemUKtwo$ follows by inputting the wrong message into a new ePassport session this can be picked up by performing one more action, otherwise continue.
\item $\SystemUKthree$ leads with transitions labelled $\co{c}(\varw)$ and then $d(y)$. If $\SpecUKthree$ follows with $\varw = \getchallenge$ we are done, otherwise continue.
%\item Substitution $\sub{\varnt,\vary}{\varv,\varw}$ is applied.
\item $\SystemUKfour\sub{\varnt,\vary}{\varv,\varw}$ leads with transitions $\co{c}(\varz)$.
This can only be followed by a transition from $\SpecUKfour\sub{\varnt,\vary}{\varv,\varw}$ reaching a state where $\varz = \error$ or $\varz = \getchallenge$.
\end{enumerate}
The problem now is that open bisimilarity (Def.~\ref{def:open}) does not satisfy any notion of completeness,
hence a distinguishing strategy may be a spurious counterexample.
Spurious counterexamples, cannot be transformed into counterexamples for strong early bisimilarity (Def.~\ref{def:strong}) and are less likely to indicate the presence of an attack.

The above strategy does not describe a spurious counterexample; and furthermore it can be turned into a real attack.
In order to show that it is not a spurious counterexample,
our methodology is to construct a modal logic formula from the distinguishing strategy.
Instead of using a modal logic characterising open bisimilarity (which would be a generalisation of intuitionistic $\OM$~\cite{Ahn2017} to the applied $\pi$-calculus, making used of the notion of reachability in Def.~\ref{def:reach}) we employ a modal logic characterising
strong early bisimilarity called classical $\FM$.
%Defined in Fig.~\ref{fig:FM}

\subsubsection{Introducing classical $\FM$}

The syntax of modal logic
\textit{classical} $\FM$ ($\mathcal{F}$ is for free inputs, $\mathcal{M}$ is for match~\cite{Milner1993})
%A syntax for \textit{classical} $\FM$
is presented below.
\begin{gather*}
%\begin{array}{l}
%\left.
\begin{array}{rlr}
\phi \Coloneqq&  M = N & \mbox{equality} \\
%\ttt & \mbox{top} \\
%          \mid& \fff & \mbox{bottom} \\
%          \mid& \diam{M = N}\phi & \mbox{diamond match} \\
%          \mid& \diam{M \not= N}\phi & \mbox{diamond mismatch}\\
%          \mid& \boxm{M = N}\phi & \mbox{box match}\\
%          \mid& \boxm{M \not= N}\phi & \mbox{box mismatch}\\
          \mid& \phi \wedge \phi & \mbox{conjunction} \\
          \mid& \diam{\pi}\phi & \hspace{21pt} \mbox{diamond} \\
%          \mid& \phi \vee \phi & \mbox{disjunction} \\
%          \mid& \phi \yields \phi & \mbox{implication} \\
          \mid& \neg\phi & \mbox{negation}
\end{array}
%\right\}
%\mbox{classical logic}
%\\
%\left.
%\begin{array}{rlr}
%\hspace{17pt}
%\\
%          \mid& \boxm{\pi}\phi & \mbox{box}
%\end{array}
%\right\}
%\mbox{weak modalities}
%\\\\
\qquad\qquad
\begin{array}{rl}
\mbox{abbreviations:}&
\ttt \triangleq M = M
\\
&M \not= N \triangleq \neg (M = N)%
\\
&\boxm{\pi}\phi \triangleq \neg\diam{\pi}\neg\phi
\\
&\phi \vee \psi \triangleq \neg\left(\neg\phi \wedge \neg\psi\right)
%\\
%\ttt \triangleq M = M
%\\
%\fff \triangleq \neg\ttt
%\\
%\boxm{M = N}\phi \triangleq M = N \yields \phi
%\\
%\diam{M = N}\phi \triangleq M = N \wedge \phi
%\\
%\boxm{M \not= N}\phi \triangleq M \not= N \yields \phi
%\\
%\diam{M \not= N}\phi \triangleq M \not= N \wedge \phi
%\end{array}
\end{array}
\end{gather*}
The semantics of classical $\FM$ is given by the least relation $A \vDash \phi$ between extended processes $A$ and formulae $\phi$ satisfying the conditions in Fig.~\ref{figure:FM}.
%in Fig.~\ref{figure:FM}.
\begin{figure}[h]
\begin{gather*}
\begin{array}{lcl}
%A \vDash \ttt &&\mbox{always holds.} \\
\mathopen{\nu \vec{x}.}\left( \theta \cpar P \right) \vDash M = N
&\mbox{iff}&
%\begin{array}[t]{l}
M\theta \mathrel{=_E} N\theta  ~\mbox{ and }~
%\\
%\qquad
\mbox{$\vec{x}$  are fresh for $M$ and $N$}
%\cap\left(\fv{M}\cup\fv{N}\right) = \emptyset
%\end{array}
%\\
%A \vDash \phi_1 \lor \phi_1 &\mbox{iff}&
%  A \vDash \phi_1 ~\mbox{or}~ A \vDash \phi_2.
%\\
%A \vDash \phi_1 \yields \phi_2 &\mbox{iff}&
%\mbox{whenever } \reachable{A}{\sigma,\vec{x}.\rho}{A'},  ~\mbox{we have}~
%A' \vDash \phi_1\sigma ~\mbox{implies}~ A' \vDash \phi_2\sigma.
\\
A \vDash \diam{\pi}\phi &\mbox{iff}&
%\begin{array}[t]{l}
 \mbox{there exists }B ~\mbox{ such that }~
%\\
%\qquad
 A \lts{\pi} B ~\mbox{ and }~ B \vDash \phi.
%\end{array}
\\
A \vDash \phi_1 \land \phi_2 &\mbox{iff}&
  A \vDash \phi_1 ~\mbox{ and }~ A \vDash \phi_2.
\\
A \vDash \neg\phi &\mbox{iff}&
  A \vDash \phi ~\mbox{ does not hold.}
%B \mbox{ such that } A \Lts{\pi} B ~\mbox{and}~ B \vDash \phi.
%\\
%A \vDash \diam{\co{M}(z)}\phi &\mbox{iff}&
%  \mbox{there exists } B \mbox{ such that } A \lts{\co{M}(z)} B ~\mbox{and}~ B \vDash \phi.
%\\
%A \vDash \boxm{\pi}\phi &\mbox{iff}&
%\mbox{whenever } \reachable{A}{\sigma,\vec{x}.\rho}{A'}, \mbox{ and }
%    A' \lts{\pi\sigma} B,  ~\mbox{we have}~
%    B \vDash \phi\sigma.
%\\
%A \vDash \boxm{\co{M}(z)}\phi &\mbox{iff}&
%  \mbox{whenever } \reachable{A}{\sigma}{A'} \mbox{ and }
%\mbox{for any } B, \vec{x}, \mbox{and fresh } \vec{z} \mbox{ and any } \sigma \respecting \dom{\theta},
%\\&&\hfill
%\left(\nu \vec{x},\vec{y}.\left(\sigma\cdot\sub{\vec{z}}{\vec{x}} \cpar P\right)\right)\sigma
%A' \lts{\co{M\sigma}(z)} B, ~\mbox{we have}~ B \vDash \phi\sigma.
\end{array}
\end{gather*}
%\vspace*{-.75em}
\caption{The semantics of modal logic ``classical $\FM$''.
}\label{figure:FM}
\end{figure}

The following theorem formulates what it means for classical $\FM$ to characterise strong early bisimilarity.
\begin{thm}\label{thm:char}
$P \sim Q$, whenever, for all formula $\phi$, we have $P \vDash \phi$ if and only if $Q \vDash \phi$.
\end{thm}
The proof is provided in Appendix~\ref{sec:app2}.
From the contrapositive of the above theorem,
whenever $P \not\sim Q$,
there exists a formula $\phi$ such that $P \vDash \phi$ holds, but $Q \nvDash \phi$.
Such a formula is called a distinguishing formula.

\subsubsection{The attack on BAC as a formula.}%
\label{sec:formula}

We are now in a position to prove Theorem~\ref{thm:BAC}, restated below for convenience,
which establishes that strong unlinkability of the BAC protocol fails.
\begin{thm}[Theorem~\ref{thm:BAC} restated]\label{thm:BACagain}
$\Sys \not\approx \Spec$.
\end{thm}
\begin{proof}
In order to establish the failure of strong unlinkability of the BAC protocol, we make use of the following classical $\FM$ formula $\psi$.
\[
\psi \triangleq
\begin{array}[t]{l}
\diam{\co{c}(\varuone)}
\diam{\co{c}(\varutwo)}
\diam{d\,\getchallenge}
\diam{\co{c}(\varv)}
\big(
\\
\qquad
\begin{array}[t]{l}
\varv \not= \getchallenge \wedge {}
%\\
% y = \getchallenge \wedge z \not= \getchallenge\ \wedge
\\
\boxm{d\, \varv}\big(
\begin{array}[t]{l}
\diam{\co{c}(\varw)}
\diam{d\,\varw}
\diam{\co{c}(\varz)}\big(
\varw \neq \getchallenge \wedge
%u \neq \error \wedge \\
\varz \neq \getchallenge \wedge \varz \neq \error
\big)
%\\
%\vee
%\diam{\co{c}(w)}\left( w = \error \right)
\\
\vee\
\boxm{\co{c}(\varw)}\left( \varw = \getchallenge \right)
\big)
\big)
\end{array}
\end{array}
\end{array}
\]
For this formula we can verify $\AltSys \vDash \psi$ holds; while $\AltSpec \nvDash \psi$.
Hence, by Theorem~\ref{thm:char}, $\AltSys \not\sim \AltSpec$; thereby by Theorem~\ref{thm:strong},
$\Sys \not\approx \Spec$, as required.
%stated in Theorem~\ref{thm:BAC}.
\end{proof}

This closes the initial question of whether or not unlinkability holds for the BAC protocol; the answer is that the BAC protocol does not satisfy unlinkability, at least as specified originally in CSF'10~\cite{Arapinis2010}.
This leads to several immediate questions.
Firstly, how do we construct the above formula from the distinguishing strategy given by open bisimilarity (to be addressed in Sec.~\ref{sec:construct})?
Secondly, can we explain why the formula is distinguishing, and from that explanation describe a practical attack?
Thirdly, how do we approach the problem of constructing a formula in general and how do we handle cases when a spurious counterexample is discovered?
We focus mainly on the first question in this paper. The second question we we return to in the next section; while the third question is worthy of future work, since it would enable tool support.
%\RH{The second two questions are not yet answered in this paper!}

We emphasise at this point that there are infinitely many alternative distinguishing formulae for this problem, not only $\psi$.
Some such distinguishing formulae use no box modality and instead employ conjunction, where conjunction also appeals to the branching time nature of bisimilarity.
We postpone discussing alternatives until Sec.~\ref{sec:II}, where we propose an alternative model of unlinkability, where the meaning of distinguishing strategies is clearer.

\subsubsection{How to construct the distinguishing formula}%
\label{sec:construct}

We construct the formula named $\phi$, used as the distinguishing formula in the proof of Theorem~\ref{thm:BAC} (reiterated as Theorem~\ref{thm:BACagain}), by induction on the depth of the distinguishing strategy summarised at the top of Sec.~\ref{sec:dist}.
To do so, we work backwards through the distinguishing strategy. Note we refer to processes representing intermediate states of an execution previously defined throughout Sec.~\ref{sec:symbolic}.

% summarised at the top of Sec.~\ref{sec:dist}.
Firstly, observe that when the system is in state $\SystemUKfive$
the specification must be in a state where either $\error = \varz$, $\getchallenge = \varz$ or $\getchallenge = \varw$, where each pair of messages shows static equivalence is violated.
Since for $\SystemUKfive$ both $\varw$ and $\varz$ cannot be unified with $\getchallenge$ or $\error$ with respect to the history at that point, the intuitionistic negation and classical negation coincide for these equalities, leading to the following formula distinguishing $\SystemUKfive$ from any state the specification can reach at that point.
\[
\SystemUKfive \vDash \error \not= \varz \wedge \getchallenge \not= \varz \wedge  \getchallenge \neq \varw
\]
Since the system leads in order to reach this state using transition $\SystemUKfour\sub{\varnt,\vary}{\varv,\varw} \lts{\co{c}(\varz)} \SystemUKfive$, we add a diamond modality to the formula, as follows.
\[
\SystemUKfour\sub{\varnt,\vary}{\varv,\varw} \vDash
\diam{\co{c}(\varz)}\left( \error \not= \varz \wedge \getchallenge \not= \varz \wedge  \getchallenge \neq \varw \right)
\]
The above step is standard for constructing modal logic formulae for distinguishing strategies; however the next step requires care.
Firstly, note that the substitution $\sub{\varnt,\vary}{\varv,\varw}$ concerns input variables.
Thus we push these substitutions back through the distinguishing strategy until the relevant input is instantiated.
At this point, since the input action in the distinguishing strategy reaching state $\SystemUKfour\sub{\varnt,\vary}{\varv,\varw}$ introduced variable $y$, that variable is instantiated immediately, and $\varnt$ is pushed back through the strategy, refining the distinguishing strategy to obtain the following transitions using the rules of Fig.~\ref{figure:active}.
\[
\SystemUKthree\sub{\varnt}{\varv}
\lts{\co{c}(\varw)}
\lts{d\,\varw}
\SystemUKfour\sub{\varnt,\vary}{\varv,\varw}
\]
Since the system was leading, $\SystemUKthree\sub{\varnt}{\varv}$ can be
distinguished from $\SpecUKthree\sub{\varnt}{\varv}$ by the following formula.
\[
\SystemUKthree\sub{\varnt}{\varv}
\vDash
\diam{\co{c}(\varw)}
\diam{d\,\varw}
\diam{\co{c}(\varz)}\left( \error \not= \varz \wedge \getchallenge \not= \varz \wedge  \getchallenge \neq \varw \right)
\]
%\SM{I don't understand the next sentence.}\RH{Is it clearer now?}
The next step in the distinguishing strategy involves a change of leading player, where the specification leads with an action.
By pushing back the substitution through the strategy, instantiating the input variable on the label,
we have the following transition led by the specification:
$\SpecUKtwo
\lts{ d\,\varv }
\SpecUKthree\sub{\varnt}{\varv}$.
Since the specification leads at this point and the system follows in any way it can,  we write the box modality in the distinguishing formula for the system followed by a disjunction of formulae, where each formula distinguishes $\SpecUKthree\sub{\varnt}{\varv}$ from any state reachable from
$\SystemUKtwo$
by an input transition labelled with  $d\,\varv$.

\begin{figure}
\small
\xymatrix{
(\dagger)~\AltSys
\ar^{\co{c}(\varuone)}[d]
 && \AltSpec \ar^{\co{c}(\varuone)}[d] &
\\
\ar^{\co{c}(\varutwo)}[d]
 &&  \ar^{\co{c}(\varutwo)}[d] &
\\
\ar^{d\,\getchallenge}[d]
 &&  \ar^{d\,\getchallenge}[d] &
\\
\SystemUKone
\ar^{\co{c}(\varv)}[d]
 &&
\SpecUKone
\ar^{\co{c}(\varv)}[d]
\ar^{\co{c}(\varv)}[dr]
 &
\\
\SystemUKtwo
\ar^{d\,\varv}[d]
\ar^{d\,\varv}[dr]
 &&
(\dagger)~\SpecUKtwo
\ar^{d\,\varv}[d] &
%\AltSpec^{\mathrm{alt.II}} \vDash
\varv = \getchallenge
\\
(\dagger)~\SystemUKthree\sub{\varnt}{\varv}
\ar^{\co{c}(\varw)}[d]
 &
\ar^{\co{c}(\varw)}[d]
&
\SpecUKthree\sub{\varnt}{\varv}
\ar^{\co{c}(\varw)}[d]
\ar^{\co{c}(\varw)}[dr]
 &
\\
\ar^{d\,\varw}[d]
 &
%\AltSys^{alt.\mathrm{III}} \vDash
\varw = \getchallenge
&
%\SpecUKthree\sub{\varnt}{\varv}
\ar^{d\,\varw}[d] &
\varw = \getchallenge
\\
\SystemUKfour\sub{\varnt, \vary}{\varv, \varw}
\ar^{\co{c}(\varz)}[d]
 &
%\AltSys^{alt.\mathrm{III}} \vDash
&
\SpecUKfour\sub{\varnt, \vary}{\varv, \varw}
\ar^{\co{c}(\varz)}[d]
\ar^{\co{c}(\varz)}[dr]
 &
\\
\varz \not= \error
\wedge
\varz \not= \getchallenge
& &
\varz = \error
&
\varz = \getchallenge
}

\caption{Distinguishing strategy implied by distinguishing formula $\psi$.
}%
\label{fig:game}
\end{figure}

Observe that, as well as $\SystemUKthree\sub{\varnt}{\varv}$ for which we constructed a distinguishing formula above, there is another, quite distinct, process reachable from $\SystemUKtwo$ by a $d\,\varv$ transition that can be distinguished by formula $\boxm{\co{c}(\varw)}(\varw = \getchallenge)$ representing that the case when input transition labelled with $d\,\varv$ results in feeding $\varv$ into a new reader session, which kills the possibility of continuing an existing session with an output transition. Note this formula is also constructed algorithmically, as we are describing, but this branch is more due to a limitation of the original model unlinkability communicated in CSF'10, that we address next in Sec.~\ref{sec:II}. Hence we draw no further attention to that branch at this point.

Thereby, we obtain the following formula distinguishing $\SystemUKtwo$ from $\SpecUKtwo$.
\[
\SystemUKtwo
\vDash
\boxm{d\, \varv}\big(
\begin{array}[t]{l}
\diam{\co{c}(\varw)}
\diam{d\,\varw}
\diam{\co{c}(\varz)}\big(
\varw \neq \getchallenge \wedge
%u \neq \error \wedge \\
\varz \neq \getchallenge \wedge \varz \neq \error
\big)
%\\
%\vee
%\diam{\co{c}(w)}\left( w = \error \right)
\\
\vee\
\boxm{\co{c}(\varw)}\left( \varw = \getchallenge \right)
\big)
\end{array}
\]
The rest of the construction of formula $\psi$ follows the pattern of steps already described above, where a diamond modality is appended whenever the system leads and substitutions are pushed back through the distinguishing strategy until they instantiate the relevant input, or reach the root of the formula.

We present an informal graphical depiction of the game that $\psi$ describes in Fig.~\ref{fig:game}.
In the figure, annotation $(\dagger)$ indicates where a process takes over as the leading process in the strategy. When a process is not leading it may have the option to try more than one counter move, represented by the branches in the strategy.

%% file: improved.tex
\newcommand{\NewSys}{\textit{NewSys}}
\newcommand{\NewSpec}{\textit{NewSpec}}
\newcommand{\SysBAC}{\textit{System}_\textit{BAC}}
\newcommand{\SpecBAC}{\textit{Spec}_\textit{BAC}}
\newcommand{\SysPACE}{\textit{System}_\textit{PACE}}
\newcommand{\SpecPACE}{\textit{Spec}_\textit{PACE}}
\newcommand{\ePassport}{\textit{Prover}}
\newcommand{\eReader}{\textit{Verifier}}
\newcommand{\PBAC}{P_{\textit{BAC}}}
\newcommand{\VBAC}{V_{\textit{BAC}}}
\newcommand{\pass}{\textit{passport}}
\newcommand{\reader}{\textit{reader}}

\section{A New Chapter for Unlinkability: Refining the Model of Unlinkability}\label{sec:II}

While the previous sections closed a chapter in the story of unlinkability,
by proving that there is an attack on the model of unlinkability of the BAC protocol as originally communicated in the proceedings of CSF'10~\cite{Arapinis2010}, this section opens a new chapter by justifying a new model of unlinkability.
This new model of unlinkability is a modest improvement on the model previously proposed. It incorporates some explicit observables reflecting the ability of the attacker to observe and hence control the creation of radio frequency communication channels.
We demonstrate here why our proposed model of unlinkability more accurately models the distinguishing power of an attacker; and how descriptions of attacks on the BAC protocol, given by modal logic formulae, become clearer.
This section can also been seen as introducing preliminaries required for Section~\ref{sec:pace}, where we show how the model proposed discovers new attacks on the PACE protocol that follow a similar pattern to the attacks discovered on the BAC protocol.

The model of strong unlinkability originally proposed in CSF'10~\cite{Arapinis2010} has many merits.
However, a limitation we would like to draw attention to is that it matters whether or not we include the $\getchallenge$ message that the reader sends to initiate the protocol.
That $\getchallenge$ message happens to be essential for the attacks on unlinkability described in the previous section since, by controlling the $\getchallenge$ messages sent and received, we can count the number of sessions that are present and thereby infer when a message sent is a ciphertext in an existing session of the protocol rather than a fresh nonce at the beginning of a new protocol.
%We cannot use the nonce for this purpose, since in that model it is indistinguishable from a ciphertext later in the protocol or sent by another party.

This is perhaps clearest in Fig.~\ref{fig:game} of the previous section.
Observe that in order to reach the bottommost state in the figure, we ensure that no additional sessions are started beyond the two reader sessions and one ePassport session at the beginning triggered by the topmost three actions in the figure.
This causes problems, three of which are highlighted below, which are all due to the modelling decision where all parties use the same channel $c$ for outputs and $d$ for inputs.
\begin{itemize}
\item Limitation 1. It is inconvenient and confusing to, throughout the strategy, add branches that have the effect of saying ``at this point we don't start a new session.''
\item Limitation 2. In the reality, the attacker can directly observe whether or not two inputs or outputs are performed within the same session of a protocol and, furthermore, can distinguish between a session with a reader or with an ePassport.
This is not only because the attacker must be aware of the physical location of each entity, but also because, for each session, the attacker must open a new channel using the underlying transport protocol, as standardised in ISO/IEC 14443~\cite{ISO14443}.
\item Limitation 3. Finally, the fact that the $\getchallenge$ message is useful for counting the number of each type of session initiated, is rather a misuse of that message.
Message $\getchallenge$ contributes nothing to this authentication protocol (other than impeding the stronger authentication property synchronisation~\cite{Cremer2006}, which is immediately violated in protocols with a constant message).
Hence removing it from a model of BAC should not result in attacks ceasing to exist.
%hanve we should be able to control the number of sessions that are present, even without
%remove the parts of the protocol $\getchallenge$ from the model, which  it without trouble.
\end{itemize}

\noindent
The above limitations of existing models used to analyse the unlinkability of the BAC protocol, as employed in previous sections, can be addressed simply by declaring a fresh public channel for each session.
To do so, we extend the model with two channels, say $\pass$ and $\reader$, that are used to model the creation of a new channel in the respective roles of either an ePassport or a reader. In the applied $\pi$-calculus, we achieve this by sending a fresh channel to the environment on these channels for each session, as in the following new scheme for the system and specification.
\[
\begin{array}{ll}
\mbox{Scheme for System:}
& \bang\mathopen{\nu \vec{k}.} \bang\left(\mathopen{\nu c.}\cout{\pass}{c}.\ePassport(c,\vec{k}) \cpar \mathopen{\nu c.}\cout{\reader}{c}.\eReader(c,\vec{k})\right)
\\
\mbox{Scheme for Specification:}
& \bang\mathopen{\nu \vec{k}.} \left(\mathopen{\nu c.}\cout{\pass}{c}.\ePassport(c,\vec{k}) \cpar \mathopen{\nu c.}\cout{\reader}{c}.\eReader(c,\vec{k})\right)
\end{array}
\]
The above we propose as a general scheme for RFID protocols employing symmetric keys $\vec{k}$, which are shared through another channel that the attacker cannot intercept. Recall, in the case of ePassport protocols, this is usually achieved by an OCR session; but may be achieved by other means such as sending the key to the reader via a secure connection between a personal device and the reader. Thus a fundamental assumption in all these models of unlinkability is that we are considering use cases where intercepting and manipulating RFID communication is easier than intercepting the keys, which would trivially break unlinkability.

Notice we directly employ a presentation of processes that does not involve $\tau$-transitions.
This simplifies the problem such that strong notions of bisimilarity may be employed, without loss of modelling power.
Of course, to do so, we should assume each instance of $\ePassport$ and $\eReader$ is a sequential process (or apply another suitable restriction for forbidding $\tau$-transitions internal to a single reader or ePassport).

\subsection{The unlinkability of the BAC protocol, simplified}
Following the above scheme for the BAC protocol, $\vec{k}$ is $ke,km$,
and
$\ePassport(c,ke,km)$ and $\eReader(c,ke,km)$
are instantiated with the processes
$\PBAC(c,ke,km)$ and $\VBAC(c,ke,km)$
defined below.
\[
\begin{array}{rl}
\PBAC(c,ke,km) \triangleq&
\begin{array}[t]{l}
 \nu nt.\cout{c}{nt}.c(y). \\
     %& \clet{m_e,m_m}{\fst{y},\snd{y}} \\
 \texttt{if}\,\snd{y} = \mac{\fst{y}, km}\,\texttt{then} \\
\qquad
 \begin{array}[t]{l}
 \texttt{if}\,nt = \fst{\snd{\dec{\fst{y}}{ke}}}\,\texttt{then} \\
%      & \begin{array}[t]{ll}
%        \texttt{then}
\qquad
 \begin{array}[t]{l}
                         \nu kt.\clet{m}{\enc{\pair{nt}{\pair{\fst{\dec{\fst{y}}{ke}}}{kt}}}{ke}} \\
                         \cout{c}{m, \mac{m, km}}
                         \end{array}  \\
 \texttt{else}\,\cout{c}{error}
\end{array}
\\
 \texttt{else}\,\cout{c}{error}
        \end{array}
\\
\VBAC(c,ke,km)
\triangleq&
\begin{array}[t]{l}
 c(nt).\nu nr.\nu kr.
\\
 \clet{m}{\enc{\pair{nr}{\pair{nt}{kr}}}{ke}}
 \cout{c}{m, \mac{\pair{m}{km}}}
\end{array}
\end{array}
\]
Notice the above processes are simply $\Pass(ke,km)\sub{d}{c}$ and $\R(ke,km)\sub{d}{c}$ from the previous sections, but
with prefixes concerning the $\getchallenge$ message removed.

In summary, we propose that the problem of whether there is an attack on the unlinkability of the BAC protocol
can be resolved by proving that the following theorem holds.
\begin{thm}\label{thm:BAC2}
$\SysBAC \not\sim \SpecBAC$, where
\[
\begin{array}{rl}
\SysBAC \triangleq
& \bang\mathopen{\nu ke, km.} \bang\left(
    \mathopen{\nu c.}\cout{\pass}{c}.\PBAC(c, ke, km)   \cpar
    \mathopen{\nu c.}\cout{\reader}{c}.\VBAC(c, ke, km) \right)
\\
\SpecBAC \triangleq
& \bang\mathopen{\nu ke, km.} \left(
    \mathopen{\nu c.}\cout{\pass}{c}.\PBAC(c, ke, km)   \cpar
    \mathopen{\nu c.}\cout{\reader}{c}.\VBAC(c, ke, km) \right)
\end{array}
\]
\end{thm}
\begin{proof}
Consider the $\FM$ formula below.
\[
\varphi \triangleq
\begin{array}[t]{l}
\diam{\co{\reader}(c_1)}
\diam{\co{\reader}(c_2)}
\diam{\co{\pass}(c_3)}
\diam{\co{c_3}(nt)}
\Big(
\\
\qquad
\begin{array}[t]{rl}&
\diam{c_1\, nt}
\diam{\co{c_1}(\varw)}
\diam{c_3\,\varw}
\diam{\co{c_3}(\varz)}
\left( \varz \neq \error \right)
\\
\wedge &
\diam{c_2\, nt}
\diam{\co{c_2}(\varw)}
\diam{c_3\,\varw}
\diam{\co{c_3}(\varz)}
\left( \varz \neq \error \right)
~\Big)
\end{array}
\end{array}
\]
Since $\SysBAC \vDash \varphi$,
but $\SpecBAC \nvDash \varphi$,
 by Theorem~\ref{thm:char}, $\SysBAC \not\sim \SpecBAC$.
\end{proof}

\begin{figure}
\small
\xymatrix{
&
(\dagger)~\SysBAC
\ar^{\co{\reader}(c_1)}[d]
 &&
&
\SpecBAC
\ar^{\co{\reader}(c_1)}[d]
%\ar^{\co{\reader}(c_1)}[drr]
&
\\
&
\ar^{\co{\reader}(c_2)}[d]
& &&  \ar^{\co{\reader}(c_2)}[d] &
%&  \ar^{\co{\reader}(c_2)}[d] &
\\
&
\ar^{\co{\pass}(c_3)}[d]
 &&&
 \ar_{\co{\pass}(c_3)}[dl]
 \ar^{\co{\pass}(c_3)}[dr] &
%&  \ar^{\co{\pass}(c_3)}[d] &
\\
&
\ar^{\co{c_3}(nt)}[d]
 &&
\ar^{\co{c_3}(nt)}[d]
 &  &
\ar^{\co{c_3}(nt)}[d]
 &
\\
&
\ar_{c_1\,nt}[dl]
\ar^{c_2\,nt}[dr]
 &&
%(\dagger)~
\ar^{c_1\,nt}[d]
&  &
%(\dagger)~
\ar^{c_2\,nt}[d]
\\
%(\dagger)~
\ar^{\co{c_1}(\varw)}[d]
& &
%(\dagger)~
\ar^{\co{c_2}(\varw)}[d]
 &
\ar^{\co{c_1}(\varw)}[d]
 &&
\ar^{\co{c_2}(\varw)}[d]
\\
\ar^{c_3\,\varw}[d]
 &&
\ar^{c_3\,\varw}[d]
 &
\ar^{c_3\,\varw}[d]
 &&
\ar^{c_3\,\varw}[d]
\\
\ar^{\co{c_1}(\varz)}[d]
 &&
\ar^{\co{c_2}(\varz)}[d]
 &
\ar^{\co{c_1}(\varz)}[d]
 &&
\ar^{\co{c_2}(\varz)}[d]
\\
\varz \not= \error
&&
\varz \not= \error
&
\varz = \error
&&
\varz = \error
}

\caption{Distinguishing strategy implied by distinguishing formula $\varphi$.
%Annotation $(\dagger)$ indicates where a process takes over as the leading process in the strategy.
}%
\label{fig:game2}
\end{figure}

Now compare the distinguishing strategy generated by $\varphi$, presented in Fig.~\ref{fig:game2},
to the distinguishing strategy for  $\psi$ in the previous section, presented in Fig.~\ref{fig:game}.
The attacks described start in a similar fashion.
In both figures, the system starts up two readers and an ePassport with the same keys.
The specification can only follow by starting two readers with different keys;
hence when the ePassport is initialised is has different keys from at least one of the readers.
We draw attention to two key differences between the strategies presented below.

The first notable difference between the strategies is that Fig.~\ref{fig:game2}
does not require several sub-branches of the strategy involving $\getchallenge$ messages.
Those branches that appear throughout Fig.~\ref{fig:game} are no longer required, since we can directly observe the number of sessions that are present, rather than implicitly controlling the number of sessions by preventing new sessions from initialising.
This difference is beneficial for cleaning up messy strategies, making them easier to explain, and allowing more protocols to be analysed without having to insert constant messages into the model of the protocol.

The second notable difference between the strategies is that, in Fig.~\ref{fig:game2},
the system always leads, including at the point where branching occurs.
By the time the system decides which branch to take, the specification has already committed to a state where the the ePassport has different keys from either the reader on channel $c_1$ or the reader on channel $c_2$.
The strategy of the system is to choose to communicate with the reader which has keys that are different to those of the ePassport.
Thereby the system wins the game since it can reach a state where no error is produced by the chosen ePassport --- a strategy that cannot be matched by the specification.
In contrast, the strategy in Fig.~\ref{fig:game} achieved a similar effect but in a different way: the attacker changes perspective by switching to a view where the leading player is the specification, i.e., what should hypothetically happen. Recall, in Fig.~\ref{fig:game}, after the change of player, the strategy is for the specification to choose to communicate with a reader that should produce an error message; but the system has no way to produce such an error message hence the system clearly is not equivalent to the hypothetical situation modelled by the specification.

%%%%%%%%%%%%%%%%%%%%%%%%%%%%%%%%%%% REVISION

\subsection{Alternative formulas for describing attacks}\label{sec:new}

The formula in the proof of Theorem~\ref{thm:BAC2} is not the only formula describing an attack on unlinkability of the BAC protocol.
Indeed, in Section~\ref{sec:formula} we noted that, when there is an attack, there are infinitely many alternative formulae.
In this section, we introduce and explain another formula that is a little longer than $\varphi$, defined in the previous section, but is useful for situating the attacks discovered.

Consider the $\FM$ formula below.
\[
\varsigma \triangleq
\begin{array}[t]{l}
\diam{\co{\reader}(c_1)}
\diam{\co{\pass}(c_3)}
\Big(
\\\quad
\begin{array}[t]{rl}&
\diam{\co{c_3}(nt)}
\diam{c_1\, nt}
\diam{\co{c_1}(\varw)}
\diam{c_3\,\varw}
\diam{\co{c_3}(\varz)}
\left( \varz \neq \error \right)
\\
\wedge &
\diam{c_1\, nt'}
\diam{\co{c_1}(w')}
\diam{\co{\reader}(c_2)}
\diam{\co{c_3}(nt)}
\diam{c_2\, nt}
\diam{\co{c_2}(\varw)}
\diam{c_3\,\varw}
\diam{\co{c_3}(\varz)}
\left( \varz \neq \error \right)
~\Big)
\end{array}
\end{array}
\]
The above formula also serves as an alternative proof certificate for Theorem~\ref{thm:BAC2}.
Observe that $\SysBAC \vDash \varsigma$,
but $\SpecBAC \nvDash \varsigma$;
and hence, by Theorem~\ref{thm:char}, $\SysBAC \not\sim \SpecBAC$.
The distinguishing strategy described by $\varsigma$ is depicted in Fig.~\ref{fig:game3}.

\begin{figure}
\small
\xymatrix{
&
(\dagger)~\SysBAC
\ar^{\co{\reader}(c_1)}[d]
 &&
&
\SpecBAC
\ar^{\co{\reader}(c_1)}[d]
%\ar^{\co{\reader}(c_1)}[drr]
&
\\
&
\ar^{\co{\pass}(c_3)}[d]
 &&&
 \ar_{\co{\pass}(c_3)}[dl]
 \ar^{\co{\pass}(c_3)}[dr] &
%&  \ar^{\co{\pass}(c_3)}[d] &
\\
&
\ar_{\co{c_3}(nt)}[dl]
\ar^{{c_1}\,nt'}[dr]
 &&
\ar^{\co{c_3}(nt)}[d]
 &  &
\ar^{{c_1}\,nt'}[d]
\\
\ar_{c_1\,nt}[d]
&
 &
\ar^{\co{c_1}(w')}[d]
&
%(\dagger)~
\ar^{c_1\,nt}[d]
&  &
%(\dagger)~
\ar^{\co{c_1}(w')}[d]
\\
%(\dagger)~
\ar^{\co{c_1}(\varw)}[d]
& &
%(\dagger)~
\ar^{\co{\reader}(c_2)}[d]
 &
\ar^{\co{c_1}(\varw)}[d]
 &&
\ar^{\co{\reader}(c_2)}[d]
\\
\ar^{c_3\,\varw}[d]
 &&
\ar^{\co{c_3}(nt)}[d]
 &
\ar^{c_3\,\varw}[d]
 &&
\ar^{\co{c_3}(nt)}[d]
\\
\ar^{\co{c_1}(\varz)}[d]
 &&
\ar^{c_2\,nt}[d]
 &
\ar^{\co{c_1}(\varz)}[d]
 &&
\ar^{c_2\,nt}[d]
\\
\varz \not= \error
&&
\ar^{\co{c_2}(\varw)}[d]
&
\varz = \error
&&
\ar^{\co{c_2}(\varw)}[d]
\\
&&
\ar^{c_3\,\varw}[d]
&&&
\ar^{c_3\,\varw}[d]
\\
&&
\ar^{\co{c_2}(\varz)}[d]
&&&
\ar^{\co{c_2}(\varz)}[d]
\\
&&
\varz \not= \error
&&&
\varz = \error
}

\caption{Another distinguishing strategy $\varsigma$, where readers are created sequentially.}%
\label{fig:game3}
\end{figure}

Both formulas $\varsigma$ and $\varphi$ describe attack strategies on the unlinkability of the BAC protocol.
The advantage of $\varphi$ is that it is more compact. Also, declaring two readers at the top of the strategy and making a choice between them makes clear the key idea: that there is a strategy for testing whether two readers are capable of authenticating the same ePassport.

Formula $\varsigma$ is presented to point out that the two readers need not be simultaneously active.
Observe that in $\varsigma$ the event representing the creation of the second reader, indicated in both strategies as $\cout{\reader}{c_2}$, is pushed later in the attack strategy compared to in $\varphi$. The distinguishing strategy then proceeds as follows.
\begin{enumerate}
\item A run of a reader and ePassport is created. What is important at this point is that the system has an opportunity to start a reader and ePassport run that match, i.e., the reader is loaded with the key of the ePassport in question.
The specification has two choices, which are the first and second respective branches taken by $\SpecBAC$ in Fig.~\ref{fig:game3}:
\begin{enumerate}
\item\label{itema} the specification can start a run with an ePassport with different keys to the reader;
\item\label{itemb} or using the same keys as the reader.
\end{enumerate}

\item
Now consider the two branches of the conjunction in the formula $\varsigma$, which occurs after one reader and ePassport are created.
\begin{enumerate}[(i)]
\item
In the first branch of the conjunction, the second reader is never used. There are only the events required for the run of the ePassport and reader, initially created on channels $c_3$ and $c_1$ respectively, to authenticate.
This branch of the conjunction can be played by the attacker whenever the specification takes its first branch $(\ref{itema})$, where the run of an ePassport involves keys different from those loaded into the reader initially created.

\item
In the second branch of the conjunction, we do make use of a second reader.
To emphasise that no actions of readers need be concurrent for this attack strategy, we first consume the actions of the reader on channel $c_1$ by using the dummy nonce $nt'$ and effectively ignoring the response $w'$ from that reader.
At that point, the second reader is created on channel $c_2$ such that it is again loaded with the same keys.
This allows the second reader to successfully authenticate the ePassport created at the beginning of the attack.
This branch of the conjunction is played in response to the specification taking its second branch $(\ref{itemb})$, where the run of the ePassport and first reader on $c_1$ match, and hence, since no further run may use the same keys, the second reader must fail to authenticate in that idealised setting.
\end{enumerate}
\end{enumerate}

\noindent
The first thing to observe about the attack described by $\varsigma$ is that, since no two reader runs are concurrently active, we know that designing a system to force reader runs involving the same ePassport to be conducted sequentially will not prevent attacks on unlinkability.
We raise this point, since, as explored in related work~\cite{Baelde2021} sequentialising reader sessions does prevent certain kinds of attacks on the unlinkability of certain protocols (such as a previously known attack on the PACE protocol that we will come to in the next Section~\ref{sec:pace}).
In that work, an operator, with symbol $\rotatebox[origin=c]{180}{!}$, is used as a prefix for reader sessions,
which acts like a Kleene star~\cite{Baeten2016}, creating infinitely many copies of a thread sequentially rather than in parallel.
The strategy $\varsigma$, shows clearly that such a strengthened model of the system, where runs of readers with the same ePassport are sequentialised using a Kleene star, will not prevent the attack we describe. Thus strengthening the specification of strong unlinkability only by sequentialising readers will not allow strong unlinkability to be verified. Furthermore, clearly only one ePassport is needed for all attacks we present, so also sequentialising the runs of an ePassport will not affect our analysis.

\begin{figure}[b]
\includegraphics[scale=0.8]{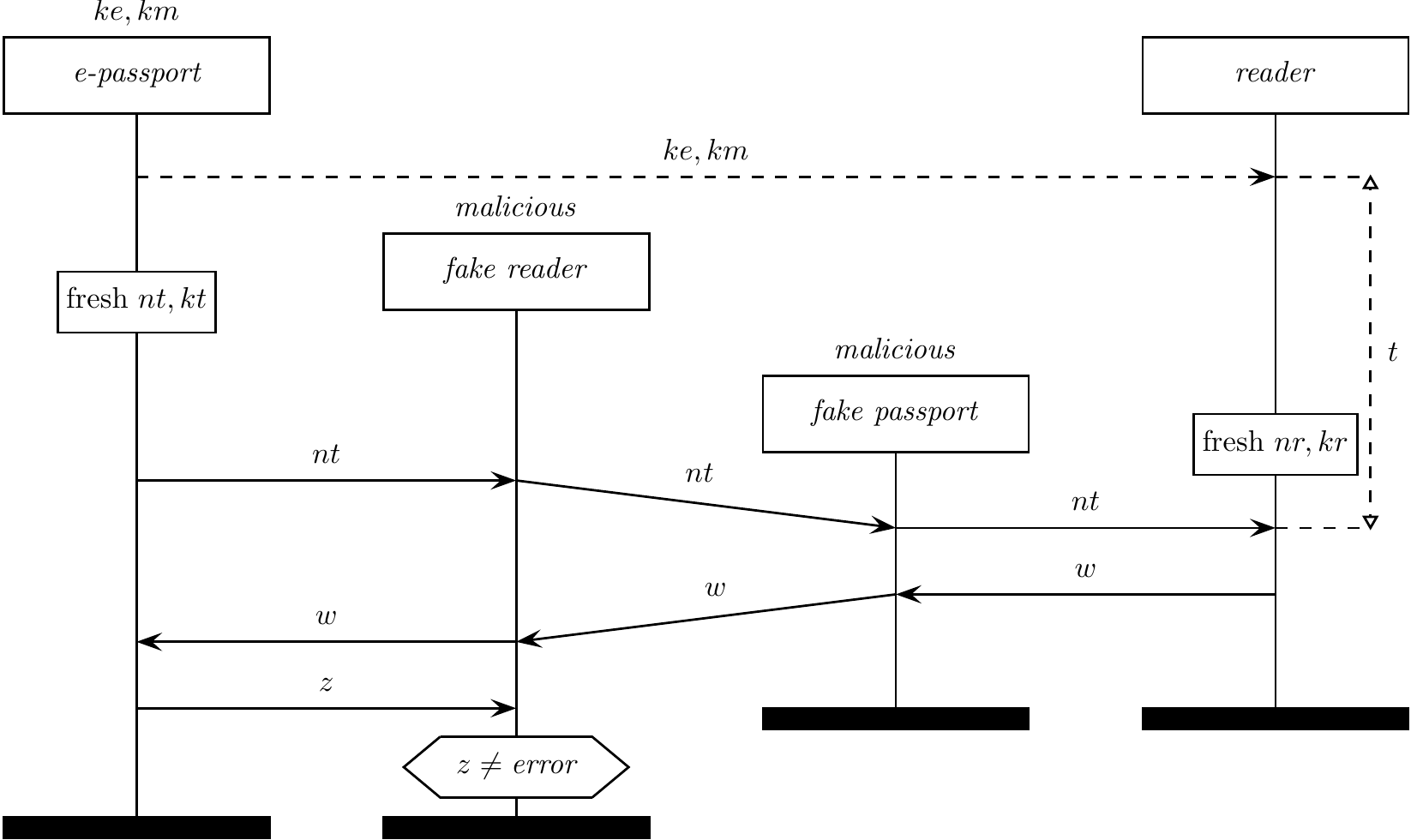}
\caption{Message sequence chart representation of the left branch in Fig.~\ref{fig:game3} for $\SysBAC$.}\label{fig:MSC}
\end{figure}

The second thing to observe about $\varsigma$ is that we can extract message sequence charts (MSCs) from Fig.~\ref{fig:game3} describing the actions of an attacker required to realise either branch of this attack strategy.
Such an MSC is presented in Fig.~\ref{fig:MSC}.
An MSC is fundamentally not designed to represent the branching time and different perspectives in games, however the MSC presented is an accurate depiction of the left branch of the strategy in Fig.~\ref{fig:game3} from the perspective of the system.
Observe that, in Fig.~\ref{fig:MSC}, the reader concerned is loaded with the keys of the ePassport present, corresponding to action $\co{\pass}{(c_1)}$ and $\co{\reader}{(c_3)}$, and then the attacker relays messages between the ePassport and the reader,
where the actions of the honest reader and ePassport correspond to the five actions in the left branch of the strategy,
which eventually result in a non-error message.
 In contrast, if the run of the reader was loaded with different keys, as in the left branch of the specification in Fig.~\ref{fig:game3}, then the final message would be an error. Hence this is the distinguishing strategy, when, according to the specification, the honest session of the reader will be, or should have been with high likelihood, present elsewhere. This is how we interpret what the right branch of the strategy accounts for.

The MSC conveys different information to the depiction of the strategy.
The fake reader and fake ePassport are implicit in this symbolic model, since the attacker is the observer interacting with the observables of the honest participants. However, in the MSC we make the steps that an attacker must perform explicit, which assists us with communicating the attack to stakeholders. ICAO and ISO experts acknowledged they understood the attack presented in this way for the purpose of responsible disclosure~\cite{Delano}.

The MSC diagram in Fig.~\ref{fig:MSC} was easy for students to understand
%, and indeed the students named in the conference version of this paper~\cite{ESORICS},
as the basis of an implementation to demonstrate the feasibility of the attack and for the evaluation of the wider socio-technical context\footnote{
\url{https://github.com/bboyifeel/bac-protocol-unlinkability-exploitation}
Repository with implementations of fake reader and fake ePassport for testing the feasibility of attacks. Maintained by Igor Filimonov.
}.
The dissemination of that broader study is ongoing. We found off-the-shelf readers typically keep the keys of an ePassport loaded until new keys are loaded, making the attack strategy easier than expected by allowing multiple attempts at reidentifying ePassports. The attack may even be triggered inadvertently, simply by having an ePassport in the proximity of a reader loaded with the wrong keys, in fact, this happened live in a lecture where the objective was just to show ePassports could be read by off-the-shelf apps -- the fact that an ePassport in the vicinity inadvertently triggered an error message told us it was not the same as the one the was previously used.
This makes the vulnerability uncovered a real risk, particularly, when the ICAO 9303 standard is being deployed for multiple purposes, not just at airport gates; a risk confounded by a proliferation of powerful RFID readers and covert components, such as 180 micron thick overlay cards~\cite{Anderson2021}.
Stakeholders implementing and deploying ePassport readers should be aware of this risk and possible mitigation strategies.
% destroying the key information within a short time bound.
One may argue that there are more serious side-channel attacks; however, side-channel attacks may be addressed by better implementations of ePassports, while the vulnerability we uncover will remain since it is tied to the specification.

In science in general, interpreting the outputs of a model requires domain expertise.
Security is no exception; hence we expect that there exist further attacks on the BAC protocol that can be uncovered by selecting a distinguishing formula and interpreting it meaningfully in a range of socio-technical scenarios where eDocuments play a role.
For example, for the attacks presented several actions can be permuted, such as the first two actions in $\varsigma$ creating the ePassport and reader channels, without changing fundamental strategy. Some of these permutations may give rise to slightly different scenarios to those described. The essence will however be the same -- the attacker can chose between multiple combinations of devices in order to attempt an authentication session, when, ideally, only one combination of ePassport and reader run should work.
%; whereas, according to the idealised specification, we should not be able to determine whether more than one run of a reader may successfully authenticate an ePassport.

%%%%%%%%%%%%%%%%%%%%%%%%%%% End revision

\subsection{Without else branches there is still an attack.}\label{sec:no-else}

Some papers that analyse the BAC protocol drop the else branch~\cite{Delaune2020}.
This is convenient since not all tools and methods handle else branches.
For example, we may instead try to use the following model of an ePassport in the system and specification processes.
{\small\[
\begin{array}{rl}
\PBAC^{no\_else}(c,ke,km) \triangleq&
\begin{array}[t]{l}
 \nu nt.\cout{c}{nt}.c(y). \\
     %& \clet{m_e,m_m}{\fst{y},\snd{y}} \\
 \match{ \snd{y} = \mac{\fst{y}, km} }
 \match{nt = \fst{\snd{\dec{\fst{y}}{ke}}}}
\\
 \nu kt.\clet{m}{\enc{\pair{nt}{\pair{\fst{\dec{\fst{y}}{ke}}}{kt}}}{ke}}
                         \cout{c}{m, \mac{m, km}}
\end{array}
\end{array}
\]}
However, even for this variant, using bisimilarity, we discover an attack on unlinkability.
To see why, observe that if authentication succeeds, instead of checking the final message sent by the ePassport is not an error, it is sufficient to check that some message is sent at that point. If no message was transmitted, then we can infer the tests on the nonce failed.

A classical $\FM$ formula distinguishing the system from the specification, obtained using the process above to model the ePassport in the scheme at the top of this section, is as follows.
\[
\varphi' \triangleq
\begin{array}[t]{l}
\diam{\co{\reader}(c_1)}
\diam{\co{\reader}(c_2)}
\diam{\co{\pass}(c_3)}
\diam{\co{c_3}(nt)}
\Big(
\\
\qquad
\begin{array}[t]{rl}&
\diam{c_1\, nt}
\diam{\co{c_1}(\varw)}
\diam{c_3\,\varw}
\diam{\co{c_3}(\varz)}\ttt
\\
\wedge &
\diam{c_2\, nt}
\diam{\co{c_2}(\varw)}
\diam{c_3\,\varw}
\diam{\co{c_3}(\varz)}\ttt
~\Big)
\end{array}
\end{array}
\]
Notice the difference compared to $\varphi$ is that $\varphi'$ does not need to test to check that $\varz$ is not an error. Indeed, when we reach the final action of the distinguishing strategy, the specification is unable to perform any action on channel $c_3$, hence cannot simulate the behaviour of the system.

\subsection{From bisimilarity to notions of similarity}\label{sec:sim}
As we have already pointed out for Fig.~\ref{fig:game2},
 the leader in the strategy in Fig.~\ref{fig:game3} is always the same.
Thus we do not require the full power of bisimilarity to discover either of these attack strategies, nor even the strategy described by $\varphi'$ in Sec.~\ref{sec:no-else}.
Indeed, to specify this unlinkability problem, it is sufficient to use a similarity preorder, which is obtained from strong early bisimilarity in Def.~\ref{def:strong} by dropping the requirement that the relation is symmetric.
To be explicit, in the specification of unlinkability, we could employ the following notion of similarity instead of bisimilarity.
\begin{defi}[similarity]\label{def:sim}
A relation between extended processes
$\mathrel{\mathcal{R}}$ is a strong early simulation only if,
whenever $A \mathrel{\mathcal{R}} B$ the following hold:
\begin{itemize}
\item $A$ and $B$ are statically equivalent.
%\footnote{Note that even though similarity defines a preorder, static equivalence is used, rather than ``static implication'' which is the one-sided version of static equivalence, used in related work on applied calculi~\cite{}. A basic reason is that we would yield a definition that is not conservative for the $\pi$-calculus and

%Recent work proposes using ``static implication'' rather than static equivalence. The use of static implication would yield a distinct semantics from the }
\item If $A \lts{\pi} A'$ there exists $B'$ such that $B \lts{\pi}
B'$ and $A' \mathrel{\mathcal{R}} B'$.
\end{itemize}
Process $P$ is simulated by processes $Q$, written $P \preceq Q$, whenever there exists a strong early simulation $\mathcal{R}$ such that $P \mathrel{\mathcal{R}} Q$.
\end{defi}
In both the old scheme for unlinkability, as communicated in CSF'10, and in our new updated scheme introduced in this section, it is trivial that the process modelling the specification, e.g., $\SpecBAC$, is simulated by the system process $\SysBAC$, i.e., $\SpecBAC \preceq \SysBAC$ holds. To see why, intuitively, observe that if we have full control of the system we can always make it behave like the specification by never using the same ePassport twice thus  the specification can be simulated by the system. Hence, the problem of checking unlinkability, when cast as a similarity problem, can be formulated by the problem of checking whether $\SysBAC \preceq \SpecBAC$, i.e., checking whether any observable behaviours the system can perform can be simulated by behaviours of the specification. Stated as a theorem, we have the following which tightens Theorem~\ref{thm:BAC}, where the proof follows from the same strategy as presented in Fig.~\ref{fig:game2}.
\begin{thm}
$\SysBAC \preceq \SpecBAC$ does \textbf{not} hold.
\end{thm}

It is helpful to know that similarity is sufficient for this problem, since similarity preorders have compelling attacker models, e.g., in terms of probabilistic testing semantics~\cite{Deng2008}. Indeed, it is standard in cryptography to assume that an adversary has the power of a probabilistic polynominal-time Turing machine, even if the protocol does not contain probabilistic choices, which is reflected in power made available to the attacker by adopting similarity.

\begin{rem}
There are several variants of similarity in the linear-time / branching-time spectrum~\cite{Glabbeek2001}, which, as touched on in Sec.~\ref{sec:jos}, could be compelling choices for modelling the capabilities of attackers.
%Moving along the linear-time/branching-time spectrum also allows us to design fragments of variants of similarity. For example, it is easy to adapt the proof of Theorem~\ref{thm:char}, to show that by restricting $\FM$ to allow negation on equality and not on modalities, we characterises the notion of similarity as defined above.
We argue that a more broadly applicable design decision would be to employ a stronger notion of similarity called \textit{failure similarity}. For example,
failure similarity can distinguish process
${\bang \cout{a}{go}.\cout{a}{error}} \cpar {\bang \cout{a}{go}}$
from
$\bang \cout{a}{go}.\cout{a}{error}$,
whereas similarity  in Definition~\ref{def:sim} cannot distinguish these processes.
 Thus we can test that an event does not happen, e.g., by using a timeout~\cite{Glabbeek2020}, which is a distinction that cannot be made using similarity.

Such additional expressive power is not required in order to detect unlinkability attacks on the BAC protocol, but might be useful in some scenarios where, for example, the attacker can explicitly observe an error due to the presence of a message, but cannot explicitly observe a success. In such scenarios, a success can be inferred by observing that an error does not occur within an expected time window.
For example, the ICAO 9303 specification of the PACE protocol~\cite{MRTD} only requires the ePassport to send an error message at the end of an unsuccessful authentication session, but does not require it to send any message if the session results in authentication being successful; hence successful authentication from the perspective of an ePassport can be inferred by the absence of an error message at the end of the session within an expected time window.
However, when we model the PACE protocol  in Sec.~\ref{sec:pace}, we use explicit observables for success so as to align with related work, thereby avoiding unnecessary debate about whether our attacks are particular to how we model the PACE protocol (they are not).
Thus, it is safest to verify unlinkability with respect to bisimilarity, which covers all such attacks, including the richer strategy in Fig.~\ref{fig:game}. Recall, in Sec.~\ref{sec:symbolic}, that, by using domain knowledge, we were able to assign a practical meaning to the strategy in Fig.~\ref{fig:game}, which failure similarity does not detect.
% and expanded on further in the conference version of this paper~\cite{ESORICS}.
\end{rem}

%% file: pace.tex
\section{Unlinkability of the PACE protocol}\label{sec:pace}

We address the public communication from the office of the secretary general of ICAO, discussed in the introduction, which challenges whether the unlinkability vulnerability discovered on the BAC protocol is valid for more recent versions of the ICAO 9303 ePassport standard~\cite{MRTD}.
This is a reasonable question, since the $7^{th}$ edition of the ICAO 9303 standard recommends the Password Authenticated Connection Establishment protocol (PACE), as a more secure alternative to BAC\@.

The PACE protocol does improve on the security of the BAC protocol, making attacks giving access to private data stored on the chip more difficult. For example, the PACE protocol satisfies \textit{forward secrecy}~\cite{Bender2009,Coron12}; whereas the BAC protocol does not, that is: if an attacker intercepts ciphertexts in anticipation of, in the future, discovering the key for the ePassport, then she cannot go back and use the key to discover the session key and reveal the encrypted secrets from those old runs of the protocol.

The PACE protocol also eliminates unlinkability vulnerabilities caused by using different errors when the protocol fails for different reasons, as was the case for an implementation of the BAC protocol for French ePassports. Thus we believe that some of the most serious types of attack on unlinkability exploiting the BAC protocol have been addressed in the PACE protocol, where such attacks on the BAC protocol allow an ePassport holder with an implementation interpreting the specification in a particular way to be tracked forever after the messages from one session with a trusted reader have been intercepted.

The clause of the standard that restricts the use of error messages is the following line in section 4.4.2 of part 11 of the ICAO 9303 standard.
\begin{quote}
``An eMRTD chip that supports PACE SHALL respond to unauthenticated read attempts (including selection of (protected)
files in the LDS) with ``Security status not satisfied'' (0x6982).'' % chktex 29
\end{quote}
%Thus we include a 0x6982 error message in the model of PACE, where the ePassport fails authentication.
The above explicit statement is an improvement over the specification of the BAC protocol;
however there are still unlinkability vulnerabilities in the PACE protocol as specified in the ICAO 9303 standard.
Similarly to the vulnerability in the BAC protocol, studied throughout previous sections, there are vulnerabilities in the PACE protocol valid due to differences between a successful and failed authentication session observable to an attacker,
such as the presence of the error message highlighted above.
We formally analyse this vulnerability using bisimilarity following our revised approach to unlinkability justified in Sec.~\ref{sec:II}.

\subsection{The PACE protocol}

\begin{figure}[b]
\centerline{
\includegraphics[width=0.9\textwidth]{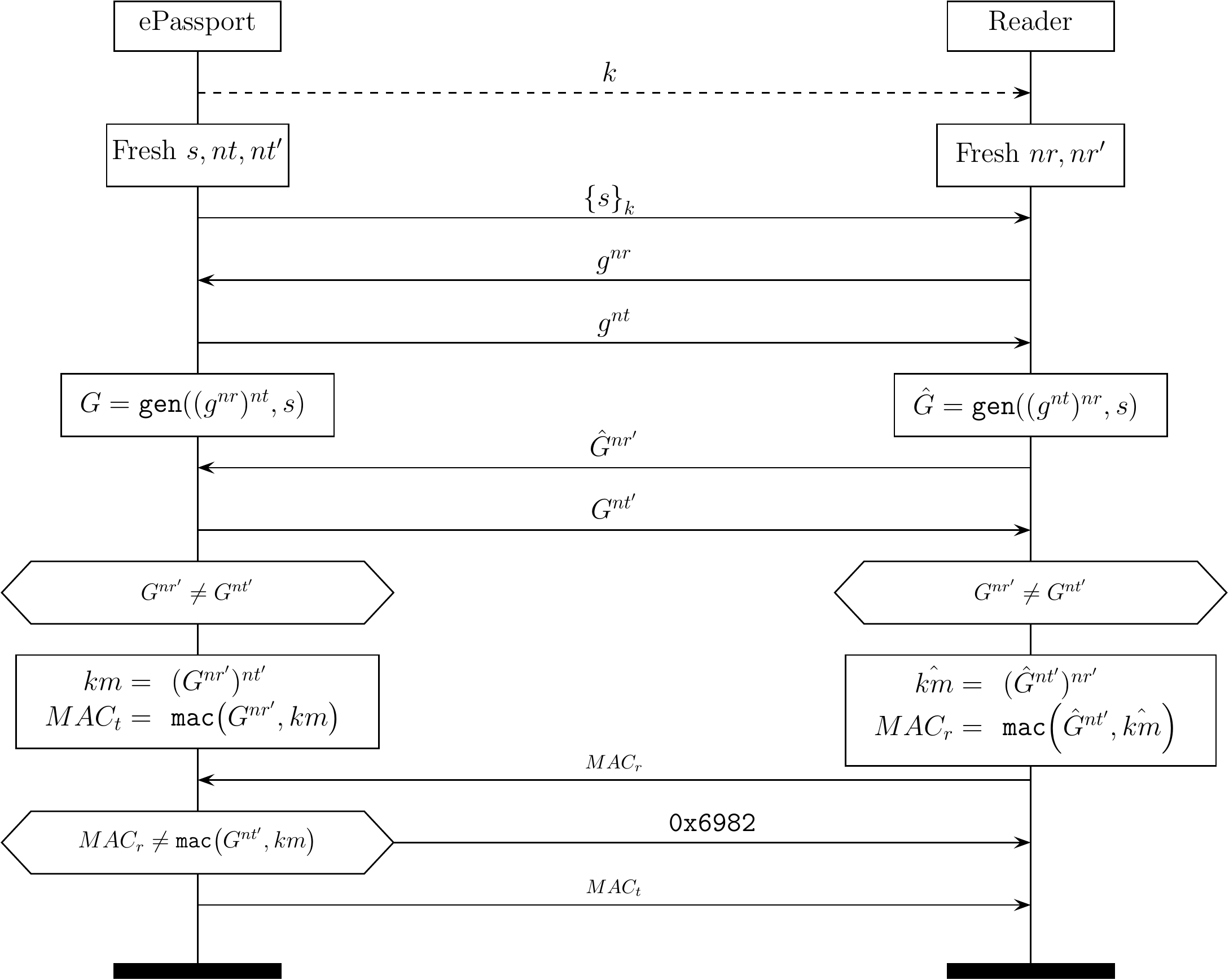}}
%\caption{The ICAO 9303 basic access control mechanism.}
    \caption{The PACE protocol, using a generic mapping based on Diffie-Hellman Key Agreement.
% the attacker may intercept.
%\ZS{Reviewer 3: `get' is ambiguous, use   `get\_challenge'}
}%
\label{fig:pace}
%\vspace{-5ex}
\end{figure}

There are multiple ways to interpret the PACE protocol since it has various operational modes
that permit a number of cryptographic primitives to be used at each stage for establishing shared keys.
We model here the generic mapping which uses a Diffie-Hellman key exchange.
The message exchange, presented in Fig.~\ref{fig:pace}, follows closely related work communicated in the Journal of Computer Security~\cite{Hirschi2019}, thereby avoiding unnecessary debate on how the protocol is interpreted.

The message flow in Fig.~\ref{fig:pace} is as follows.
\begin{enumerate}
\item The ePassport shares information for generating a key $k$ with the reader, usually via an OCR session with the biometric page of the ePassport.
This is represented by the dotted line at the top of the figure.
PACE uses better sources of randomness than BAC, however this does not affect our unlinkability analysis.
\item The ePassport key uses the key $k$ to transmit an encrypted nonce $\enc{s}{k}$ to the reader.
\item The ePassport and reader employ one of several operational modes to create additional randomness for each session. We model the ``generic mapping'' operational mode which employs a Diffie-Hellman handshake.
This information is used to generate shared key $G = \gen{ {(g^{nr})}^{nt}, s}$, where $\gen{\cdot, \cdot}$ is key generation function. Notice $G= \hat{G}$ in Fig.~\ref{fig:pace}.
\item A Diffie-Hellman handshake is performed using $G$ as the generator, which is used to compute a MAC key $km$. Also an  encryption key for the secure messaging phase, which we do not model, is generated at this point.
Again $km = \hat{km}$.
%\RH{following the spec we $km$ should also use $G$.}
%\item
The checks $G^{nr} \neq G^{nt}$ at this point avoid reflection attacks, where an ePassport or reader is used to authenticate itself.
\item Finally the ePassport and reader exchange and verify MACs, using the MAC key $km$, authenticating the public keys exchanged in the previous step.
If authentication fails at this point an error message is produced.
Notice we include the above mentioned error message (0x6982) if authentication fails at the end of the protocol. % chktex 29
\end{enumerate}

\newcommand{\PPACE}{P_{\textit{PACE}}}
\newcommand{\VPACE}{V_{\textit{PACE}}}

\subsection{PACE in the applied \texorpdfstring{$\pi$}{pi}-calculus}
For the PACE protocol we require an extended message theory.
We require symmetric encryption, where decryption is not detectable (modelled by the same equations as employed for the BAC protocol).
For the Diffie-Hellman exchanges we require exponentiation and also a key generating map, which acts like a two parameter hash function.
As for the BAC protocol, MACs are modelled as a two parameter hash function.
This message theory is presented below.
\begin{gather*}
\begin{array}{rlr}
M, N \Coloneqq& x & \mbox{variable} \\
          \mid& \mac{M, N} & \mbox{mac} \\
          \mid& \gen{M, N} & \mbox{generator} \\
	  \mid& M^{N}      & \mbox{exponentiation} \\
          %\mid& \left\langle M, N\right\rangle & \mbox{pair} \\
          %\mid& \fst{M} & \mbox{left} \\
          %\mid& \snd{M} & \mbox{right} \\
          \mid& \enc{M}{N} &\!\!\!\! \mbox{encryption} \\
          \mid& \dec{M}{N} & \mbox{decryption}
\end{array}
\qquad\qquad
\begin{array}{c}
%\fst{\pair{M}{N}} =_{E} M
%\quad
%\snd{\pair{M}{N}} =_{E} N
%\\
\dec{\enc{M}{K}}{K} =_{E'} M
\\[12pt]
\enc{\dec{M}{K}}{K} =_{E'} M
\\[12pt]
{(M^{N})}^K =_{E'} {(M^{K})}^N
\end{array}
\end{gather*}

The ePassport and reader for the PACE protocol can be modelled in the applied $\pi$-calculus as follows.
\begin{gather*}
\PPACE(c, k)
\triangleq
 \begin{array}[t]{l}
 \nu s.\cout{c}{\enc{s}{k}}.
 c(x).\cout{c}{g^{nt}}. \\
 \clet{G}{\gen{s, x^{nt}}} \\
 c(y).\nu nt'.\cout{c}{G^{nt'}} \\
 \match{G^{nt'} \neq y}
 c(z).\\
 \clet{km}{z^{nt'} } \\
 \texttt{if}\,z = \mac{G^{nt'},km}\,
 \\\texttt{then}\,
 \cout{c}{\mac{z, km}}\, \\
 \texttt{else}\,\cout{c}{\error}
 \end{array}
\qquad
\VPACE(c, k)
\triangleq
 \begin{array}[t]{l}
 c(x).\cout{c}{g^{nr}}.c(y). \\
 \clet{G}{\gen{\dec{x}{k}, y^{nr}}}\\
 \nu nr'.\cout{c}{G^{nr'}}.c(z). \\
 \match{G^{nr'} \neq z}
 \clet{km}{z^{nr'} } \\
 \cout{c}{\mac{z, km}} \\
 c(m).
 \match{m = \mac{G^{nr'}, km}}
 c(n)
 \end{array}
\end{gather*}
Notice only the ePassport features an error message if authentication fails in the final step.
Also, we add a dummy event
 $c(n)$ at the end of the reader session, for the sake of modelling that a reader will proceed to do something after successfully authenticating. This is to align with related work in communicated in the Journal of Computer Security~\cite{Hirschi2019}, in order to facilitate a comparison of results obtained.

Using the above processes, and our revised scheme for unlinkability in the previous section
we obtain the following result confirming there are attacks on the unlinkability of PACE\@.
\begin{thm}\label{thm:pace}
$\SysPACE \not\sim \SpecPACE$, where
\[
\begin{array}{rl}
\SysPACE \triangleq
& \bang\mathopen{\nu k.} \bang\left(
    \mathopen{\nu c.}\cout{\pass}{c}.\PPACE(c, k)   \cpar
    \mathopen{\nu c.}\cout{\reader}{c}.\VPACE(c, k) \right)
\\
\SpecPACE \triangleq
& \bang\mathopen{\nu {k}.}\left(
    \mathopen{\nu c.}\cout{\pass}{c}.\PPACE(c, {k})   \cpar
    \mathopen{\nu c.}\cout{\reader}{c}.\VPACE(c, {k}) \right)
\end{array}
\]
\end{thm}
\begin{proof}
Consider the following $\FM$ formula.
\[
\xi
\triangleq
\begin{array}[t]{l}
\diam{\co{\reader}(c_1)}
\diam{\co{\reader}(c_2)}
\diam{\co{\pass}(c_3)}
\diam{\co{c_3}(t)}
\Big(
\\
\begin{array}[t]{rl}&
\diam{c_1\, t}
\diam{\co{c_1}(u)}
\diam{c_3\,u}
\diam{\co{c_3}(v)}
\diam{c_1\,v}
\diam{\co{c_1}(w)}\\&\quad
\diam{c_3\,w}
\diam{\co{c_3}(x)}
\diam{c_1\,x}
\diam{\co{c_1}(y)}
\diam{c_3\,y}
\diam{\co{c_3}(\varz)}
\left( \varz \neq \error \right)
\\[5pt]
\wedge \!\!\!\!\!&
\diam{c_2\, t}
\diam{\co{c_2}(u)}
\diam{c_3\,u}
\diam{\co{c_3}(v)}
\diam{c_2\,v}
\diam{\co{c_2}(w)}\\&\quad
\diam{c_3\,w}
\diam{\co{c_3}(x)}
\diam{c_2\,x}
\diam{\co{c_2}(y)}
\diam{c_3\,y}
\diam{\co{c_3}(\varz)}
\left( \varz \neq \error \right)
~\Big)
\end{array}
\end{array}
\]
Since $\SysPACE \vDash \xi$,
but $\SpecPACE \nvDash \xi$,
 by Theorem~\ref{thm:char}, $\SysPACE \not\sim \SpecPACE$.
\end{proof}
The formula $\xi$ proving that unlinkability does not hold for the PACE protocol
follows a similar pattern to the formula $\varphi$, used in the previous section to certify that BAC fails unlinkability.
In this strategy, the system starts two readers and an ePassport with the same keys.
The specification can only follow using a strategy where one of the readers will fail authentication.
To win this game, the system simply chooses to authenticate with the reader that is expected to fail in the specification.
That reader will obviously successfully authenticate in the system, thereby concluding our distinguishing strategy.

\subsection{Another attack strategy from related work.}
Infinitely many distinguishing strategies exist violating the unlinkability of the PACE protocol.
Indeed related work~\cite{Hirschi2019} discovered a violation of the unlinkability of the PACE protocol, that can be described as a trace.
An attack on unlinkability formulated using trace equivalence is always also an attack on bisimilarity.
The added value that our methodology brings to that attack, is that we can certify their attack by using the following classical $\FM$ formula, describing the distinguishing trace discovered in the above mentioned related work.
The $\FM$ formulae characterising trace equivalence are those consisting of diamond modalities only ending with some formula that does not involve modalities.
\[
\vartheta
\triangleq
\begin{array}[t]{l}
\diam{\co{\reader}(c_1)}
\diam{\co{\reader}(c_2)}
\diam{\co{\pass}(c_3)}
\diam{\co{c_3}(t)}
\diam{c_1\, t}
\diam{c_2\, t}
\\
\diam{\co{c_1}(u_1)}
\diam{\co{c_2}(u_2)}
\diam{c_2\,u_1}
\diam{c_1\,u_2}
\diam{\co{c_1}(v_1)}
\diam{\co{c_2}(v_2)}
\diam{c_2\,v_1}
\diam{c_1\,v_2}
\\
\diam{\co{c_1}(w_1)}
\diam{\co{c_2}(w_2)}
\diam{c_2\,w_1}
\diam{c_1\,w_2}
\diam{\co{c_1}(x_1)}
\diam{\co{c_2}(x_2)}
\diam{c_2\,x_1}
\diam{c_1\,x_2}
\\
\diam{\co{c_1}(y_1)}
\diam{\co{c_2}(y_2)}
\diam{c_2\,y_1}
\diam{c_1\,y_2}
\diam{\co{c_1}(z_1)}
\diam{\co{c_2}(z_2)}
\diam{c_2\,z_1}
\diam{c_2\,m}\ttt
\end{array}
\]
In the strategy described by the formula above, the encrypted nonce sent by the ePassport at the beginning of the protocol
is replayed to two different readers.
The two readers are then used to authenticate each other, exploiting the fact that the protocol is symmetric in the role of the reader and ePassport.
Authentication will only be successful if both readers have the same keys, otherwise either reader will fail the check on the MAC at the final step of the PACE protocol.
Thus, assuming that a reader does something after authentication, we learn whether both readers talked with the same ePassport.

The attack $\vartheta$ is quite different from attack $\xi$. Notice $\vartheta$ requires two readers to be fully active during the attack, not simply present in principle, as a choice in a game, suggesting that it may be more difficult to exploit, despite being described as a trace.
Also, notice the attack can be mitigated in several ways, e.g., by initially responding to the ePassport in the same way regardless of whether it authenticates or not, hence, since an ePassport cannot also be authenticated at the same time, there is no way to continue with the secure messaging phase.
%Note this fix is within the scope of the ICAO 9303 standard;
An alternative fix preventing this attack is proposed in related work~\cite{Hirschi2019}, modifying the protocol
such that additional role specific information is added to the handshake.
However, these fixes will not mitigate the more serious problem described in $\xi$, so do not really improve unlinkability.
%\SM{So they claim to have a fix for their attack $\vartheta$ and then we say that it is not a fix for $\vartheta$? That's strange. Do we mean to say that although it's a fix for $\vartheta$ it doesn't fix our attack $\xi$?}

%Hence such probabilistic attacks, count as attacks and get worse as the observable diversity of implementations increase.

%% file: conclusion.tex
\section{Related and future work}\label{sec:related}

A closely related paper, that is not already covered by remarks in the body of the paper was communicated in S\&P'18~\cite{deepsec}. That paper announces the discovery of attacks on the BAC protocol using the bounded trace equivalence checker DeepSec, but without further discussion.
An interesting difference between that paper and the current work is that, while we build on the  original formulation of unlinkability, as communicated in CSF'10~\cite{Arapinis2010}; the S\&P'18 paper proposes another model where, instead of using a specification process, unlinkability is modelled in terms of two systems where the number of identical users in the system differ.
In their alternative model,
%in the two session case,
% (the unbounded case in their approach is not immediately obvious),
one process models two sessions featuring the same ePassport twice, which is compared to another process featuring two sessions each featuring a different ePassport.
That is, unlinkability is formulated such that the following two processes are compared, using trace equivalence.
\[
\textit{Diff}
\triangleq
\begin{array}[t]{l}\mathopen{\nu ke. \nu km.}\left( \VBAC(c_r,ke,km) \cpar \PBAC(c_p,ke,km) \right) \cpar \\ \mathopen{\nu ke. \nu km.}\left( \VBAC(c_r,ke,km) \cpar \PBAC(c_p,ke,km) \right) \end{array}
\]
\[
\textit{Same}
\triangleq
\begin{array}[t]{l}\mathopen{\nu ke. \nu km.}\bigl(\VBAC(c_r,ke,km) \cpar \PBAC(c_p,ke,km) \cpar \\\phantom{\mathopen{\nu ke. \nu km.xx}}\VBAC(c_r,ke,km) \cpar \PBAC(c_p,ke,km)  \bigr) \end{array}
\]
In the above $c_r$ and $c_p$ are used as fixed channels for communications with all ePassports or all readers respectively.
Attack traces discovered using the above method, communicated in S\&P'18,
%presented in Fig.~\ref{fig:deepSec},
can be used to confirm that two sessions are certainly not with the same ePassport. In particular, consider the following formula describing a trace that holds for the first process but does not hold for the second process.
\[
\textit{Diff}
\vDash
\begin{array}[t]{l}
\diam{\co{c_p}(n_1)}
\diam{\co{c_p}(n_2)}
\diam{c_r\,n_2}
\diam{c_r\,n_2}
\\
\qquad
\diam{\co{c_r}(m_1)}
\diam{\co{c_r}(m_2)}
\diam{c_p\,m_1}
\diam{c_p\,m_2}
\diam{\co{c_p}(e_1)}
\diam{\co{c_p}(e_2)}
\left(
e_1 = \error
\wedge
e_2 = \error
\right)
\end{array}
\]
In the first line of the above trace, the first two messages correspond to sending two nonces, and only the second nonce is fed as an input to both readers. On the second line of the formula, the protocol continues for both sessions and the protocol ends with both ePassports sending error messages.
If both readers are using the same nonce, then both ePassports can only send error messages at the last step if both ePassports are different; if both were the same ePassport then one of the two sessions would successfully authenticate, hence there could not have been two error messages. I.e., it is impossible for $\textit{Same}$ to satisfy the above formula.

The limitation we see is the above mentioned approach communicated in S\&P'18 discovers attack traces that cannot be used to positively confirm that two sessions are with the same ePassport.
What we mean is that there is no trace that holds for the process $\textit{Same}$ that does not hold for the process $\textit{Diff}$; but for a trace-like attack on unlinkability we should surely be able to provide a trace that links two sessions. Thus we should be careful interpreting the above result --- it does not mean that the above method discovers an attack on unlinkability that is in the form of a trace.

Further discussion on the above model of unlinkability appears in the conference version of this paper~\cite{ESORICS}, where it is clarified that the above limitation is due to modelling decisions and is not a feature of the DeepSec tool that the S\&P'18 paper showcases. The DeepSec tool can also be used to verify finite  formulations of the unlinkability problem using our preferred ``system v.s.\ specification'' approach --- in which case DeepSec discovers no attacks that are in the form of a trace for the BAC protocol. In particular, DeepSec can verify that \textit{Diff} is trace equivalent to the following process where either there is a choice between starting the second session with the same keys as the first or with the new keys, which is a bounded approximation of the system following established schemes for unlinkability.
\[
\mathopen{\nu ke. \nu km.}\Big(
  \begin{array}[t]{l}
  \VBAC(c_r,ke,km) \cpar \PBAC(c_p,ke,km) \cpar
  \\
  \big( \begin{array}[t]{l}
              \left( \VBAC(c_r,ke,km) \cpar \PBAC(c_p,ke,km) \right) +  \\
              \mathopen{\nu·ke,km.}\left(\VBAC(c_r,ke,km) \cpar \PBAC(c_p,ke,km) \right)~\big)~\Big)
           \end{array}
  \end{array}
\]
In the above process, $+$ is non-deterministic choice, which is easy to add to the applied $\pi$-calculus.

An approach similar to the approach communicated in S\&P'18, where two systems are compared in which users in a system are permuted, has been
thoroughly investigated and demonstrated to be the preferable approach for formulating voter privacy, which is a property of eVoting systems~\cite{Delaune2009}. We should clarify that we are not arguing against using a ``permutations of a system'' approach to voter privacy. What we are arguing is that the ``system vs.\ specification'' approach adopted in the current paper and in the CSF'10 paper is appropriate for unlinkability, since if an attack trace is discovered the attack trace will be able to positively confirm that two sessions are with the same ePassport, i.e., the sessions will be linked.
%Notice that the limitation highlighted above, where a property is violated by only discovering attacks confirming two users are different, is not an issue for voter privacy since all voters are different already.

%\begin{figure}
%\includegraphics[scale=0.5]{mscs/deepSec}
%\caption{}\label{fig:deepSec}
%\end{figure}

%The unlinkability of BAC has been analysed in several other papers... List them all briefly.

%\subsection{Balancing bisimilarity and trace equivalence}

\subsection{
A preliminary discussion on mitigation strategies
}\label{ref:unlink}
We summarise here three quite different mitigation strategies and present some preliminary findings.

\subsubsection{Timeouts}
A mitigation strategy, which we have recommended to stakeholders, is to guarantee that implementations of readers only hold ePassport keys for a short period of time.
This would render the reader useless in an attack strategy where the goal is to reidentify someone in the future.
In Fig.~\ref{fig:MSC}, the time to keep small is indicated by $t$: the time between the reader being prepared with they keys of a particular passport and the RFID session being triggered.
%making the attack strategy implementable only when the ePassport holder was alread in the vicinity of the reader.
This eliminates use cases where an ePassport holder who has recently passed through a checkpoint is reidentified, but might not mitigate other use cases. For example, a user may be expected to approach a reader, e.g., to provide their identity at a service desk, and hence the reader is loaded with a particular key in a time window. In that case, an attacker may attempt to find who was intended to be at the service desk without them being present. An extensive study of use cases emerging we push to future work.

The above mitigation strategy can be modelled by imposing causal dependencies in our model.
A timeout built into the reader would allow a single run of an ePassport to occur before the reader session times out. Thus the creation of a new run of an ePassport causally depends on any events of an ongoing run of a reader involving the same ePassport, which must first be completed.
In addition, we can assume that ePassport runs themselves are sequentially ordered, by the nature of the chip.
We may also assume that the readers are loaded with the keys of a single ePassport sequentially, e.g., along a path through checkpoints.
Surprisingly, forcing all these causal assumptions would still not be enough to prevent the attack $\varsigma$ on unlinkability presented in Section~\ref{sec:new}, since we have not stipulated that a single ePassport run can be held open for a length of time sufficient for the attacker to choose between relaying message to one of two readers in a sequence, both of which may successfully authenticate the ePassport.
Thus there must be a fourth causal dependency imposed to prevent attacks, which is not enforced by a timeout for the reader:
we should prevent a new reader from being initiated before any run of an ePassport involving the same keys terminates. This does not model a timeout imposed by the reader, but instead a timeout imposed in practice by how long an attacker can keep a device in the vicinity of an ePassport holder.

If we make all four causal assumptions in the paragraph above, the effect is each run of an ePassport and its corresponding reader session are sequentialised together, in parallel.
That is, a run of an ePassport and a reader must both be used entirely before any new run can be created involving the keys of the same ePassport.
This can be modelled by making use of a Kleene star operator $\rotatebox[origin=c]{180}{!}$, proposed in related work~\cite{Baelde2021}, in the following alternative scheme for the system.
\[
\bang\mathopen{\nu \vec{k}.} \mathopen{\rotatebox[origin=c]{180}{!}}\left(\mathopen{\nu c.}\cout{\pass}{c}.\ePassport(c,\vec{k}) \cpar \mathopen{\nu c.}\cout{\reader}{c}.\eReader(c,\vec{k})\right)
\]
Our hypothesis is that the above process is bisimilar to the scheme for the specification for both the BAC and PACE protocols.
An account of the semantics of $\rotatebox[origin=c]{180}{!}$ in this context and a proof of this claim are pushed to future work.

\subsubsection{Obscuring messages.}

Another alternative is to probabilistically encrypt the error message, or produce random noise when the ePassport fails to authenticate the reader. Note this is within the scope of the BAC protocol specification, since the specification does not fix the form of the error message. Nevertheless, to our best knowledge, real ePassports implementing the BAC protocol send errors as constant plaintext messages.
This mitigation strategy would only be effective for use cases where the attacker cannot observe the consequences of using an eDocument. For example, one may consider a ``polite'' registration system where you are offered to register using your electronic identity, but choosing not to or not succeeding to do so is permitted. Perhaps those who do not provide their electronic identities will be picked up by other safety nets such as a human attendant who later checks participants against a register rather than as an immediate effect such as the opening of a gate.

Such a refinement of the BAC protocol where error messages are obscured, does in fact satisfy unlinkability, as long as nothing happens after executing the BAC protocol.
We can model one such variant of the BAC protocol as follows.
\[
\begin{array}{rl}
\PBAC^{fixed}(c,ke,km) \triangleq&
\begin{array}[t]{l}
 \nu nt.\cout{c}{nt}.c(y). \\
     %& \clet{m_e,m_m}{\fst{y},\snd{y}} \\
 \texttt{if}\,\snd{y} = \mac{\fst{y}, km}\,\texttt{then} \\
\qquad
 \begin{array}[t]{l}
 \texttt{if}\,nt = \fst{\snd{\dec{\fst{y}}{ke}}}\,\texttt{then} \\
%      & \begin{array}[t]{ll}
%        \texttt{then}
\qquad
 \begin{array}[t]{l}
                         \nu kt.\clet{m}{\enc{\pair{nt}{\pair{\fst{\dec{\fst{y}}{ke}}}{kt}}}{ke}} \\
                         \cout{c}{m, \mac{m, km}}
                         \end{array}  \\
 \texttt{else}\,
                         \nu r.\clet{m}{\enc{\pair{r}{\error}}{ke}} % chktex 1
                         \cout{c}{m, \mac{m, km}}
\end{array}
\\
 \texttt{else}\,
                         \nu r.\clet{m}{\enc{\pair{r}{\error}}{ke}} % chktex 1
                         \cout{c}{m, \mac{m, km}}
        \end{array}
\end{array}
\]
Let $\SysBAC^{fixed}$ and $\SpecBAC^{fixed}$ denote, respectively, the system and specification for unlinkability of BAC where the above model of the ePassport role is used.
For this model, we have a proof that unlinkability does hold, i.e., $\SysBAC^{fixed} \sim \SpecBAC^{fixed}$.
The publication of a proof for this claim is pushed to future work, in the interest of focussing on attacks in this paper.
%Presenting details about how to systematically find proofs for such unlinkability claims, as opposed to finding attacks, requires additional heuristics to be explained.

A problem with the above refinement of the BAC protocol satisfying unlinkability is that making BAC unlinkable does not guarantee that the whole ePassport protocol satisfies unlinkability. The BAC protocol is just for authentication and establishing a session key. After authenticating the ePassport proceeds with a \textit{secure messaging phase} that uses the session key to transmit personal data stored in the ePassport~\cite{MRTD}. Thus it is sufficient for an attacker to look at whether the protocol proceeds with secure messaging or not in order to determine whether authentication was successful. That knowledge can be used to the same effect as observing whether or not an error message was sent. Thus, for the above fix to be fully effective even in a ``polite'' system, the secure messaging phase should proceed even if the ePassport does not authenticate, transmitting dummy data indistinguishable to an observer from the real data.
Such a mitigation strategy is outside the scope of the current ICAO specification and does not significantly improve unlinkability for the standard use cases for ePassports.

\subsubsection{Using one-time keys.}
An arguably better mitigation strategy is to use Time-based One-time Passwords (TOTP) to make the six-digit key for PACE change periodically, as explored in the thesis of our student~\cite{Filimonov2020}. Verification of that strategy is immediate, since TOTP has the effect of generating a new key for every run of the protocol, making the system and specification trivially bisimilar.
Implementing TOTP has further security and privacy advantages, forgoing other attacks on the system, including some social attacks, and has been deployed in card form for ePayments, so could be easily integrated into electronic ID cards implementing the ICAO 9303 standard.
%, and does not significantly violate the ICAO 9303 specification.
The challenge for ePassports is more likely to be at the policy level, since questions may be raised at international checkpoints if countries do not agree that TOTP is acceptable technology.
Dissemination of the socio-technical evaluation of this mitigation strategy we push to future work.

\subsection{Further risks to unlinkability}
There are many potential privacy risk for the PACE protocol that are not directly captured by the symbolic models in this work.
Some are simple to exploit, such as the fact that the PACE protocol offers several different operational modes.
In order for a reader to determine the appropriate operational model, before starting the protocol, the ePassport declares the operational modes of the PACE protocol that it implements.
In an environment, such as an airport, where many different ePassports implementing different operational modes coexist, a user may be tracked with a probability better than a random guess.
% of whether two sessions are with the same ePassport.

The actual probability of guessing correctly that the same ePassport is involved in two sessions depends on the number of different implementations of ePassport and the expected movements of their holders.
We expect the advantage gained by such a strategy to be non-negligible and it is standard in security and privacy models that gaining a non-negligible advantage counts as an attack.
Permitting may different implementations of protocols is an oversight of the ICAO 9303 standard.

\section{Conclusion}\label{sec:conclusion}

This paper confirms there are attacks on the authentication protocols proposed in the latest ICAO 9303 specification for ePassports.
Both the BAC and PACE protocols feature attacks that can be described as a distinguishing strategy in a game played according to a specification of what it means to satisfy unlinkability. Attacks on the BAC protocol, communicated in Theorem~\ref{thm:BAC} and Theorem~\ref{thm:BAC2}, are interesting since the former resolves flawed claims that no such attack exists according to a formulation of unlinkability dating back to CSF'10; while the latter irons out limitations of that original model concerning the capabilities of an attacker to distinguish messages from different readers and ePassports.

An interesting aspect of the new attack we discover on the PACE protocol, as formulated in Theorem~\ref{thm:pace}, is that it is not mitigated by defensive strategies that may be introduced to mitigate attacks previously discovered using trace equivalence as communicated in the Journal of Computer Security~\cite{Hirschi2019} (e.g., sequentialising reader sessions will not mitigate the newly discovered attack, nor will adding role information distinguishing messages of the same form originating from a reader and from an ePassport). Furthermore, the attacks discovered previously using trace equivalence require two honest readers to  actively participate in the attack, whereas the new attack discovered using bisimilarity requires only one honest reader to actively participate in the attack, meaning that the attack can be realised in a broader range of scenarios. This observation challenges claims communicated in the Journal of Computer Security, where it is argued that bisimilarity is too strong and hence trace equivalence should be employed.
Their argument is provided to support their model that relies on trace equivalence in order to prove that unlinkability of the BAC protocol holds. Since their results lead to contradictory advice compared to ours, differences are worth clarifying.

The crux of their argument is based on remarks communicated in CSF'10~\cite{Arapinis2010} claiming that bisimilarity may distinguish processes due to their internal state --- an argument we contest since bisimilarity is all about the games played between the adversary and its environment using observations only; and never can distinctions be made based on differences in internal state (that is the point of such observational equivalences).
Presenting multiple viewpoints is healthy for academic debate; however, the fact that the attacks we discover using bisimilarity are not spurious and furthermore are easier to realise is evidence that trace equivalence is insufficient for verifying interactive systems such as security protocols.
To further support our argument that the use of bisimilarity (or, as a compromise, a suitable notion of similarity, as discussed in Sec.~\ref{sec:sim}) is important for such security and privacy problems, we remark that it should not be a surprise to cryptographers that the adversary can play a strategy in a game to gain a non-negligible advantage,
%when attempting to linking sessions of a protocol, or indeed in a range of security and privacy problems,
%This  as modelled by bisimilarity,
since related assumptions about the adversary are standard in the long-established school of \textit{computational security} used to formally reason about cryptographic primitives.
We quote R.L.~Rivest on the topic of games in cryptography~\cite{Menezes1996}:
\begin{quote}
``Cryptography is also fascinating because of its game-like adversarial nature. A good
cryptographer rapidly changes sides back and forth in his or her thinking, from attacker
to defender and back. Just as in a game of chess, sequences of moves and countermoves
must be considered until the current situation is understood.''
\end{quote}

We argue that the above remark should also hold for the \textit{symbolic verification of cryptographic protocols}, which encompasses the methodology employed in this work.
In the setting of this work, the game is defined by a bisimilarity problem specifying what it means for a protocol to satisfy unlinkability, while attacks are strategies improving the chances of an attacker winning the game.
Such strategies can be conveniently described using modal logic formulae.

The insight obtained, concerning the existence of unlinkability attacks on the latest ePassport standards, is impactful for society, since ePassports and electronic ID cards are used by the citizens of over 150 countries at the time of writing. This amounts to an estimated 4 billion eDocuments in circulation that implement the ICAO 9303 standard.
The manufactures and operators of readers should be made aware of mitigation strategies.
Some preliminary ideas on mitigation strategies are discussed in Sec.~\ref{ref:unlink}.
% which are discussed in Section~\ref{} the conference version of this paper~\cite{ESORICS}.
%A longer term fix would be to redesign the protocols in the ICAO 9303 specification such that ePassport holders need not rely on the manufacturers and operators of readers to address such vulnerabilities.
%Further mitigation strategies are possible,

%Thus, however, benign you judge this unlinkability attack to be,

Further to uncovering the above mentioned attacks, this paper makes multiple technical contributions.
We proposed and justified a new scheme for unlinkability problems in Sec.~\ref{sec:II}.
We proposed a definition of open bisimilarity for the applied $\pi$-calculus (Def.~\ref{def:open}).
Our formulations of weak and strong early bisimilarity and similarity (Defs.~\ref{definition:early},~\ref{def:strong}, and~\ref{def:sim}) for the applied $\pi$-calculus are also new for the applied $\pi$-calculus, incorporating minor improvements facilitating verification, such as the adoption of a set of rules that make labelled transitions image finite.
These improvements are conventional from the perspective of established work on the $\pi$-calculus; thus, in that direction, we simply modernise the applied $\pi$-calculus with respect to advances in the $\pi$-calculus literature.
The formulation of the modal logic ``classical $\FM$'' is also new, as is the rather short and neat proof (in Appendix~\ref{sec:app2}) of the fact that classical $\FM$ characterises strong early bisimilarity for the applied $\pi$-calculus (Theorem~\ref{thm:char}).
The methodology of using a classical $\FM$ formula to certify attacks is new, as is the use of open bisimilarity to discover distinguishing strategies that are transformed into attacks whenever the distinguishing strategy is not spurious.
A reformulation of unlinkability removing $\tau$-transitions has appeared in related work~\cite{Hirschi2019}, but the proof that the new specification preserves the original specification in terms of bisimilarity (Theorem~\ref{thm:strong}) is new, as is the observation that this transformation reduces the unlinkability problem to a problem were image finiteness holds and strong notions of bisimilarity may be applied.
In short, in order to solve this problem, we have set up a rich tool chain of methods that can be applied beyond the problem of analysing the unlinkability of the ICAO 9303 standard.

\subsubsection*{Acknowledgements.}

The formulation of the sequent calculus in Fig.~\ref{fig:constraints} for calculating most general recipes is new, but should be attributed to Alwen Tiu.
He made the observation that annotating deducibility constraints in a sequent calculus presentation allows us to calculate the recipes required for an attacker to deduce a message.
We thank Davide Sangiorgi and the jury Luca Aceto, Jos Baeten, Patricia Bouyer-Decitre, Holger Hermanns, and Alexandra Silva for the explicit mention of the results communicated in this paper in the report on his CONCUR 2020 Test-of-Time Award~\cite{Aceto2020}.
Last, but not least, we congratulate Jos Baeten on the occasion of his retirement and dedicate this work to his leadership in the field of concurrency.

%% file: proof.tex
\section{Reducing weak to strong bisimilarity}\label{sec:app1}

We provide here a proof for Lemma~\ref{lemma:strong}, which is used to prove that unlinkability when expressed in terms of a strong bisimilarity problem, is equivalent to a formulation of unlinkability in terms of weak bisimilarity (Theorem~\ref{thm:strong}).
In the proof of the lemma below we employ equivariance, which simply allows names to be swapped.
\begin{defi}
Equivariance is the least congruence extending $\alpha$-conversion such that $\nu x. \nu y.P \equiv \nu y. \nu x.P$, thereby allowing the order of name binders to be ignored.
\end{defi}
Working up to equivariance has been shown to significantly reduce the search space when constructing a bisimulation~\cite{Tiu2016}.

\begin{lem}[Lemma~\ref{lemma:strong}]
For any $P$ and $Q$ such that $c_k$ is fresh for $P$ and $Q$, we have
\[
\mathopen{\nu c_k.}\left( \bang c_k(\vec{k}).P \cpar \bang\mathopen{\nu \vec{k}.}\cout{c_k}{\vec{k}}.Q \right)
\approx
\bang\mathopen{\nu \vec{k}.}\left(P \cpar Q \right)
\]
\end{lem}

\begin{proof}
Define $\mathcal{R}$ to be the least symmetric relation, upto
%associativity and commutativity of $\cpar$ and
%scope extrusion and
equivariance,
such that for any $R_i$ and $S_i$ such that $c_k$ is fresh for $R_i$ and $S_i$
and $\vec{r}$ is fresh for $P$ and $Q$,
we have that the following extended process
\[
A \triangleq
\begin{array}[t]{l}
 \mathopen{\nu c_k, \vec{k_1}, \vec{k_2}, \ldots \vec{k_n}, \vec{r}.}\Big(~\sigma \cpar
%\\ \qquad
     R_1 \cpar \ldots R_n \cpar !c_k(\vec{k}).P \cpar
%\\\qquad
     S_1 \cpar \ldots S_n \cpar {!\nu \vec{k}.\cout{c_k}{\vec{k}}.Q}~\Big)
\end{array}
\]
is related by $\mathcal{R}$ to the following extended process
\[
B \triangleq
\begin{array}[t]{l}
 \mathopen{\nu \vec{k}_{{f(1)}}, \vec{k}_{{f(2)}}, \ldots \vec{k}_{{f(m)}}, \vec{r}.}\Big(~\sigma \cpar
%\\ \qquad
     R_{f(1)} \cpar S_{f(1)} \cpar \ldots \cpar R_{f(m)} \cpar S_{f(m)} \cpar
%\\\qquad
    !\nu \vec{k}.\left(P \cpar Q \right)~\Big)
\end{array}
\]
where $f \colon \left\{ 1 \ldots m \right\} \rightarrow \left\{ 1 \ldots n \right\}$ is injective
and $R_i = P\sub{\vec{k}}{\vec{k_{i}}}$ and $S_i = Q\sub{\vec{k}}{\vec{k_{i}}}$ for $i \in \left\{1 \ldots n \right\} \setminus f(\left\{ 1 \ldots m \right\})$.

There are two cases to pay attention to concerning extra $\tau$-transitions in $B$.
Firstly, consider
\[
A \lts{\tau}
\begin{array}[t]{l}
 \mathopen{\nu c_k, \vec{k_1}, \vec{k_2}, \ldots \vec{k_n}, \vec{k_{n+1}}, \vec{r}.}\Big(~\sigma \cpar
\\ \qquad
     R_1 \cpar \ldots R_n \cpar P\sub{\vec{k}}{\vec{k_{n+1}}} \cpar !c_k(\vec{k}).P \cpar
\\\qquad
     S_1 \cpar \ldots S_n \cpar    {!\nu \vec{k}.\cout{c_k}{\vec{k}}.Q}~\Big)
\end{array}
\]
This can be matched by $B$ by performing zero transitions, whilst staying in the relation $\mathcal{R}$.

The second important case to consider is when for some $j \in \left\{1 \ldots n \right\} \setminus f(\left\{ 1 \ldots m \right\})$
we have $P\sub{\vec{k}}{\vec{k_j}}$ or $Q\sub{\vec{k}}{\vec{k_j}}$ acts (or indeed they interact),
possibly extruding some active substitution $\sigma$ and fresh names $\vec{s}$, as follows.
\[
\!\!\!\!
\begin{array}[t]{l}
\begin{array}[t]{l}
 \mathopen{\nu c_k, \vec{k_1}, \vec{k_2}, \ldots \vec{k_n}, \vec{r}.}\Big(~\sigma \cpar
\\ \qquad
     R_1 \cpar \ldots P\sub{\vec{k}}{\vec{k_j}} \ldots \cpar R_n \cpar !c_k(\vec{k}).P \cpar
\\\qquad
     S_1 \cpar \ldots Q\sub{\vec{k}}{\vec{k_j}} \ldots \cpar S_n \cpar {!\nu \vec{k}.\cout{c_k}{\vec{k}}.Q}~\Big)
\end{array}
\\\quad
\lts{\pi}
\begin{array}[t]{l}
 \mathopen{\nu c_k, \vec{k_1}, \vec{k_2}, \ldots \vec{k_n}, \vec{r}, \vec{s}.}\Big(~\sigma \cpar \theta \cpar
\\ \qquad
     R'_1 \cpar \ldots R'_j \ldots \cpar R'_n \cpar !c_k(\vec{k}).P \cpar
\\\qquad
     S'_1 \cpar \ldots Q'_j \ldots \cpar S'_n \cpar {!\nu \vec{k}.\cout{c_k}{\vec{k}}.Q}~\Big)
\end{array}
\end{array}
\]
In this case, $B$ stays within relation $\mathcal{R}$ by using transition
\[
\begin{array}{l}
%\begin{array}[t]{l}
% \mathopen{\nu \vec{k_1}, \vec{k_2}, \ldots \vec{k_n}, \vec{r}.}\Big(~\sigma
%     R_{f(1)} \cpar S_{f(1)} \cpar \ldots \cpar R_{f(m)} \cpar S_{f(m)} \cpar
%    !\nu \vec{k}.\left(P \cpar Q \right)~\Big)
%\end{array}
%\\\qquad
B \lts{\pi}
\begin{array}[t]{l}
 \mathopen{\nu \vec{k}_{f(1)}, \vec{k}_{f(2)}, \ldots \vec{k}_{f(m)}, \vec{k}_{f(m+1)}, \vec{r}, \vec{s}.}\Big(~\sigma \cpar \theta \cpar
\\ \qquad\qquad
     R'_{g(1)} \cpar S'_{g(1)} \cpar \ldots \cpar R'_{g(m)} \cpar S'_{g(m)} \cpar R'_{g(m+1)} \cpar S'_{g(m+1)}
%\\\qquad
    \cpar !\nu \vec{k}.\left(P \cpar Q \right)~\Big)
\end{array}
\end{array}
\]
where $g \colon \left\{ 1 \ldots m+1 \right\} \rightarrow \left\{ 1 \ldots n \right\}$
such that $g(i) = \left\{\begin{array}{lr} j &  \mbox{if}~i = m+1 \\  f(i) & \mbox{otherwise} \end{array}\right.$, which is clearly injective.
Note we should also consider when two distinct $j, j' \in \left\{1 \ldots n \right\} \setminus f(\left\{ 1 \ldots m \right\})$
interact in $A$, which has a similar pattern, except we require pairs of processes to be added to $B$ using the rule $\textsc{Rep-close}$.
\end{proof}

\section{Classical \texorpdfstring{$\FM$}{FM} characterises strong early bisimilarity}\label{sec:app2}

In this paper, we prove that unlinkability properties are violated by exhibiting a distinguishing formula in classical $\FM$.
A distinguishing formula is sufficient evidence to show that two processes specifying the unlinkability property are not bisimilar, as long as the modal logic characterises bisimilarity.
Therefore the proof of soundness and completeness of strong early bisimilarity with respect to classical $\FM$
is critical for this work.
Indeed, other parts of our reasoning may be incomplete, e.g., using open bisimilarity to seek a distinguishing strategy,  but if our method discovers an attack that can be described using an $\FM$ formula that we confirm is distinguishing, then we are certain that unlinkability does not hold as formulated.

We reiterate Theorem~\ref{thm:char}. The proof is standard for a classical Milner-Parrow-Walker logic, for which reason it appears in this appendix.
In fact, the use of static equivalence simplifies the analysis compared to the $\pi$-calculus, since there are no special cases for bound actions.

\begin{thm}
[{Theorem~\ref{thm:char}}]
$P \sim Q$, whenever, for all $\phi$, we have $P \vDash \phi$ if and only if $Q \vDash \phi$.
\end{thm}
\begin{proof}
Let $\mathcal{R} = \left\{ (A, B) \colon \forall \phi, A \vDash \phi \,\mbox{iff}\, B \vDash \phi  \right\}$.
%\RH{Well formed in the sense that dom of active substitutions free and in normal form.}
We aim to prove $\mathcal{R}$ is a strong early bisimulation.
Symmetry is immediate.
In the following cases assume $A \mathrel{\mathcal{R}} B$.

\textit{Case of static equivalence.}
By definition of $\mathcal{R}$ for any $M$ and $N$,
we can apply $\alpha$-conversion to $A$ and $B$ such that
$A = \nu \vec{x}.(\theta \cpar P)$
and
$B = \nu \vec{y}.(\sigma \cpar Q)$
and $\left( \vec{x} \cup \vec{y} \right) \cap \left( \fv{M} \cup \fv{N} \right) = \emptyset$.
If $M\theta = N\theta$,
then by definition of satisfaction, $A \vDash M = N$
hence, by definition of $\mathcal{R}$, $B \vDash M = N$,
hence by definition of satisfaction, $M\sigma = N\sigma$.
Therefore $A$ and $B$ are statically equivalent.

\textit{Case of actions.}
Suppose $A \lts{\pi} A'$.
%and consider any $\phi$ such that $A' \vdash \phi$.
Hence $A \vDash \diam{\pi}\ttt$, so by definition of $\mathcal{R}$, we have
$B \vDash \diam{\pi}\ttt$ and hence for some $B'$ we have $B \lts{\pi} B'$.
By image-finiteness, there are finitely many $B_i$ such that $B \lts{\pi} B_i$.
Suppose for contradiction that $A' \mathrel{\mathcal{R}} B_i$ does not hold for all $i$.
Then for all $i$, there exists $\phi_i$ such that $A' \vDash \phi_i$
and $B_i \nvDash \phi_i$. Hence $A \vDash \diam{\pi}\left( \bigwedge_i \phi_i \right)$
but $B \nvDash \diam{\pi}\left( \bigwedge_i \phi_i \right)$, contradicting the assumption that $A \mathrel{\mathcal{R}} B$.
Hence for some $i$, $A' \mathrel{\mathcal{R}} B_i$, as required.

Thus, $\mathcal{R}$ is a strong early bisimulation.
Hence, if for any processes $P$ and $Q$ it holds that, for all formula $\phi$, we have $P \vDash \phi$ if and only if $Q \vDash \phi$,
then we have $P \mathrel{\mathcal{R}} Q$, and hence $P \sim Q$.

The converse direction follows by induction on the structure of $\phi$.
Assume $P \sim Q$,
hence there is some strong early bisimulation $\mathcal{S}$ such that $P \mathrel{\mathcal{S}} Q$.
In the following, assume that $A \mathrel{\mathcal{S}} B$ holds.
% and $A$ and $B$ are well formed extended processes.

\textit{Case of equality.}
Consider when $A \vDash M = N$.
By $\alpha$-conversion, $A = \nu \vec{x}. ( \theta \cpar P )$
such that $\vec{x} \cap \left( \fv{M} \cup \fv{N} \right) = \emptyset$.
hence $M\theta = N\theta$.
Now, by $\alpha$-conversion we have $B = \nu \vec{y}.( \sigma \cpar Q )$
such that $\vec{y} \cap \left( \fv{M} \cup \fv{N} \right) = \emptyset$.
So, by static equivalence, $M\sigma = N\sigma$;
and hence $A \vDash M = N$, as required.

\textit{Case of conjunction.}
Consider when $A \vDash \phi \wedge \psi$
hence $A \vDash \phi$ and $A \vDash \psi$.
So, by the induction hypothesis $B \vDash \phi$ and $B \vDash \psi$
and hence $B \vDash \phi \wedge \psi$.

\textit{Case of negation.}
Consider when $A \vDash \neg \phi$,
hence $A \nvDash \phi$.
Hence, by the induction hypothesis, $B \nvDash \phi$ hence $B \vDash \neg\phi$.

\textit{Case of action.}
Consider when $A \vDash \diam{\pi} \phi$.
Hence $A \lts{\pi} A'$ such that $A' \vDash \phi$.
Since $\mathcal{S}$ is a strong early bisimulation, there exists $B'$ such that $B \lts{\pi} B'$ and $A' \mathrel{\mathcal{S}} B'$.
Hence, by the induction hypothesis,
$B' \vDash \phi$. Hence $B \vDash \diam{\pi} \phi$.

Hence, by induction on the structure of $\phi$, for all formulae $\phi$, and for all $A$, $B$ such that $A \mathrel{\mathcal{S}} B$, we have $A \vDash \phi$ iff $B \vDash \phi$; and hence $P \vDash \phi$ iff $Q \vDash \phi$, since $P \mathrel{\mathcal{S}} Q$.
\end{proof}

%% file: main.bbl
\newcommand{\etalchar}[1]{$^{#1}$}
\begin{thebibliography}{DvGHM08}
\expandafter\ifx\csname url\endcsname\relax
  \def\url#1{\texttt{#1}}\fi
\expandafter\ifx\csname doi\endcsname\relax
  \def\doi#1{\burlalt{doi:#1}{http://dx.doi.org/#1}}\fi
\expandafter\ifx\csname urlprefix\endcsname\relax\def\urlprefix{URL }\fi
\expandafter\ifx\csname href\endcsname\relax
  \def\href#1#2{#2}\fi
\expandafter\ifx\csname burlalt\endcsname\relax
  \def\burlalt#1#2{\href{#2}{#1}}\fi

\bibitem[ABBD{\etalchar{+}}20]{Aceto2020}
Luca Aceto, Jos Baeten, Patricia Bouyer-Decitre, Holger Hermanns, and Alexandra
  Silva.
\newblock {CONCUR Test-Of-Time Award 2020 Announcement (Invited Paper)}.
\newblock In Igor Konnov and Laura Kov{\'a}cs, editors, {\em 31st International
  Conference on Concurrency Theory (CONCUR 2020)}, volume 171 of {\em Leibniz
  International Proceedings in Informatics (LIPIcs)}, pages 5:1--5:3, Dagstuhl,
  Germany, 2020. Schloss Dagstuhl--Leibniz-Zentrum f{\"u}r Informatik.
\newblock \doi{10.4230/LIPIcs.CONCUR.2020.5}.

\bibitem[ABF17]{Abadi2017}
Mart\'{\i}n Abadi, Bruno Blanchet, and C{\'e}dric Fournet.
\newblock The applied pi calculus: Mobile values, new names, and secure
  communication.
\newblock {\em Journal of the ACM}, 65(1):1:1--1:41, 2017.
\newblock \doi{10.1145/3127586}.

\bibitem[ABH{\etalchar{+}}16]{Avoine16}
Gildas Avoine, Antonin Beaujeant, Julio Hernandez{-}Castro, Louis Demay, and
  Philippe Teuwen.
\newblock A survey of security and privacy issues in epassport protocols.
\newblock {\em {ACM} Comput. Surv.}, 48(3):47:1--47:37, 2016.
\newblock \doi{10.1145/2825026}.

\bibitem[ABW06]{Andova06}
Suzana Andova, Jos C.~M. Baeten, and Tim A.~C. Willemse.
\newblock A complete axiomatisation of branching bisimulation for probabilistic
  systems with an application in protocol verification.
\newblock In Christel Baier and Holger Hermanns, editors, {\em {CONCUR} 2006 -
  Concurrency Theory, 17th International Conference, {CONCUR} 2006, Bonn,
  Germany, August 27-30, 2006, Proceedings}, volume 4137 of {\em Lecture Notes
  in Computer Science}, pages 327--342. Springer, 2006.
\newblock \doi{10.1007/11817949\_22}.

\bibitem[ACRR10]{Arapinis2010}
Myrto Arapinis, Tom Chothia, Eike Ritter, and Mark Ryan.
\newblock Analysing unlinkability and anonymity using the applied pi calculus.
\newblock In {\em 23rd IEEE Computer Security Foundations Symposium}, pages
  107--121, 2010.
\newblock \doi{10.1109/CSF.2010.15}.

\bibitem[AG99]{abadi99ic}
Martin Abadi and Andrew~D. Gordon.
\newblock A calculus for cryptographic protocols: {The} spi calculus.
\newblock {\em Information and Computation}, 148(1):1--70, 1999.
\newblock \doi{10.1006/inco.1998.2740}.

\bibitem[AHT17]{Ahn2017}
Ki~Yung Ahn, Ross Horne, and Alwen Tiu.
\newblock A characterisation of open bisimilarity using an intuitionistic modal
  logic.
\newblock In Roland Meyer and Uwe Nestmann, editors, {\em 28th International
  Conference on Concurrency Theory, {CONCUR} 2017, September 5-8, 2017, Berlin,
  Germany}, volume~85 of {\em LIPIcs}, pages 7:1--7:17, 2017.
\newblock \doi{10.4230/LIPIcs.CONCUR.2017.7}.

\bibitem[And21]{Anderson2021}
Ross Anderson.
\newblock {\em Security Engineering (third edition)}.
\newblock John Wiley \& Sons, Inc., 2021.

\bibitem[Bae21]{Baelde2021}
David Baelde.
\newblock {\em Contributions {\`a} la V{\'e}rification des Protocoles
  Cryptographiques}.
\newblock habilitation thesis, Paris-Saclay, 2021.

\bibitem[BB91]{Baeten91}
Jos C.~M. Baeten and Jan~A. Bergstra.
\newblock Real time process algebra.
\newblock {\em Formal Aspects of Computing}, 3(2):142--188, 1991.
\newblock \doi{10.1007/BF01898401}.

\bibitem[BB93]{Baeten93}
Jos C.~M. Baeten and Jan~A. Bergstra.
\newblock Non interleaving process algebra.
\newblock In Eike Best, editor, {\em CONCUR'93}, pages 308--323. Springer,
  1993.
\newblock \doi{10.1007/3-540-57208-2_22}.

\bibitem[BB98]{Baeten98}
Jos C.~M. Baeten and Jan~A. Bergstra.
\newblock Deadlock behaviour in split and {ST} bisimulation semantics.
\newblock In Ilaria Castellani and Catuscia Palamidessi, editors, {\em Fifth
  International Workshop on Expressiveness in Concurrency, {EXPRESS} 1998,
  Satellite Workshop of {CONCUR} 1998, Nice, France, September 7, 1998},
  volume~16 of {\em Electronic Notes in Theoretical Computer Science}, pages
  61--74. Elsevier, 1998.
\newblock \doi{10.1016/S1571-0661(04)00117-3}.

\bibitem[BBK87]{Baeten87}
Jos C.~M. Baeten, Jan~A. Bergstra, and Jan~Willem Klop.
\newblock Ready-trace semantics for concrete process algebra with the priority
  operator.
\newblock {\em Comput. J.}, 30(6):498--506, 1987.
\newblock \doi{10.1093/comjnl/30.6.498}.

\bibitem[BBMV91]{BaetenMauw}
Jos C.~M. Baeten, Jan~A. Bergstra, Sjouke Mauw, and Gert~J. Veltink.
\newblock A process specification formalism based on static {COLD}.
\newblock In Jan~A. Bergstra and Loe M.~G. Feijs, editors, {\em Algebraic
  Methods {II}: Theory, Tools and Applications}, pages 303--335. Springer,
  1991.
\newblock \doi{10.5555/109462.109475}.

\bibitem[BFK09]{Bender2009}
Jens Bender, Marc Fischlin, and Dennis K{\"u}gler.
\newblock Security analysis of the {PACE} key-agreement protocol.
\newblock In Pierangela Samarati, Moti Yung, Fabio Martinelli, and Claudio~A.
  Ardagna, editors, {\em Information Security}, pages 33--48. Springer, 2009.
\newblock \doi{10.1007/978-3-642-04474-8_3}.

\bibitem[BLMvT16]{Baeten2016}
Jos C.~M. Baeten, Bas Luttik, Tim Muller, and Paul van Tilburg.
\newblock Expressiveness modulo bisimilarity of regular expressions with
  parallel composition.
\newblock {\em Mathematical Structures in Computer Science}, 26(6):933–968,
  2016.
\newblock \doi{10.1017/S0960129514000309}.

\bibitem[BN07]{briais06entcs}
S{\'e}bastien Briais and Uwe Nestmann.
\newblock Open bisimulation, revisited.
\newblock {\em Theoretical Computer Science}, 386(3):236--271, 2007.
\newblock \doi{j.tcs.2007.07.010}.

\bibitem[CDS20]{Delaune2020}
V{\'e}ronique Cortier, St{\'{e}}phanie Delaune, and Vaishnavi Sundararajan.
\newblock A decidable class of security protocols for both reachability and
  equivalence properties.
\newblock Technical Report hal-02446170, Loria \& Inria Grand Est; Irisa, 2020.
\newblock \urlprefix\url{https://hal.inria.fr/hal-02446170/}.

\bibitem[CGIP12]{Coron12}
Jean{-}S{\'{e}}bastien Coron, Aline Gouget, Thomas Icart, and Pascal Paillier.
\newblock Supplemental access control {(PACE} v2): Security analysis of {PACE}
  integrated mapping.
\newblock In David Naccache, editor, {\em Cryptography and Security: From
  Theory to Applications - Essays Dedicated to Jean-Jacques Quisquater on the
  Occasion of His 65th Birthday}, volume 6805 of {\em Lecture Notes in Computer
  Science}, pages 207--232. Springer, 2012.
\newblock \doi{10.1007/978-3-642-28368-0\_15}.

\bibitem[CKR18]{deepsec}
Vincent {Cheval}, Steve {Kremer}, and Itsaka {Rakotonirina}.
\newblock {DEEPSEC}: Deciding equivalence properties in security protocols
  theory and practice.
\newblock In {\em 2018 IEEE Symposium on Security and Privacy (S\&P)}, pages
  529--546, 2018.
\newblock \doi{10.1109/SP.2018.00033}.

\bibitem[CMdV06]{Cremer2006}
Cas Cremers, Sjouke Mauw, and Erik~P. de~Vink.
\newblock Injective synchronisation: An extension of the authentication
  hierarchy.
\newblock {\em Theoretical Computer Science}, 367(1):139--161, 2006.
\newblock \doi{10.1016/j.tcs.2006.08.034}.

\bibitem[Cre08]{Cremers2008}
Cas Cremers.
\newblock The {Scyther} tool: Verification, falsification, and analysis of
  security protocols.
\newblock In {\em International Conference on Computer Aided Verification},
  pages 414--418. Springer, 2008.
\newblock \doi{10.1007/978-3-540-70545-1_38}.

\bibitem[CS10]{originalBACAttack}
Tom Chothia and Vitaliy Smirnov.
\newblock A traceability attack against e-passports.
\newblock In Radu Sion, editor, {\em Financial Cryptography and Data Security,
  14th International Conference, {FC} 2010, Tenerife, Canary Islands, Spain,
  January 25-28, 2010, Revised Selected Papers}, volume 6052 of {\em Lecture
  Notes in Computer Science}, pages 20--34. Springer, 2010.
\newblock \doi{10.1007/978-3-642-14577-3\_5}.

\bibitem[Del19]{Delano}
Uni researchers discover e-passport flaw.
\newblock Delano Magazine, Luxembourg, September 2019.
\newblock
  \urlprefix\url{https://delano.lu/d/detail/news/uni-researchers-discover-e-passport-flaw/207929}.

\bibitem[DKR09]{Delaune2009}
St{\'{e}}phanie Delaune, Steve Kremer, and Mark Ryan.
\newblock Verifying privacy-type properties of electronic voting protocols.
\newblock {\em Journal of Computer Security}, 17(4):435--487, 2009.
\newblock \doi{10.3233/JCS-2009-0340}.

\bibitem[DvGHM08]{Deng2008}
Yuxin Deng, Rob van Glabbeek, Matthew Hennessy, and Carroll Morgan.
\newblock Characterising testing preorders for finite probabilistic processes.
\newblock {\em Logical Methods in Computer Science}, 4(4), 2008.
\newblock \doi{10.2168/LMCS-4(4:4)2008}.

\bibitem[FHMS19]{ESORICS}
Ihor Filimonov, Ross Horne, Sjouke Mauw, and Zach Smith.
\newblock Breaking unlinkability of the {ICAO} 9303 standard for e-passports
  using bisimilarity.
\newblock In Kazue Sako, Steve Schneider, and Peter Y.~A. Ryan, editors, {\em
  Computer Security -- ESORICS 2019}, pages 577--594. Springer, 2019.
\newblock \doi{10.1007/978-3-030-29959-0_28}.

\bibitem[Fil20]{Filimonov2020}
Ihor Filimonov.
\newblock {\em Analysis of privacy attacks on ePassports and a mitigation
  strategy using TOTP}.
\newblock master thesis, University of Luxembourg, 2020.

\bibitem[HALT18]{Horne2018}
Ross Horne, Ki~Yung Ahn, Shang-Wei Lin, and Alwen Tiu.
\newblock Quasi-open bisimilarity with mismatch is intuitionistic.
\newblock In Anuj Dawar and Erich Gr{\"a}del, editors, {\em In Proceedings of
  33rd Annual ACM/IEEE Symposium on Logic in Computer Science, Oxford, United
  Kingdom, July 9-12, 2018}, pages 26--35, 2018.
\newblock \doi{10.1145/3209108.3209125}.

\bibitem[HBD16]{BACUnboundedUnlinkability}
Lucca Hirschi, David Baelde, and St{\'e}phanie Delaune.
\newblock A method for verifying privacy-type properties: the unbounded case.
\newblock In {\em Security and Privacy (S\&P), 2016 IEEE Symposium on}, pages
  564--581. IEEE, 2016.
\newblock \doi{10.1109/SP.2016.40}.

\bibitem[HBD19]{Hirschi2019}
Lucca Hirschi, David Baelde, and St{\'{e}}phanie Delaune.
\newblock A method for unbounded verification of privacy-type properties.
\newblock {\em Journal of Computer Security}, 27(3):277--342, 2019.
\newblock \doi{10.3233/JCS-171070}.

\bibitem[HL95]{Hennessy1995}
Matthew Hennessy and Huimin Lin.
\newblock Symbolic bisimulations.
\newblock {\em Theoretical Computer Science}, 138(2):353--389, 1995.
\newblock \doi{10.1016/0304-3975(94)00172-F}.

\bibitem[Hor18]{ArXiv}
Ross Horne.
\newblock A bisimilarity congruence for the applied $\pi$-calculus sufficiently
  coarse to verify privacy properties.
\newblock {\em CoRR}, (arXiv:1811.02536), 2018.
\newblock \urlprefix\url{https://arxiv.org/abs/1811.02536}.

\bibitem[H{\"u}t03]{Huttel2003}
Hans H{\"u}ttel.
\newblock Deciding framed bisimilarity.
\newblock {\em Electronic Notes in Theoretical Computer Science}, 68(6):1--18,
  2003.
\newblock \doi{10.1016/S1571-0661(04)80530-9}.

\bibitem[ISO18]{ISO14443}
Cards and security devices for personal identification — contactless
  proximity objects — part 3: Initialization and anticollision.
\newblock ISO/IEC 14443-3, 2018.
\newblock \urlprefix\url{https://www.iso.org/standard/73598.html}.

\bibitem[Kri65]{Kripke1965}
Saul~A. Kripke.
\newblock Semantical analysis of intuitionistic logic {I}.
\newblock In J.N. Crossley and M.A.E. Dummett, editors, {\em Formal Systems and
  Recursive Functions}, volume~40 of {\em Studies in Logic and the Foundations
  of Mathematics}, pages 92--130. Elsevier, 1965.
\newblock \doi{https://doi.org/10.1016/S0049-237X(08)71685-9}.

\bibitem[Lab19a]{PaperJam1}
Thierry Labro.
\newblock Une faille dans les passeports {\'e}lectroniques.
\newblock {\em Paperjam, Luxembourg}, September 2019.
\newblock
  \urlprefix\url{https://paperjam.lu/article/faille-dans-passeports-electro}.

\bibitem[Lab19b]{PaperJam2}
Thierry Labro.
\newblock Une faille qui devrait alerter les autorit{\'e}s.
\newblock {\em Paperjam, Luxembourg}, September 2019.
\newblock
  \urlprefix\url{https://paperjam.lu/article/faille-qui-devrait-alerter-aut}.

\bibitem[LL12]{Liu}
Jia Liu and Huimin Lin.
\newblock A complete symbolic bisimulation for full applied pi calculus.
\newblock {\em Theoretical Computer Science}, 458:76--112, 2012.
\newblock \doi{https://doi.org/10.1016/j.tcs.2012.07.034}.

\bibitem[{Low}97]{Lowe1997}
Gavin {Lowe}.
\newblock A hierarchy of authentication specifications.
\newblock In {\em Proceedings 10th Computer Security Foundations Workshop},
  pages 31--43, June 1997.
\newblock \doi{10.1109/CSFW.1997.596782}.

\bibitem[MDBdV12]{Markovski2012}
Jasen Markovski, Pedro~R. D'Argenio, Jos C.~M. Baeten, and Eric~P. de~Vink.
\newblock Reconciling real and stochastic time: the need for probabilistic
  refinement.
\newblock {\em Formal Aspects of Computing}, 24(4):497--518, 2012.
\newblock \doi{10.1007/s00165-012-0230-y}.

\bibitem[MPW93]{Milner1993}
Robin Milner, Joachim Parrow, and David Walker.
\newblock Modal logics for mobile processes.
\newblock {\em Theor. Comput. Sci.}, 114(1):149--171, 1993.
\newblock \doi{10.1016/0304-3975(93)90156-N}.

\bibitem[MRT15]{MRTD}
Machine readable travel documents. part 11: Security mechanisms for {MRTDs}.
\newblock International Civil Aviation Organization (ICAO), Doc 9303. Seventh
  Edition, 2015.
\newblock
  \urlprefix\url{https://www.icao.int/publications/Documents/9303_p11_cons_en.pdf}.

\bibitem[MvOV96]{Menezes1996}
Alfred Menezes, Paul~C. van Oorschot, and Scott~A. Vanstone.
\newblock {\em Handbook of Applied Cryptography}.
\newblock {CRC} Press, 1996.
\newblock \doi{10.1201/9781439821916}.

\bibitem[TD10]{Tiu2010CSF}
Alwen Tiu and Jeremy Dawson.
\newblock Automating open bisimulation checking for the spi calculus.
\newblock In {\em 2010 23rd IEEE Computer Security Foundations Symposium},
  pages 307--321. IEEE, 2010.
\newblock \doi{10.1109/CSF.2010.28}.

\bibitem[TGD10]{Tiu2010LMCS}
Alwen Tiu, Rajeev Gore, and Jeremy Dawson.
\newblock {A Proof Theoretic Analysis of Intruder Theories}.
\newblock {\em {Logical Methods in Computer Science}}, {Volume 6, Issue 3},
  2010.
\newblock \doi{10.2168/LMCS-6(3:12)2010}.

\bibitem[Tiu07]{tiu07aplas}
Alwen Tiu.
\newblock A trace based bisimulation for the spi calculus: An extended
  abstract.
\newblock In {\em Programming Languages and Systems. APLAS 2007}, volume 4807
  of {\em Lecture Notes in Computer Science}, pages 367--382. Springer, 2007.
\newblock \doi{10.1007/978-3-540-76637-7_25}.

\bibitem[TNH16]{Tiu2016}
Alwen Tiu, Nam Nguyen, and Ross Horne.
\newblock {SPEC}: An equivalence checker for security protocols.
\newblock In Atsushi Igarashi, editor, {\em Programming Languages and Systems.
  APLAS 2016}, pages 87--95. Springer, 2016.
\newblock \doi{10.1007/978-3-319-47958-3_5}.

\bibitem[vDMR08]{Deursen08}
Ton van Deursen, Sjouke Mauw, and Sasa Radomirovic.
\newblock Untraceability of {RFID} protocols.
\newblock In Jose~Antonio Onieva, Damien Sauveron, Serge Chaumette, Dieter
  Gollmann, and Constantinos Markantonakis, editors, {\em Information Security
  Theory and Practices. Smart Devices, Convergence and Next Generation
  Networks, Second {IFIP} {WG} 11.2 International Workshop, {WISTP} 2008,
  Seville, Spain, May 13-16, 2008. Proceedings}, volume 5019 of {\em Lecture
  Notes in Computer Science}, pages 1--15. Springer, 2008.
\newblock \doi{10.1007/978-3-540-79966-5\_1}.

\bibitem[vG01]{Glabbeek2001}
Rob van Glabbeek.
\newblock The linear time -- branching time spectrum {I}.
\newblock In Jan~A. Bergstra, Alban Ponse, and Scott~A. Smolka, editors, {\em
  Handbook of Process Algebra}, pages 3 -- 99. Elsevier Science, Amsterdam,
  2001.
\newblock \doi{https://doi.org/10.1016/B978-044482830-9/50019-9}.

\bibitem[vG21]{Glabbeek2020}
Rob van Glabbeek.
\newblock Failure trace semantics for a process algebra with time-outs.
\newblock {\em Logical Methods in Computer Science}, 17(2):11:1--11:40, 2021.
\newblock \doi{10.23638/LMCS-17(2:11)2021}.

\end{thebibliography}
